%% file: main_arxiv.tex
\RequirePackage{fix-cm}
\documentclass{svjour3}                     
\smartqed  
\usepackage{amsmath,amssymb}
\usepackage[dvipsnames]{xcolor}
\usepackage{enumerate}
\usepackage{epstopdf}
\usepackage{relsize}
\usepackage{url}
\usepackage{multirow}
\usepackage[left=1in,top=1in,right=1in,bottom=1in,letterpaper]{geometry}
\usepackage{mathrsfs}
\usepackage{algorithm}
\usepackage{subfig}
\usepackage{epsfig}
\usepackage{graphicx}
\usepackage{epstopdf}
\usepackage{tabularx}
\usepackage{bm}
\usepackage{multirow}
\usepackage{algpseudocode}
\usepackage{framed}
\usepackage{tikz-cd}
\usepackage{booktabs}
\usepackage{tensor}
\usepackage[T1]{fontenc}
\usepackage[font=small,labelfont=bf]{caption}

\newtheorem{example-numbered}{Example}[section]

\newcommand{\SO}{\mathrm{SO}}
\newcommand{\Unitary}{\mathrm{U}\!}

\usepackage{hyperref}
\hypersetup{
  colorlinks   = true,    
  urlcolor     = blue,    
  linkcolor    = blue,    
  citecolor    = red,     
}

\DeclareMathOperator*{\argmin}{arg\,min}
\DeclareMathOperator*{\argmax}{arg\,max}

\allowdisplaybreaks

\newcommand{\RID}{\mathrm{RID}}

\usepackage{lineno}

\def\makeheadbox{\relax}
\begin{document}

\title{{Representation Theoretic Patterns in Multi-Frequency Class Averaging for Three-Dimensional Cryo-Electron Microscopy} 
\thanks{ZZ and YF acknowledge the support from Strategic Research Initiatives in the University of Illinois at Urbana-Champaign and NSF grant DMS-1854791. TG acknowledges the support from NSF DMS-1854831, an AMS-Simons Travel Grant, and partial support from DARPA D15AP00109 and NSF grants IIS-1546413.}
}

\titlerunning{Representation Theoretic Patterns in Multi-Frequency
  Class Averaging for 3D Cryo-EM}        

\author{Yifeng Fan \and Tingran Gao \and Zhizhen Zhao}

\authorrunning{Yifeng Fan, Tingran Gao, Zhizhen Zhao} 

\institute{
           Yifeng Fan \at
           Department of Electrical and Computer
           Engineering\\
           University of Illinois at Urbana--Champaign\\
           \email{yifengf2@illinois.edu}
           \and
           Tingran Gao \at
           Committee on Computational and Applied
           Mathematics\\
           Department of Statistics\\
           University of Chicago\\
           \email{tingrangao@galton.uchicago.edu}
           \and
           Zhizhen Zhao \at
           Department of Electrical and Computer
           Engineering\\
           Coordinated Science Laboratory\\
           University of Illinois at Urbana--Champaign\\
           \email{zhizhenz@illinois.edu}
}

\date{Received: date / Accepted: date}

\maketitle

\begin{abstract}
We develop in this paper a novel intrinsic classification algorithm --- \emph{multi-frequency class averaging} (MFCA) --- for classifying noisy projection images obtained from three-dimensional
cryo-electron microscopy (cryo-EM) by the similarity among their viewing directions. This new algorithm leverages multiple irreducible representations of the unitary group to introduce additional redundancy into the representation of the optimal in-plane rotational alignment, extending and outperforming the existing class averaging algorithm 
that uses only a single representation. The formal algebraic model and
representation theoretic patterns of the proposed MFCA algorithm extend the framework of Hadani and Singer to arbitrary irreducible representations of the unitary group. We conceptually establish the consistency and
stability of MFCA by inspecting the spectral properties of a generalized local parallel transport operator through the lens of Wigner $D$-matrices. We demonstrate the efficacy of the proposed algorithm with numerical experiments.

\keywords{Representation theory \and Spectral theory \and Differential geometry
  \and Wigner matrices \and Cryo-electron microscopy \and Mathematical biology}
\subclass{20G05 \and 33C45 \and 33C55 \and 55R25}
\end{abstract}

\input{intro_v2}

\input{algorithm}

\input{prelim_v2}

\input{model}

\input{theory_v2}

\input{noise_model}

\input{numerical}

\input{future}

\begin{acknowledgements}
The authors thank Vera Mikyoung Hur, Jared Bronski, Shmuel Weinberger, and Shamgar Gurevich for useful discussions.
\end{acknowledgements}

\appendix
\renewcommand*{\thesection}{Appendix~\Alph{section}}
\normalsize

\input{app}

%
%

\bibliographystyle{spmpsci}      
\bibliography{ref}   

\end{document}

%% file: intro_v2.tex
\section{Introduction}
\label{sec:introduction}

The past decades have witnessed an emerging and continued impact of cryo-electron microscopy (cryo-EM), the Nobel Prize winning imaging technology for determining three-dimensional structures of macromolecules, on a wide range of natural scientific fields \cite{DCPK+1998,Mackinnon2004,frank2006three,HM2013,CF2015,OJ2017}. Compared with its predecessor, X-ray crystallography, of which the success builds upon the potentially difficult procedure of crystallization, cryo-EM is able to image the macromolecules in their native states and produces large numbers of projection images for samples of molecules rapidly frozen in a thin layer of vitreous ice. The projection images can be thought of as tomographic projections of many copies of an identical molecule at unknown and random orientations. A major computational challenge in reconstructing the three-dimensional molecular structure from these projection images is the extremely low signal-to-noise ratio (SNR) caused by the limited allowable electron dose (so as to avoid damaging the molecule before the imaging completes). It is thus customary to improve the SNR by performing \emph{class averaging} --- the procedure of aligning and then averaging out projection images taken along nearby viewing directions --- from rotationally invariant pairwise comparisons of the projection images \cite{PZF1996,frank2006three}, before the downstream reconstruction workflow such as angular reconstitution \cite{VG1986,VanHell1987,SLFV+2017}. In addition to its scientific value, the rich geometric structure in the cryo-EM imaging model has also inspired many mathematical and algorithmic investigations \cite{Singer2011a,Gurevich2011,SS2012,SingerWu2012VDM,BSS2013,WangSinger2013,zhao2014rotationally,BSAB2014,Gao2015,bandeira2015non,HDM2016,YeLim2017,GBM2019}. 
\subsection{Background: The Mathematical Model of
  Cryo-Electron Microscopy and Class Averaging}
\label{sec:background}
Following \cite{singer2011viewing,hadani2011representation2}, we view the collection of projection images $\left\{ I_i\in\mathbb{R}^{L\times L}\mid i=1,\dots,N \right\}$ as tomographic projection images for the same three-dimensional object along projection directions uniformly sampled from the two-sphere $S^2$, as it is more convenient to consider the imaging model in the molecule's own lab frame, where the molecule is fixed and observed by an electron microscope at various orientations. For simplicity, we assume the projection images are all centered, i.e. the center of mass of the clean projection images are at the center of the images. The goal is to identify and classify projection images produced from similar projection directions, hereafter referred to as \emph{viewing directions}.

A point $x\in \SO ( 3)$ is identified with an orthonormal basis $(\mathbf{e_1}, \mathbf{e_2}, \mathbf{e_3})$ of $\mathbb{R}^3$, with orientation compatible with the canonical orthonormal coordinate frame of $\mathbb{R}^3$. We identify $\mathbf{e_3}\in S^2$ with the viewing direction and denote it for $\pi \left( x \right)$ for the ease of notations. The 2D image obtained by the microscope observed at a spatial orientation $x$ is a real valued function $I:  \mathbb{R}^2 \rightarrow \mathbb{R}$, given by the X-ray transform along the viewing direction:
\begin{equation}
\label{eq:Xray}
I(s, t) = \int_{\mathbb{R}} \phi(s \mathbf{e_1} + t \mathbf{e_2} + r \mathbf{e_3})\,\mathrm{d}r\quad\textrm{for all $\left( s,t \right)\in\mathbb{R}^2$}
\end{equation}
where $\phi:\mathbb{R}^3\rightarrow\mathbb{R}$ is a real-valued function modeling the electromagnetic potential induced from the charges of the molecule. We assume the images $I \left( s,t \right)$ are all supported on a bounded set of $\mathbb{R}^2$ which fits into the size of the projection images.

To measure the similarity between any two projection images $I_i$ and $I_j$, obtained by the tomographic projection along viewing directions $\pi\left(x_i\right)\in S^2$ and $\pi\left(x_j\right)\in S^2$ respectively, we compute a rotationally invariant distance between $I_i$ and $I_j$ defined as
\begin{equation}
    d_{\RID} (I_i, I_j) = \min_{\theta \in [0, 2\pi)}\left\|I_i- R_\theta ( I_j ) \right\|_{\mathrm{F}},
    \label{eq:RID}
\end{equation}
where $R_\theta( I_j ) $ stands for the operation of rotating image $I_j$ by an angle $\theta\in [0,2\pi)$ in the counterclockwise orientation, and $\left\| \cdot \right\|_{\mathrm{F}}$ is the matrix Frobenius norm. The optimal alignment angle between $I_i$ and $I_j$ will be denoted as
\begin{equation}
    \theta_{ij} = \argmin_{\theta \in [0, 2\pi)}\left\|I_i- R_\theta ( I_j ) \right\|_{\mathrm{F}}. 
    \label{eq:opt_align}
\end{equation}

For images $I_x$ and $I_y$ obtained from viewing directions $\pi \left( x \right)$ and $\pi \left( y \right)$ for $x,y\in \SO ( 3)$ and without noise contamination, \cite{hadani2011representation2} models the optimal alignment angle as the \emph{transport data} encoding the angle of in-plane rotation needed to align frames $x,y$ after one of them is parallel-transported to the fibre of the other using the canonical Levi-Civita connection on the unit sphere equipped with an induced Riemannian structure from the ambient space $\mathbb{R}^3$. A rough idea for filtering out far-apart viewing directions is through thresholding the rotationally invariant distances between pairs of projection images against a preset threshold parameter $\epsilon>0$ that should be tuned to reflect the confidence in the accuracy of the imaging process. The pairwise comparison information after thresholding can be conveniently encoded into an \emph{observation graph} $G=\left( V,E \right)$, where each vertex of $G$ stands for one of the projection images, and an edge $\left( i,j \right)$ belongs to the edge set $E$ if and only if the rotationally invariant distance $d_{\RID}\left( I_i,I_j \right)$ is smaller than the threshold. In an ideal noiseless world, the geometry of the graph $G$ is a neighborhood graph on the unit sphere 
$S^2$, namely, two images are connected if and only if their viewing directions $\pi(x_i)$ and $\pi(x_j)$ are close on the unit sphere, $\langle \pi(x_i), \pi(x_j) \rangle \geq 1 - h$, for $h\ll 1$. From the noisy cryo-EM images, the rotationally invariant distances $d_\RID$ are affected by noise and $d_\RID$-based similarity measure will connect images of very different views, introducing short-cut edges on the unit sphere.  The main problem here is thus to distinguish the ``good'' edges from the ``bad'' ones in the graph $G$, or, in other words, to distinguish the true neighbors from the outliers. The existence of outliers makes the classification problem non-trivial. Without excluding the outliers, averaging rotationally aligned images with small invariant distance~\eqref{eq:RID} yields a poor estimate of the true signal, rendering infeasible the 3D \emph{ab initio} reconstruction from denoised images. We refer interested readers to \cite{EKW2016,LMQW2018}  for more detailed statistical analysis of the rotationally invariant distance \eqref{eq:RID}. The focus of this paper is to rectify the noise-contaminated empirical transport data using the spectral information of an integral operator constructed from the initial local transport data. 

\subsubsection{The Class Averaging Algorithm}
\label{sec:class-aver}

One of the most natural ideas for performing class averaging is through the eigenvectors of the class averaging matrix constructed from the empirical transport data $\{ e^{\iota \theta_{ij}} \}_{(i, j) \in E}$ \cite{singer2011viewing,hadani2011representation2}. We briefly recapture the main steps in the class averaging algorithm below. Detailed discussions and the analysis of representation theoretical patterns can be found in \cite{singer2011viewing,hadani2011representation2}. In this section we use notation $[N]=\{1,2,\dots,N\}$ for $N\in\mathbb{N}$.

The algorithm begins with computing rotationally invariant distances $d_{ij}$ between all pairs of projection images $I_i$ and $I_j$, along with the corresponding optimal alignment angles $\theta_{ij}$. After that, construct an $N$-by-$N$ Hermitian matrix $H$ by
\begin{align}
\label{eq:H}
  H_{ij}=
  \begin{cases}
    e^{\iota\theta_{ij}}&\textrm{if $\left( i,j \right)\in E$}, \\
    0 & \textrm{otherwise},
  \end{cases}
\end{align}
where the edge set $E\subset [N]\times [N]$ is obtained by thresholding the pairwise distances $\left\{ d_{ij}:1\leq i,j\leq N \right\}$, i.e., $\left( i,j \right)\in E$ if and only if $d_{ij}$ is below a preset threshold $\epsilon>0$, i.e.,
\begin{align}
  \label{eq:edge-threshold}
  E:= \left\{ \left( i,j \right)\in [N]\times [N]: d_{\RID}\left( I_i, I_j \right)<\epsilon \right\}.
\end{align}
Set $D$ as the diagonal matrix with diagonal entries
\begin{align}
\label{eq:degree-matrix}
  D_{ii}=\sum_{j=1}^N \left| H_{ij} \right|,\quad 1\neq i\leq N
\end{align}
and compute the top three eigenvectors $\psi_1,\psi_2,\psi_3\in\mathbb{C}^N$ of the normalized Hermitian matrix
\begin{align*}
  \widetilde{H}:=D^{-1/2}HD^{-1/2}.
\end{align*}
Each projection image is then associated with a point in $\mathbb{C}^3$ by means of the embedding map
\begin{align*}
    \Psi:\left\{ I_i \right\}_{i=1}^N &\longrightarrow \mathbb{C}^3\\
    I_i&\longmapsto \left( \psi_1 \left( i \right), \psi_2 \left( i \right), \psi_3 \left( i \right) \right)
\end{align*}
where $\psi_1 \left( i \right), \psi_2 \left( i \right), \psi_3 \left( i \right)$ denotes for the $i$th entries of $\psi_1,\psi_2,\psi_3$, respectively. The measure of affinity between $I_i$ and $I_j$ is then computed using the embedding map $\Psi$:
\begin{align}
\label{eq:single-freq-affinity}
  A_{ij}:=\frac{\left| \left\langle \Psi \left( I_i \right), \Psi \left( I_j \right) \right\rangle \right|}{\left\| \Psi \left( I_i \right) \right\|\left\| \Psi \left( I_j \right) \right\|},\quad 1\leq i\neq j\leq N.
\end{align}
Finally, the neighbors of a projection image $I_i$ are determined by thresholding the affinity measures $A_{ij}$:
\begin{align*}
  \textrm{Neighbors of $I_i$ }:=\left\{ I_j\mid A_{ij}>1-\gamma \right\}
\end{align*}
where  $0<\gamma<1$ is another preset threshold parameter that controls the size of the neighborhoods.

\subsection{Main Contributions}
\label{sec:contribution}
The main contributions of this paper are (1) the introduction of the multi-frequency class averaging (MFCA) algorithm to improve the viewing direction classification of cryo-EM single particle images, and (2) a complete characterization of the spectral information of a \emph{generalized local parallel transport operator} underlying the geometric relation in MFCA.  

Specifically, motivated by recent works \cite{bandeira2015non,gao2019multi,fan2019cryo,fan19a}, which incorporate multiple representations of the pairwise comparison information into the synchronization problem,   
we propose in this paper a \emph{multi-frequency class averaging} algorithm using the extended empirical transport data $\{ e^{\iota k \theta_{ij}} \}_{(i, j) \in E}$ for $k = 1, 2, \dots, k_\mathrm{max}$. It creates more than one copy of the class averaging matrix--- one for each ``frequency channel'' corresponding to one irreducible representation of $\SO ( 2 )$ group element.
Those matrices can be viewed as the discretization of the generalized local parallel transport operators $T^{(k)}_h$. A formal definition of $T_h^{\left( k \right)}$ can be found in~\eqref{eq:defn-gen-parallel-transp-op-loc}. The new algorithm uses the top $2k+1$ eigenvectors of the class averaging matrix at frequency $k$ to embed the images into $2k+1$-dimensional complex space. The new frequency $k$-affinity measure is defined as the absolute normalized cross correlation of the embedded vectors. We also propose to aggregate the affinity measures across the frequency channels to enforce the consistency of the nearest neighbor identification. Since the performance of the algorithm depends on the properties and stability of the top eigenvectors, we perform the spectral analysis of the corresponding integral operator $T_h^{(k)}$. We show in Theorem~\ref{thm:eval} and Theorem~\ref{thm:specgap} that the top eigenspace of $T_h^{\left( k \right)}$, denoted as $\mathbb{W}^{\left( k \right)}$, is $\left( 2k+1 \right)$-dimensional. In addition, we show that the top eigenvalue of $T_{h}^{(k)}$ decreases as $k$ increases and the top spectral gap increases as $k$ increases up to a threshold determined by the local neighborhood size. The increasing spectral gap implies the advantage of using higher frequency information for class averaging, as the \emph{numerical stability} of the eigen-decomposition step in MFCA depends on the magnitude of the spectral gap. 

In addition to the characterization of the dimensionality of the top eigenspace $\mathbb{W}^{\left( k \right)}$ of $T_h^{\left( k \right)}$, we also demonstrate in Theorem~\ref{thm:morph} and Theorem~\ref{thm:inner-product} the existence of a canonical identification of $\mathbb{W}^{\left( k \right)}$ with a complex $\left( 2k+1 \right)$-dimensional linear space spanned by $\left( 2k+1 \right)$ linearly independent entry functions in the \emph{Wigner $D$-matrix} associated with the unique $\left( 2k+1 \right)$-dimensional unitary irreducible representation of $\SO ( 3)$. A direct corollary of this canonical identification is the equality between the frequency-$k$ affinity measure and the viewing angle, thus generalizing the result in \cite{hadani2011representation2} for the affinity measure \eqref{eq:single-freq-affinity}. These facts establish the \textrm{admissibility} (consistency) of the proposed MFCA algorithm.

We emphasize that these theoretical results are not straightforward extensions of the techniques in \cite{hadani2011representation2} to the generalized localized parallel transport operator $T_h^{\left( k \right)}$. The generating-function-based approach in \cite{hadani2011representation2} is not easy to generalize to our setting without heavy notation and lengthy mathematical inductions. Instead, we observed that the constructions in \cite{hadani2011representation2} can be greatly simplified using an alternative construction by means of the Wigner $D$-matrices, which has been widely used in studies in mathematical physics concerning the irreducible representation of $\SO ( 3)$.

In the clean, noiseless scenario, the multi-frequency class averaging matrices certainly carry identical information for exactly recovering the affinity among view directions of the projection images; the real advantage, as argued and demonstrated in the theoretical analysis of \cite{gao2019multi} and the experimental results of \cite{fan2019cryo,fan19a}, lies at the low SNR region where utilizing higher-moment information becomes particularly beneficial even without introducing additional independent measurements for those higher moments. Empirically, we observe that the algorithm can tolerate higher level of noise than what is allowed according to the traditional Davis-Kahan theorem~\cite{DK1970}. In addition, the performance of the single frequency-$k$ class averaging algorithm improves as $k$ increases up to a critical frequency index determined by the spectral gap, magnitudes of the top eigenvalues, and the noise level.

Besides the improved numerical stability due to increased spectral gap, using higher frequency information for class averaging can also be interpreted as leveraging the additional redundancy encoded in the consistency of the ``higher order moments,'' which is in line with our continued exploration for a ``geometric harmonic retrieval'' initiated in \cite{gao2019multi,fan2019cryo,fan19a}. Moreover, in contrast with the computationally demanding SDP approach in \cite{bandeira2015non} or the noise-type-dependent approximate message passing approach in \cite{PWBM2018}, the proposed MFCA algorithm is easily parallelizable as the eigen-decompositions for the class averaging matrices in each frequency channel are completely independent. 

\subsection{Organization of the paper}
\label{sec:organization-paper}

The rest of this paper is organized as follows. Section~\ref{sec:mult-freq-class-aver} introduces the MFCA algorithms; Section~\ref{sec:prelim} introduces the basic mathematical set-up and notations for the spectral analysis in the remainder of this paper; Section~\ref{sec:main-results} presents the main theoretical contributions; Section~\ref{sec:theor-analys-multi} interprets the admissibility of MFCA using the theoretical results;  Section~\ref{sec:noise_model} discusses the noise robustness for the algorithm under two probabilistic models. Section~\ref{sec:numerics} illustrates the efficacy of MFCA through some numerical experiments; Section~\ref{sec:concl-future-work} concludes and discusses potential future directions. The basics on group and representation theory and technical proofs are deferred to the Appendix.

%% file: algorithm.tex
\section{Multi-Frequency Class Averaging Algorithms}
\label{sec:mult-freq-class-aver}

Throughout our discussion involving multiple frequency channels, we will fix an integer $k_{\mathrm{max}}\geq 1$ for the total number of frequency channels considered. For each frequency $k=1,\dots,k_{\mathrm{max}}$, we construct a separate class averaging matrix by
\begin{align}
\label{eq:freq-k-ca-mat}
  H_{ij}^{\left( k \right)}=
  \begin{cases}
    e^{\iota k\theta_{ij}}&\textrm{if $\left( i,j \right)\in E$}\\
    0 & \textrm{otherwise}
  \end{cases}
\end{align}

\subsection{Single Frequency-$k$ Affinity Measure}
\label{sec:single}

The Hermitian matrix $H^{\left( k \right)}$ stores the empirical transport data under the $k$th irreducible representation of $\SO(2)$. We then normalize each $H^{\left( k \right)}$ using the same degree matrix $D$ as in \eqref{eq:degree-matrix}; note that all matrices $H^{\left( k \right)}$ share the same sparsity pattern determined by $E$. After performing eigen-decomposition for $\widetilde{H}=D^{-1/2}H^{\left( k \right)}D^{-1/2}$, we keep the top $\left( 2k+1 \right)$ eigenvectors $\psi_1^{\left( k \right)},\dots,\psi_{2k+1}^{\left( k \right)}\in\mathbb{C}^{N}$ and define the  embedding
\begin{align}
\label{eq:embedding-freq-k}
  \Psi^{\left( k \right)}:\left\{I_i\right\}_{i=1}^{N}&\longrightarrow \mathbb{C}^{2k+1}\\
    I_i&\longmapsto \left( \psi_1^{\left( k \right)} \left( i \right), \dots,\psi_{2k+1}^{\left( k \right)} \left( i \right) \right).\nonumber
\end{align}
We compute the affinity measure between $I_i$ and $I_j$ at frequency $k$ as
\begin{align}
\label{eq:affinity-freq-k-prac}
  A^{\left( k \right)}_{ij}:=\frac{\left| \left\langle \Psi^{\left( k \right)} \left( I_i \right), \Psi^{\left( k \right)} \left( I_j \right) \right\rangle \right|}{\left\| \Psi^{\left( k \right)} \left( I_i \right) \right\|\left\| \Psi^{\left( k \right)} \left( I_j \right) \right\|},\quad 1\leq i\neq j\leq N.
\end{align}
Obviously, $\Psi^{\left( 1 \right)}=\Psi$ and $A_{ij}^{\left( 1 \right)}=A_{ij}$ in the traditional class averaging. We can perform $\kappa$-nearest neighbor search using the affinity measure $A_{ij}^{\left( k \right)}$ computed from an individual frequency $k$. The rationale behind the specific forms of \eqref{eq:embedding-freq-k} and \eqref{eq:affinity-freq-k-prac} is the core of this paper. In a nutshell, we use a $\left( 2k+1 \right)$-dimensional embedding because by Theorem~\ref{thm:eigenvalue-asym} and Theorem~\ref{thm:specgap} we expect a spectral gap occurring between the $\left( 2k+1 \right)^\text{th}$ and $\left( 2k+2 \right)^\text{th}$ eigenvector of $H^{\left( k \right)}$ (counting multiplicities). The affinity measure \eqref{eq:affinity-freq-k-prac} is related to the closeness of two viewing directions by the relation \eqref{eq:inner-product} in Theorem~\ref{thm:inner-product}. 

\subsection{Combining Information from Multiple Frequencies}
\label{sec:combine}
Since each affinity measure in~\eqref{eq:affinity-freq-k-prac} reflects the closeness of two viewing directions, combining those scores together can enforce the consistency of the classification results at each frequency and improve the overall accuracy. We propose one way to aggregate the single frequency affinity measure as
\begin{align}
 \label{eq:all-frequency-affinity-prac}
 A^{\textrm{All}}_{ij}:= \prod_{k = 1}^{k_\mathrm{max}}   A_{ij}^{(k)}.
\end{align}
We choose aggregation \eqref{eq:all-frequency-affinity-prac} because the affinity measure \eqref{eq:affinity-freq-k-prac} is related to the viewing angle by the relation \eqref{eq:inner-product} in Theorem~\ref{thm:inner-product}. In particular, comparing \eqref{eq:inner-product} and \cite[Theorem~6]{hadani2011representation2} tells us that
\begin{align*}
  A_{ij}^{\left( k \right)}=A_{ij}^k\qquad\textrm{for all $1\leq i\neq j\leq N$.}
\end{align*}
 We defer more detailed discussions of the geometric relation of this algorithm to Section~\ref{sec:main-alg-struct} and Section~\ref{sec:theor-analys-multi}.

\begin{remark}
   Note that \eqref{eq:affinity-freq-k-prac} and \eqref{eq:all-frequency-affinity-prac} are not the only ways to distill and aggregate the affinity information from multiple irreducible representations. Other natural alternatives include
\begin{align}
\label{eq:affinity-freq-k}
  S^{\left( k \right)}_{ij}:=2\left(\frac{\left| \left\langle \Psi^{\left( k \right)} \left( I_i \right), \Psi^{\left( k \right)} \left( I_j \right) \right\rangle \right|}{\left\| \Psi^{\left( k \right)} \left( I_i \right) \right\|\left\| \Psi^{\left( k \right)} \left( I_j \right) \right\|}\right)^{\frac{1}{k}}-1,\quad 1\leq i\neq j\leq N
\end{align}
which in the noiseless scenario satisfies
\begin{align*}
  S^{\left( k \right)}_{ij} = S^{\left( 1 \right)}_{ij}=2A_{ij}-1, \quad \textrm{for all $k\geq 1$}.
\end{align*}
Therefore, it is natural to combing all $G^{\left( k \right)}_{ij}$ by arithmetic averaging
\begin{align}
  \label{eq:all-frequency-affinity} S^{\textrm{All}}_{ij}:=\frac{1}{k_{\mathrm{max}}}\sum_{k=1}^{k_{\mathrm{max}}}S_{ij}^{\left( k \right)}.
\end{align}
However, our empirical experiments suggest that it is numerically much more stable to avoid taking $k$th roots for large values of $k$. We provide a brief interpretation of this phenomenon in Section~\ref{sec:theor-analys-multi}.
\end{remark}

There can be other approaches to combine the affinity scores from multiple frequencies, such as weighted average among different frequencies or majority voting. We will explore other ways to integrate multi-frequency information in the future. 

%% file: prelim_v2.tex
\section{Preliminaries for the Spectral Analysis of MFCA}
\label{sec:prelim}
In this section, we introduce our set-up and notations for the spectral analysis of MFCA. For additional concepts in the relevant group and representation theory and Wigner $D$-matrix, please refer to~\ref{sec:app_rep}.  

\subsection{Set-up}
\label{sec:set-up}
Throughout this paper, we view $\SO( 3 )$ as a
$\SO ( 2 )$-bundle over the $2$-dimensional
sphere $S^2$ in $\mathbb{R}^3$. For any $d\in\mathbb{N}_+$, we view $\mathbb{C}^d$ as a Hilbert product space equipped
with the canonical Hermitian inner product induced from the
standard Euclidean inner product on $\mathbb{R}^d$. We will
distinguish two different types of group actions on $\SO
\left( 3 \right)$: If $g \in\SO( 3 )$, $g$ acts
on elements of $\SO( 3 )$ by left
multiplication, denoted as
\begin{equation*}
  g \vartriangleright x:=g x,\qquad\forall g, x \in\SO( 3 ).
\end{equation*}
If $ w \in\SO ( 2 )\subset \SO( 3 )$,
unless otherwise specified, $w$ is assumed to be uniquely identified
with an $\SO( 3 )$ element by 
\begin{equation}
\label{eq:SO2-form1}
   w = w \left( \theta \right)=\begin{pmatrix}
    \cos \theta & -\sin \theta & 0\\
    \sin \theta  &   \cos \theta & 0\\
    0 & 0 & 1
  \end{pmatrix},\qquad \text{for } \theta \in \left[ 0,2\pi \right),
\end{equation}
and acts on elements of $\SO \left( 3
\right)$ by right multiplication, i.e.,
\begin{equation*}
  x \vartriangleleft w :=x w,\qquad\forall x \in \SO \left( 3
  \right),\,\, w  \in\SO ( 2 ).
\end{equation*}
Unless confusions arise, we will also denote
$x \vartriangleleft g=:xg$, $g,x\in\SO( 3 )$ for
the right action of $\SO( 3 )$ on itself, when
the context is clear.

Following the convention of \cite{hadani2011representation2}, we denote the transport data between $x,y\in \SO( 3 )$ by $T \left( x,y \right)$, the unique $\SO \left( 2
\right)$ element satisfying
\begin{equation}
  \label{eq:parallel-transport-operator}
  x\vartriangleleft T \left( x,y \right)=t_{\pi \left(
      x \right),\pi \left( y \right)}y,
\end{equation}
where $t_{\pi \left(x \right),\pi \left(y\right)}$ is the parallel
  transport along the unique geodesic on $S^2$ connecting
  $\pi \left( y \right)$ to $\pi \left( x \right)$. The optimal alignment angle $\theta_{ij}$ computed from \eqref{eq:opt_align} can be used to construct  an approximation of the transport data between $x_i$ and $x_j$ (the observation frames of $I_i$ and $I_j$, respectively), at the presence of measurement and discretization error, by
  \begin{equation}
    \label{eq:empirical-transport-data}
    \widetilde{T} \left( x_i,x_j \right):=e^{\iota \theta_{ij}}.
  \end{equation}
We refer to the $\widetilde{T}\left( x_i,x_j \right)$'s as the \emph{empirical transport data}. As shown  in \cite{hadani2011representation2}, $T(x, y)$ satisfy the following properties:
  \begin{align}
      & T \left( x,y \right)=T \left( y,x
        \right)^{-1},\quad\forall x,y\in \SO(3) \tag{Symmetry}\\
      & T \left( g\vartriangleright x, g\vartriangleright y
        \right)=T \left( x,y \right), \quad\forall x,y\in \SO(3),\,\,\forall g\in \SO( 3 ) \tag{Invariance}\\
      & T \left( x\vartriangleleft w_1, y\vartriangleleft
        w_2 \right)=w_1^{-1}T \left( x,y \right)w_2,
        \quad\forall x,y\in \SO(3),\,\,\forall w_1, w_2\in \SO
        \left( 2 \right). \tag{Equivariance}
  \end{align}
If $\rho:\SO ( 2 )\rightarrow\mathbb{C}$ is any
unitary representation of $\SO \left( 2
\right)$ on $\mathbb{C}$, then the three properties above can also be
cast into
\begin{align}
      & \rho\left(T \left( x,y \right)\right)=\overline{\rho\left(T \left( y,x
        \right)\right)},\quad\forall x,y\in \SO( 3 ) \tag{Symmetry}\\
      & \rho\left(T \left( g\vartriangleright x, g\vartriangleright y
        \right)\right)=\rho\left(T \left( x,y \right)\right), \quad\forall x,y\in \SO( 3 ),\,\,\forall g\in \SO( 3 ) \tag{Invariance}\\
      & \rho\left(T \left( x\vartriangleleft w_1, y\vartriangleleft
        w_2 \right)\right)=\overline{\rho\left(w_1\right)}\,\rho\left(T \left( x,y \right)\right)\,\rho\left(w_2\right),
        \quad\forall x,y\in \SO \left(3 \right) ,\,\,\forall w_1, w_2\in \SO
        \left( 2 \right). \tag{Equivariance}
  \end{align}
We shall only assume the symmetry to be strictly satisfied by the empirical transport data; the other properties will be assumed to hold only approximately. To simplify notations, we denote for any $k\in\mathbb{Z}$
\begin{equation}
\label{eq:Tk}
  T^{\left( k \right)}\left( x,y \right):=\rho_k \left( T \left( x,y \right) \right),\quad\forall x,y\in\SO( 3 )
\end{equation}
where $\rho_k:\SO ( 2 )\rightarrow\mathbb{C}$ is the unique
unitary irreducible representation of $\SO ( 2 )$ with character $k\in\mathbb{Z}$. The corresponding notation for the empirical transport data is $\widetilde{T}^{\left( k \right)} \left( x_i,x_j \right)$. 

In any of these irreducible representations, the empirical transport
data $\{\widetilde{T}^{\left( k \right)}\left(x_i,x_j \right)\mid 1\leq i,j\leq N\}$
approximate the ground truth transport data $\left\{T^{(k)}(x_i, x_j)\mid 1\leq i,j\leq N\right\}$ only when the viewing directions $\pi(x_i)$ and $\pi(x_j)$ are
close to each other, in the sense that the vectors $\pi (x_i)$ and $\pi (x_j)$ belong to some small spherical cap of opening angle $\alpha \in [0,2\pi)$. 

\subsection{Function on $\SO(3)$ and Isotypic Decomposition}
We will use the shorthand notation $\mathcal{H} = \mathbb{C}( \SO(3) )$ for the Hilbert space of smooth complex valued functions on $\SO( 3 )$, with standard Hermitian inner product
\begin{equation}
\label{eq:inner_product}
\langle f_1, f_2 \rangle_\mathcal{H} = \int_{\SO(3)} f_1(x)
\,\overline{f_2(x)} \,\mathrm{d}x,\quad f_1,f_2\in\mathcal{H}.
\end{equation}
Here $\mathrm{d}x$ denotes the normalized Haar measure on $\SO(3)$. 

The left and right actions of the group elements induce corresponding actions on the Hilbert space $\mathcal{H}$ of complex-valued functions over $\SO(3)$:
\begin{equation}
\label{eq:Gactf}
  \begin{aligned}
    g \cdot s \left( x \right)&:= s \left(
      g^{-1}\vartriangleright x \right),\quad\forall
    f\in\mathcal{H}, x\in\SO( 3 ),g\in\SO \left(
      3 \right). \\
    w \cdot s \left( x \right)&:=s \left( x \vartriangleleft w \right),\quad\forall s \in\mathcal{H}, x\in\SO( 3 ), w \in\SO ( 2 ).
  \end{aligned}
\end{equation}

The Hilbert space $\mathcal{H}$ can also be considered as a unitary representation of $\SO ( 2 )$. Let $\rho_k:\SO ( 2 )\rightarrow\mathbb{C}$ be the unique
irreducible unitary representation of $\SO ( 2 )$
of character $k\in\mathbb{Z}$. $\mathcal{H}$ admits an isotypic decomposition
\begin{equation}
\label{eq:isotypic-decomp-1}
  \mathcal{H}=\bigoplus_{k\in\mathbb{Z}}\mathcal{H}_k, 
\end{equation}
where 
\begin{align}
 \label{eq:SO2-isotypic}
  \mathcal{H}_{k}:=\left\{ s\in\mathcal{H}\mid s \left( x \vartriangleleft w \right)= \rho_k(w) s \left( x \right)\,\,\textrm{for all $x\in\SO( 3 )$ and $ w \in \SO(2)$} \right\}.
\end{align}
Note that $\SO\left( 3 \right)$ acts on $\mathcal{H}_k$
unitarily from the left by
\begin{equation*}
  g \cdot s \left( x \right):=  s \left(
g^{-1}\vartriangleright x \right),\quad\forall g \in \SO
  \left( 3 \right),\,\, s \in\mathcal{H}_k,\,\, x\in \SO(3).
\end{equation*}
Each $\mathcal{H}_k$ thus admits an isotypic decomposition
with respect to $\SO ( 3 )$, written as
\begin{equation}
\label{eq:isotypic-decomp-2}
  \mathcal{H}_k=\bigoplus_{n\in\mathbb{N}_{\geq 0}}\mathcal{H}_{n,k}
\end{equation}
where $\mathcal{H}_{n,k}$ denotes the isotypic component corresponding to the unique irreducible representation of $\SO ( 3 )$ of dimension $\left( 2n+1 \right)$, for $n=0,1,\dots$. 
An important observation is that each
$\mathcal{H}_{n,k}$ in~\eqref{eq:SO2-isotypic} is of multiplicity $0$ or $1$ in
$\mathcal{H}_k$:
\begin{theorem}[{\cite[Theorem~7]{hadani2011representation2}}]
\label{thm:mult-one}
  If $n<\left| k \right|$ then
  $\mathcal{H}_{n,k}=0$. Otherwise, $\mathcal{H}_{n,k}$ is
  isomorphic to the unique irreducible representation of
  $\SO ( 3 )$ of dimension $\left( 2n+1 \right)$.
\end{theorem}

%% file: model.tex
\section{Main Theoretical Results}
\label{sec:main-results}

\subsection{Generalized Parallel Transport Operators}
\label{sec:gen-para-transp-op}

The motivation for considering these isotypic decompositions
is to study the top eigenspace of the \emph{generalized
  parallel transport operator} $T^{\left( k
  \right)}:\mathcal{H}\rightarrow\mathcal{H}$, defined as
\begin{equation}
  \label{eq:defn-gen-parallel-transp-op}
  \left(T^{\left( k \right)}s\right)\left( x
  \right):=\int_{\SO(3) }\rho_k \left( T \left( x,y \right)
  \right)s \left( y \right)\,\mathrm{d}y=\int_{ \SO(3) }
  T^{\left( k \right)} \left( x,y \right)s \left( y
  \right)\,\mathrm{d}y,\quad\forall 
  s\in\mathcal{H},\,\, x\in \SO(3), 
\end{equation}
for all $k\in\mathbb{Z}$. When $k=1$, $T^{\left( k \right)}$
reduces to the \emph{parallel transport operator}
$T:\mathcal{H}\rightarrow\mathcal{H}$ defined in
\cite[\S2.3]{hadani2011representation2}. Similar to
\cite[\S2.3.1]{hadani2011representation2}, we can localize
the generalized parallel transport operator $T^{\left( k
  \right)}$ for any $k\in\mathbb{Z}$ as
\begin{equation}
  \label{eq:defn-gen-parallel-transp-op-loc}
  \left(T_h^{\left( k \right)}s\right)\left( x
  \right):=\int_{B \left( x, \alpha \right)}\rho_k \left( T \left( x,y \right)
  \right)s \left( y \right)\,\mathrm{d}y=\int_{B \left( x, \alpha \right)}
  T^{\left( k \right)} \left( x,y \right)s \left( y
  \right)\,\mathrm{d}y,\quad\forall
  s\in\mathcal{H},\,\,  x\in \SO(3), 
\end{equation}
where $B \left( x, \alpha \right)=\left\{ y\in \SO(3) \mid \left( \pi
    \left( x \right),\pi \left( y \right) \right)>\cos
  \alpha =:1-h \right\}$. Using the symmetry, invariance, and equivariance
of the transport data (Section~\ref{sec:set-up}), we
establish the following basic properties of $T^{\left( k
  \right)}$ for any $k\in\mathbb{Z}$:
\begin{enumerate}[(1)]
\item\label{item:1} $T^{\left( k \right)}$ is
  self-adjoint. This can be seen from the symmetry of
  transport data: for all $s,w\in\mathcal{H}$, we have
  \begin{equation*}
    \begin{aligned}
      \left\langle T^{\left( k \right)}s,w
    \right\rangle_{\mathcal{H}}&=\int_{ \SO(3) }\!\!\int_{\SO(3)}\rho_k
    \left( T \left( x,y \right) \right)s \left( y
    \right)\overline{w \left( x \right)}\,\mathrm{d}y\mathrm{d}x\\
     &=\int_{\SO(3) }\!\!\int_{ \SO(3) }s \left( y
    \right)\overline{\rho_k
    \left( T \left( y,x \right) \right)w \left( x \right)}\,\mathrm{d}y\mathrm{d}x=\left\langle s,T^{\left( k \right)}w \right\rangle_{\mathcal{H}}.
    \end{aligned}
  \end{equation*}
\item\label{item:2} $T^{\left( k \right)}$ commutes with the
  action of $\SO ( 3 )$ on $\mathcal{H}$: by the
  invariance of transport data we have for all $g\in \SO
  \left( 3 \right)$ and $s\in\mathcal{H}$, $x\in \SO(3)$, 
  \begin{equation*}
    \begin{aligned}
      \left(T^{\left( k \right)} \left( g\cdot s
        \right)\right)\left( x \right)&=\int_{ \SO(3) }\rho_k
      \left( T \left( x,y \right)\right)s \left(
        g^{-1}\vartriangleright y
      \right)\,\mathrm{d}y \\
      & \stackrel{z:=g^{-1}\vartriangleright
      y}{=\!=\!=\!=\!=\!=\!=}\int_{ \SO(3) }\rho_k
      \left( T \left( g\vartriangleright
          \left(g^{-1}\vartriangleright
            x\right),g\vartriangleright z \right)\right)s
      \left(z\right)\,\mathrm{d}z\\
      &=\int_{ \SO(3) }\rho_k
      \left( T \left( g^{-1}\vartriangleright
            x, z \right)\right)s
      \left(z\right)\,\mathrm{d}z=\left( T^{\left( k
          \right)}s \right)\left( g^{-1}\vartriangleright x
      \right)=\left( g\cdot \left( T^{\left( k \right)}s \right) \right)\left( x \right).
    \end{aligned}
  \end{equation*}
\item\label{item:3} $\bigoplus_{\ell\neq
    -k}\mathcal{H}_{\ell}\subset \mathrm{\ker}\,T^{\left( k
    \right)}$, and $T^{\left( k \right)}$ can
  be viewed as an operator from $\mathcal{H}_{-k}$ to
  itself. This can be verified using the equivariance of
  $T^{\left( k \right)}$. First, note that for
  any $s\in \mathcal{H}$ we have $T^{\left( k
    \right)}s\in\mathcal{H}_{-k}$, since for any  $w\in\SO
  \left( 2 \right)$ we have
  \begin{equation*}
    \begin{aligned}
      w\cdot \left(T^{\left( k \right)}s\right) \left( x
      \right)&=\left(T^{\left( k \right)}s\right)\left(
        x\vartriangleleft w \right)\\
        &=\int_{B \left( x,\alpha \right)}\rho_k \left( T \left( x\vartriangleleft w,y
        \right)\right)s\left( y
      \right)\,\mathrm{d}y=\int_{B \left( x, \alpha \right)} \overline{\rho_k \left( w
        \right)}\rho_k \left( T \left( x,y \right)
  \right)s \left( y \right)\,\mathrm{d}y\\
        &=\rho_{-k} \left( w
        \right)\int_{B \left( x, \alpha \right)} \rho_k \left( T \left( x,y \right)
  \right)s \left( y \right)\,\mathrm{d}y=\rho_{-k}\left( w
  \right) \left(T^{\left( k \right)}s\right) \left( x \right).
    \end{aligned}
  \end{equation*}
This proves that $T^{\left( k \right)}$ maps $\mathcal{H}$
into $\mathcal{H}_{-k}$, by the definition of isotypic
decomposition \eqref{eq:SO2-isotypic} with respect to the
$\SO ( 2 )$ action. The conclusion that $\bigoplus_{\ell\neq
    -k}\mathcal{H}_{\ell}\subset \mathrm{\ker}\,T^{\left( k
    \right)}$ then follows from Schur's Lemma \cite[Theorem 2.1]{brocker2013representations}.
\end{enumerate}
The arguments above can be applied to $T^{\left( k
  \right)}_h$, \emph{mutatis mutandis}, and thus the same
properties hold for the local generalized parallel
transport operator. Invoking Schur's Lemma for a second
time, we know that $T_h^{\left( k \right)}$ acts on
$\mathcal{H}_{n,-k}$ as a scalar, i.e.,
\begin{equation}
\label{eq:schur}
  T^{\left( k \right)}_h\big|_{\mathcal{H}_{n,-k}}=\lambda_n^{\left( k \right)}\left( h \right)\mathrm{Id}\big|_{\mathcal{H}_{n,-k}}.
\end{equation}
The multiplicity-one theorem (Theorem~\ref{thm:mult-one})
tells us that $\lambda_n^{\left( k \right)}=0$ for all
$0\leq n<\left| k \right|$. In order to calculate the remaining
$\lambda_n^{\left( k \right)}$'s ($n\geq \left| k \right|$) explicitly, it suffices to
fix a point $x_0\in\SO ( 3 )$, and pick an
arbitrary function $u\in\mathcal{H}_{n,-k}$ with $u \left(
  x_0 \right)$, and use relation $\lambda_n^{\left( k
  \right)}=\left( T_h^{\left( k \right)} u \right)\left( x_0
\right)/u\left( x_0 \right)$. We will defer such
computations for $0<h\ll 1$ to
Section~\ref{sec:spectral-property}. Next subsection
summarizes these properties, in preparation for the
discussion on the main algebraic structure of the
\emph{generalized intrinsic model} in Section~\ref{sec:main-alg-struct}.

In \cite{singer2011viewing,hadani2011representation2}, it was argued that the Hermitian matrix $H$ in~\eqref{eq:H} should be understood as the discretization (under uniform random sampling on $\SO( 3)$) of an integral operator $T_h^{(1)}$. 
Consequently, many properties of the local transport data matrix $H$ can be studied through its ``continuous limit'' $T_h$, especially the eigenvalues and eigenvectors, which converge to the eigenvalues and eigenfunctions of $T_h$ in an appropriate sense \cite{KG2000}; this perspective is common in the manifold learning literature \cite{BelkinNiyogi2005,BelkinNiyogi2007,coifman2006diffusion,SingerWu2012VDM,HDM2016}. In the class averaging setting, the integral operator $T_h$ enjoys many useful invariance and equivariance properties, which makes it relatively straightforward to study its spectral data using representation theoretic tools. Hadani and Singer noticed that $T_h$ acts on the subspace $\mathcal{H}_{-1}$ of $\mathcal{H}$. 
The space $\mathcal{H}_{-1}$ is also canonically identified with the linear space of sections of a complex line bundle over $\SO( 3 )$ induced by the unitary irreducible representation of $\Unitary \left( 1 \right)$ with character $k=1$ \cite{goldberg1967spin,campbell1971tensor,boyle2016should,eastwood1982edth,marinucci2011random,Malyarenko2011}. Furthermore, $T_h$ commutes with the induced left action of $\SO( 3 )$ on $\mathcal{H}_{-1}$, which by Schur's theorem indicates that the eigenspaces of $T_h$ coincides with the isotypic components of $\mathcal{H}_{-1}$ under the left $\SO( 3 )$ action. In particular, this mechanism can be used to show that the top eigenspace of $T_h$ is the unique isotypic component of $\mathcal{H}_{-1}$ corresponding to the unique three-dimensional unitary irreducible representation of $\SO( 3 )$ for all sufficiently small $h>0$, and that the affinity measure $2A_{ij} - 1$ is exactly identical with the cosine value of the viewing angle between $I_i$ and $I_j$ in the noise-free setting.

\subsection{Spectral Properties of the Local Parallel
  Transport Operator}
\label{sec:spectral-property}

In this subsection we summarize the spectral properties of
$T_h^{(k)}$ for $h \ll 1$ (which is the relevant
regime for class averaging). Proofs for the main theorems discussed in this subsection are deferred to
\ref{sec:spectral-analysis}.  
These proofs essentially
follow the proof ideas of \cite[Theroem~3 and
Theorem~4]{hadani2011representation2}, with technical
modification due to the complication of Jacobi polynomials
--- unlike the case for the Legendre polynomials involved in
the analysis of single-frequency class averaging, no sharp
Bernstein-type inequality is known for Jacobi polynomials
arising from the Wigner $d$-matrices. We refer interested
readers to discussions and conjectures in
\cite{CGW1994,HS2014,KKT2018} for Bernstein-type
inequalities for Jacobi polynomials.

\begin{theorem}[Eigenvalues of $T_h^{\left( k \right)}$ for
  small $h\ll1$]\label{thm:eigenvalue-asym}
\label{thm:eval}
The operator $T_h^{(k)}$ has a discrete spectrum
$\lambda^{k}_n(h)$ for all $n \in \mathbb{N}$, and
$\lambda_n^{\left( k \right)}=0$ for all $0\leq n< \left| k
\right|$. For $n\geq \left| k \right|$ and $h\in (0,2]$, the dimension of the
eigenspace of $T_h^{\left( k \right)}$ corresponding to
$\lambda_n^{\left( k \right)}$ is
$2n+1$. In addition, $\lambda_k^{(k)}$ and $\lambda_{k+1}^{(k)}$ have the following expressions:
\begin{align}
\lambda_k^{\left( k \right)}\left( h
  \right)& = \frac{1-\left( 1-h/2
           \right)^{k+1}}{k+1},\label{eq:lambda-k-k_main}\\
   \lambda_{k+1}^{\left( k \right)}\left( h
  \right)&= \frac{2(k+1)( 1- (1 - h/2)^{k+2})}{k+2} - \frac{(2k+1)(1 - (1-h/2)^{k+1})}{k + 1}.\label{eq:lambda-k-k+1_main}
\end{align}
In the regime $h \ll
1$, the eigenvalue $\lambda_n^{(k)}(h)$ ($n\geq \left| k
\right|$) adopts asymptotic expansion
\begin{equation}
\label{eq:eval}
\lambda^{(k)}_n (h) = \frac{1}{2} h - \frac{1}{8}\left(
  n^2+n-k^2 \right) h^2 + O(h^3).
\end{equation}
\end{theorem}

\begin{figure}[htpb]
\centering
\includegraphics[width=1.0\textwidth]{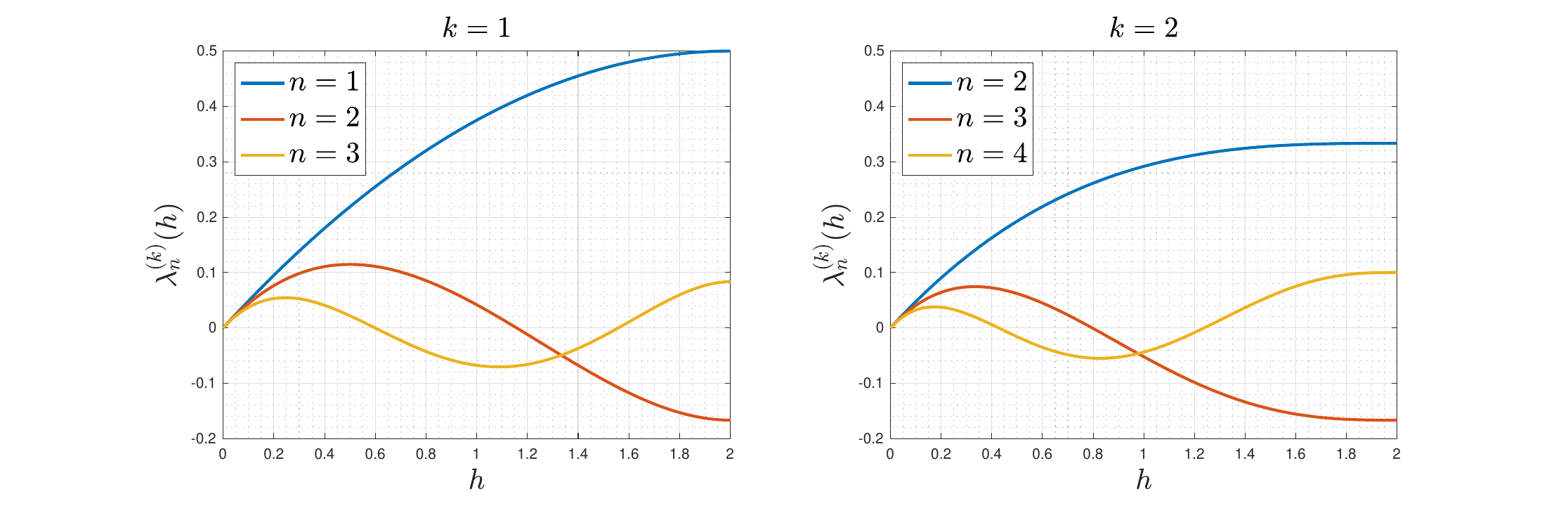}
\caption{{\small The top three eigenvalues 
$\lambda_n^{\left( k \right)}\left( h \right)$ of operator 
$T_h^{\left( k \right)}$, for $k=1$ (left) and $k=2$
  (right) over interval $h \in (0, 2]$.
}}
\label{fig:eval}
\end{figure}

\begin{remark}
  When $k=1$, Theorem~\ref{thm:eval} reduces to \cite[Theorem~3]{hadani2011representation2}. 
\end{remark}
The proof of Theorem~\ref{thm:eigenvalue-asym} in \ref{sec:proof-theor-refthm-asym}
actually proves the stronger conclusion that each eigenvalue $\lambda_n^{(k)}(h)$ is a
polynomial in $h>0$ of degree $\left( n+1 \right)$ whenever $n\geq \left| k
\right|$. The key step in the proof is identifying that the $(-k, -k)$ entry of the Wigner D-matrix $D^{n}_{-k, -k}(x) \in \mathcal{H}_{n, k}$ for $n \geq |k|$ and is an appropriate function $u$ for calculating the eigenvalues. The largest three eigenvalues for cases $k = 1$
and $k = 2$ can be explicitly written out as
\begin{align}
\lambda_1^{(1)}\left( h \right) & = \frac{1}{2}h - \frac{1}{8} h^2, \nonumber \\
\lambda_2^{(1)}\left( h \right) & = \frac{1}{2}h - \frac{5}{8} h^2 + \frac{1}{6}h^3,  \nonumber \\
\lambda_3^{(1)}\left( h \right) & = \frac{1}{2}h - \frac{11}{8} h^2 + \frac{25}{24}h^3 - \frac{15}{64}h^4,
\label{eq:eval_k1}
\end{align}
and
\begin{align}
\lambda_2^{(2)}\left( h \right) & = \frac{1}{2}h - \frac{1}{4} h^2 + \frac{1}{24} h^3, \nonumber \\
\lambda_3^{(2)}\left( h \right) & = \frac{1}{2}h -  h^2 + \frac{13}{24}h^3 - \frac{3}{32}h^4, \nonumber \\
\lambda_4^{(2)}\left( h \right) & = \frac{1}{2}h - 2 h^2 + \frac{57}{24}h^3 - \frac{70}{64}h^4 + \frac{7}{40} h^5.
\label{eq:eval_k2}
\end{align}
 Plots of $\lambda_{k+i}^{(k)}$, for $i = 0, 1, 2$ are provided
 in Figure~\ref{fig:eval}. 
 
\begin{corollary}
\label{cor:1}
  As $k$ increases, the eigenvalue $\lambda_k^{(k)}$  decreases and $\lim_{k \rightarrow \infty} \lambda_k^{(k)} = 0$.
\end{corollary}
\begin{proof}[Proof of Corollary~\ref{cor:1}]
Based on Theorem~\ref{thm:eigenvalue-asym}, the difference between $\lambda_{k+1}^{(k+1)}$ and $\lambda_k^{(k)}$ for $k \geq 1$ is, 
\begin{align}
    \label{eq:eigval_diff}
    \lambda_{k+1}^{(k+1)} - \lambda_k^{(k)} & = \frac{1 - (1-h/2)^{k+2}}{k+2} - \frac{1 - (1-h/2)^{k+1}}{k+1} =  \frac{-1 + (1-h/2)^{k+1}(1 + \frac{h}{2}(k+1))}{(k+1)(k+2)} \nonumber \\
    & \overset{(\text{a})}{<} \frac{-1 + (1 - h/2)^{k+1}(1+ h/2)^{k+1}}{(k+1)(k+2)} = \frac{-1 + (1 - h^2/4)^{k+1}}{(k+1)(k+2)} < 0, 
\end{align}
where (a) is based on the fact that $(1 + \frac{h}{2})^{k+1} > 1 + (k+1) \frac{h}{2}$ for $h \in (0, 2]$ via Taylor expansion. In addition, since $ 0 \leq 1 - h/2 < 1$, $\lim_{k \rightarrow \infty} \lambda_k^{(k)}(h) = \lim_{k \rightarrow \infty} \frac{1 - (1-h/2)^{k+1}}{k+1} = 0$. 
 \end{proof}

\noindent This is an important observation for determining the maximum frequency cutoff, which will be further discussed in Section~\ref{sec:noise_model}.

It is natural to conjecture
 that the top eigenspace of $T_h^{\left( k \right)}$ is the
 $\left( 2k+1 \right)$-dimensional space corresponding to
 eigenvalue $\lambda_k^{\left( k \right)}\left( h \right)$
 for sufficiently small $h>0$. Moreover, denote
 \begin{equation}
   \Delta_k:=\argmax_{h\in (0,2]}\lambda_{k+1}^{\left( k \right)}\left( h \right)=\frac{1}{k+1},
   \label{eq:delta_k}
 \end{equation}
we have the following characterization of the \emph{spectral
gap} for $T_h^{\left( k \right)}$ in the regime $0<h\ll1$.

\begin{theorem}
\label{thm:specgap}
For every value of $h \in (0, 2]$, the largest eigenvalue of
$T_h^{(k)}$ is $\lambda_k^{(k)}(h)$. In addition, for every
value of $h \in (0, \Delta_k]$, the spectral gap
$G^{(k)}(h)$ between the largest and the second largest
eigenvalue of $T_h^{\left( k \right)}$ is
\begin{equation}
\label{eq:specgap}
G^{(k)}(h) = \lambda_{k}^{(k)} - \lambda_{k+1}^{(k)} = \frac{2 - (1-h/2)^{k+1}\left((k+1) h + 2\right)}{ k + 2}.
\end{equation}
\end{theorem}

Again, when $k=1$, Theorem~\ref{thm:specgap} reduces to
\cite[Theorem~4]{hadani2011representation2}. The main
technicality of the proof of Theorem~\ref{thm:specgap},
which is deferred to \ref{sec:proof-theor-refthm-gap}, is to
show that $\lambda^{(k)}_n(h) \leq \lambda_{k+1}^{(k)}(h)$ for every
$h \in (0, \Delta_k ]$ and $n \geq k+1$, which appears evident from
Figure~\ref{fig:eval}. For small $0<h\ll  \Delta_k$, the
spectral gap is approximately 
\begin{equation}
G^{(k)}(h) \sim \frac{ 1 + k }{4} h^2,
\end{equation}
which gets larger as the ``angular frequency''
$k\in\mathbb{N}$ increases. More generally, we have the following Corollary. 
\begin{corollary}
\label{cor:2}
The spectral gap $G^{(k)}(h)$ increases as $k$ increases from $1$ to $k_\mathrm{max} = \left \lfloor \frac{1}{h} \right \rfloor - 1 $.   
\end{corollary}

\begin{proof}[Proof of Corollary~\ref{cor:2}]
We show that for any $k \geq 2$, the difference $G^{(k)}(h) - G^{(k-1)}(h)$ is always positive for any $h \in (0, \Delta_k]$. To begin with, we explicitly write out the difference as
\begin{equation*}
    \begin{aligned}
    G^{(k)}(h) - G^{(k-1)}(h) &= \frac{2 - (1-h/2)^{k+1}((k+1)h+2)}{k+2} - \frac{2 - (1-h/2)^k(kh+2)}{k+1}\\
    &= \frac{(1-h/2)^k((k+1)^2h^2+2kh+4) - 4}{2(k+1)(k+2)}\\
    &= \frac{(1-h/2)^k}{2(k+1)(k+2)}\underbrace{\left((k+1)^2h^2 + 2kh +4 - 4(1-h/2)^{-k}\right)}_{=:\xi(h)}\\
    &= \frac{(1-h/2)^k}{2(k+1)(k+2)} \xi(h),
    \end{aligned}
\end{equation*}
where $\xi(h)$ is defined as a function of $h$. Since the term in front of $\xi(h)$ is always positive for $h \in (0, \Delta_k]$, it suffices to show $\xi(h) > 0$ for any $k \geq 2$ and $h \in (0, \Delta_k]$. To this end, clearly $\xi(h) = 0$ when $h = 0$ then we can instead show the derivative of $\xi(h)$ is positive for any $h \in (0, \Delta_k]$. That is,
\begin{equation*}
    \frac{d\xi(h)}{dh} = 2(k+1)^2h + 2k - \frac{2k}{(1 - h/2)^{k+1}}.
\end{equation*}
Again, when $h = 0$ we observe that $\frac{d\xi(h)}{dh}|_{h = 0} = 0$. So in order to show $\frac{d\xi(h)}{dh} > 0$, for all $h \in (0, \Delta_k]$ we can instead check if the second order derivative of $\xi(h)$ is positive for any $h \in (0, \Delta_k]$. Indeed, we have
\begin{equation*}
\begin{aligned}
    \frac{d^2\xi(h)}{dh^2} &= 2(k+1)^2 - \frac{k(k+1)}{(1-h/2)^{k+2}} = \frac{(k+1)}{(1-h/2)^{k+2}}(2(k+1)(1-h/2)^{k+2} - k) \\
    &\overset{(\text{a})}{>} \frac{(k+1)}{(1-h/2)^{k+2}}(2(k+1)(1 - (k+2)h/2) - k) \overset{(\text{b})}{\geq} 0 
    \end{aligned}
\end{equation*}
where (a) comes from the inequality that $(1 - x)^{a} > 1 - xa$ for any $x \in (0, 1)$ and $a > 2$, (b) is satisfied since $2(k+1)(1 - (k+2)h/2) - k$ is linear and monotonically decreasing for $h$ and the equality only holds when $h = \Delta_k = \frac{1}{k+1}$. Therefore, we obtain that $\frac{d^2\xi(h)}{dh^2} > 0, \; \forall h \in [0, \Delta_k]$ and it follows that $\frac{d\xi(h)}{dh} > 0 , \; \forall h \in (0, \Delta_k]$, furthermore we can conclude that $G^{(k)}(h) - G^{(k-1)}(h) > 0$ for any $h \in (0, \Delta_k]$. 
\end{proof}

This justifies one benefit of
setting $k>1$ for class averaging, as larger spectral gaps
provide more robustness to noise corruption for $k$ satisfying $k < \frac{1}{h}-1$.
More detailed discussion on the performance of the algorithm under noise perturbation is in Section~\ref{sec:noise_model}. In practice, the choice of frequency cutoff depends on the neighborhood size, noise type and noise level and may need to be empirically identified.

\subsection{The Main Algebraic Structure: Generalized
  Intrinsic Model}
\label{sec:main-alg-struct}

Just as the intrinsic model established in
\cite{hadani2011representation2} equates the ``extrinsic
model'' $S^2$ with the ``intrinsic model'' of the top
eigenspace $\mathbb{W}$ of $T=T^{\left( 1 \right)}$, we will
generalize this correspondence to the setting for general
complex irreducible unitary representations of $\SO \left( 2
\right)$. More specifically, we establish the correspondence 
between the following two generalized models:

\begin{itemize}
\item \emph{Generalized Extrinsic Model}: For every point $x=x\left( \varphi,\vartheta,\psi\right) \in \SO ( 3 )$, denote by $\delta_x^{\left( k \right)}: \mathbb{C}
  \rightarrow \mathbb{C}^{2k+1}$ for the unique complex
  morphism sending $1 \in \mathbb{C}$ to the first (index-$\left(-k\right)$)
column of the Wigner $D$-matrix $D^{k}$ (detailed in~\ref{sec:app_rep}), i.e.,
  \begin{equation*}
    D_{\cdot,-k}^k \left( x \right)\!=\! \left( D_{-k,-k}^k \left( x \right), D_{-k+1,-k}^k \left( x \right), \dots, D_{k-1,-k}^k \left( x \right), D_{k,-k}^k \left( x \right)\right)^{\top}\!\in \!\mathbb{C}^{2k+1}.
  \end{equation*}
\item \emph{Generalized Intrinsic Model}: Define
  $\mathbb{W}^{\left( k \right)}$ as the top eigenspace of
  $T_h^{(k)}$, which by Theorem~\ref{thm:eval}
  and~\ref{thm:specgap}, is $(2k+1)$-dimensional. Set for
  every point $x\in\SO ( 3 )$ the map
\begin{equation}
\label{eq:main-morphism}
\varphi^{\left( k \right)}_x = \sqrt{1/(2k+1)}\cdot(\mathrm{ev}_x | \mathbb{W}^{\left( k \right)})^*: \mathbb{C} \rightarrow \mathbb{W}^{\left( k \right)},
\end{equation}
where $\mathrm{ev}_x: \mathcal{H} \rightarrow \mathbb{C}$
is the evaluation morphism at the point $x\in \SO ( 3 )$.
\end{itemize}

The main algebraic structure of the multi-frequency
intrinsic classification algorithm is summarized in the
following main theorem of this section.

\begin{theorem}
\label{thm:morph}
The morphism $\tau:\mathbb{C}^{2k+1}\rightarrow\mathcal{H}$
defined by
\begin{equation*}
  \begin{aligned}
    \tau:\mathbb{C}^{2k+1}&\longrightarrow\mathcal{H}\\
    v & \longmapsto \left( x\mapsto \sqrt{2k+1}\cdot
      \left(\delta_x^{\left( k \right)}\right)^{*} \left( v \right) \right)
  \end{aligned}
\end{equation*}
is an isomorphism between $\mathbb{C}^{2k+1}$ and $\mathbb{W}^{\left( k \right)} \subset
\mathcal{H}$ (as Hermitian vector spaces). Moreover, for
every $x\in \SO ( 3 )$ and $k=0,1,\dots$ there holds
\begin{equation}
\label{eq:composition-identity}
\tau \circ \delta^{\left( k \right)}_x= \varphi^{\left( k \right)}_x.
\end{equation}
\end{theorem}

The proof of Theorem~\ref{thm:morph} is deferred to \ref{sec:proof-theor-refthm-morph}. Our proof extends the arguments in the proof of
  \cite[Theorem~5]{hadani2011representation2}. A key
  observation is that the top eigenvector $\mathcal{H}\left(
    \lambda_{k}^{\left( k \right)}\left( h
    \right) \right)$ coincides with the isotypic subspace
  $\mathcal{H}_{k,-k}$ (see
  Section~\ref{sec:gen-para-transp-op}). Furthermore, Theorem~\ref{thm:morph} reveals the correspondence
between the generalized extrinsic and intrinsic models, in
terms of the viewing angle information they encode. This is
summarized in the following result.

\begin{theorem}
\label{thm:inner-product}
For every pair of frames $x, y \in \SO ( 3 )$, we have
\begin{equation}
\label{eq:inner-product}
\left|\langle \varphi^{\left( k \right)}_x(v),
  \varphi^{\left( k \right)}_y(u)\rangle_{\mathbb{W}^{\left(
        k \right)}}  \right | = \left( \frac{\left\langle
      \pi(x), \pi(y) \right\rangle + 1}{2}
\right)^k
\end{equation}
for any choice of unit-norm complex numbers $v, u\in \mathbb{C}$.
\end{theorem}
The proof of Theorem~\ref{thm:inner-product} is deferred to \ref{sec:proof-theor-refthm-inner-product}.

\begin{remark}
  When $k=1$, Theorem~\ref{thm:morph} and
  Theorem~\ref{thm:inner-product} reduce to
  \cite[Theorem~5]{hadani2011representation2} and
  \cite[Theorem~6]{hadani2011representation2}, respectively,
  up to a different scaling constant for $\tau$. The
  difference arises from our alternative, explicit construction of the
  isomorphism $\tau$ using Wigner $D$-matrices.
\end{remark}

%% file: theory_v2.tex
\section{Interpretation of the Theoretical Results for
  Multi-Frequency Class Averaging}
\label{sec:theor-analys-multi}

In this section, we interpret the MFCA algorithm stated in
Section~\ref{sec:mult-freq-class-aver} using the theoretical
results established in Section~\ref{sec:main-results}, and
provide conceptual explanations for the admissibility of MFCA in the noiseless regime.

First, under the assumption that the projection images
$\left\{ I_i\mid 1\leq i\leq N \right\}$ are produced from
orthonormal frames $\left\{ x_i\mid 1\leq i\leq N \right\}$
sampled i.i.d. uniformly on $\SO \left( 3 \right)$ with
respect to the normalized Haar measure, we view $\frac{1}{N} H^{\left( k \right)}$,
the scaled class averaging matrix at frequency $k$ defined in
Section~\ref{sec:mult-freq-class-aver}, as the discretization of the local parallel transport operator $T_h^{\left( k \right)}$. We know from
standard results \cite[Theorem~3.1]{KG2000} that the
 eigenvalues of $\frac{1}{N}H^{\left( k \right)}$ converges to the 
eigenvalues of the generalized localized parallel transport operator $T_h^{\left( k \right)}$ defined in
\eqref{eq:defn-gen-parallel-transp-op-loc} as the number of
samples $N$ goes to infinity and the opening angle $\alpha$ is sufficiently small. In
particular, this implies that for large sample size $N$, the
spectral gap of $\frac{1}{N}H^{\left( k \right)}$ converges to the
spectral gap of $T_h^{\left( k \right)}$, which, by
Theorem~\ref{thm:eval} and Theorem~\ref{thm:specgap}, is
roughly of size $\left( 1+k \right)h^2/4$ for $h \ll \frac{1}{k+1}$ and occurs between the $\left( 2k+1 \right)^{\text{th}}$ and the $\left( 2k+2
\right)^{\text{th}}$ eigenvalues of $H^{\left( k \right)}$ (ranked in
decreasing order). 

Moreover, as argued in
\cite[Theorem~2]{hadani2011representation2}, the MFCA
embedding $\Psi^{\left( k \right)}$ defined in
\eqref{eq:embedding-freq-k} corresponds to the morphism
\eqref{eq:main-morphism} in the following form:
\begin{align}
\label{eq:emp_space}
  \frac{\Psi^{\left( k \right)}\left( x_i\right)}{\left\| \Psi^{\left( k \right)}\left( x_i\right) \right\|} \approx \varphi_{x_i}^{\left( k \right)}\left( 1
  \right),\quad\textrm{for all $x_i\in \SO \left( 3 \right)$},
\end{align}
where $\left\| \cdot \right\|$ stands for the standard norm
on $\mathbb{C}^{2k+1}$. 
Combining~\eqref{eq:emp_space} with
Theorem~\ref{thm:inner-product} 
provides the justification for using $A^{(k)}$ to identify similar viewing angles, 
\begin{align}
\label{eq:approx-identity}
  A_{ij}^{\left( k \right)}=\frac{\left| \left\langle \Psi^{\left( k \right)} \left( I_i \right), \Psi^{\left( k \right)}
  \left( I_j \right) \right\rangle \right|}{\left\| \Psi^{\left( k \right)}
  \left( I_i \right) \right\|\left\| \Psi^{\left( k \right)} \left( I_j \right)
  \right\|} \approx \left|\langle \varphi^{\left( k \right)}_x(1),
  \varphi^{\left( k \right)}_y(1)\rangle_{\mathbb{W}^{\left(
        k \right)}}  \right | = \left( \frac{\left\langle
  \pi \left( x_i \right), \pi \left( x_j \right) \right\rangle+1}{2} \right)^k. 
\end{align}
This relation is demonstrated in the top rows of Figures~\ref{fig:prob_emb_scatter} and~\ref{fig:cryo_scatter}. In fact, Theorem~\ref{thm:inner-product} tells us that the affinity measure $S_{ij}^{\left( k \right)}$ defined in \eqref{eq:affinity-freq-k} coincides with the cosine value for the angle between the two viewing directions in the noiseless regime. The form of the approximation identity \eqref{eq:approx-identity} also suggests avoiding directly taking the $k$th root of the correlation between $\Psi^{\left( k \right)}\left( I_i \right)$ and $\Psi^{\left( k \right)}\left( I_j \right)$ as in \eqref{eq:affinity-freq-k} and \eqref{eq:all-frequency-affinity} since this approach loses control of the numerical relative error when $A^{(k)}_{ij}$ is close to 0. In contrast, it is advantageous to use the multiplicative forms \eqref{eq:affinity-freq-k-prac} and \eqref{eq:all-frequency-affinity-prac} which do not worsen the relative error. The logarithm of the combined affinity $A^\text{All}$ has the following relation with the viewing angles,
\begin{equation}
    \label{eq:logAall}
    \log \left( A^{\text{All}}_{ij} \right) = \sum_{k = 1}^{k_\mathrm{max}} \log \left( A^{(k)}_{ij} \right) \approx \frac{k_\mathrm{max}(k_\mathrm{max} + 1)}{2} \log \left(\frac{\left\langle
  \pi \left( x_i \right), \pi \left( x_j \right) \right\rangle+1}{2} \right).
\end{equation}
Using $A^\text{All}$ or $\log \left(A^\text{All} \right)$ makes small viewing angles much more prominent in the numerical procedures. One may well expect other linear combinations of the $A_{ij}^{\left( k \right)}$'s, which are degree-$k_{\mathrm{max}}$ polynomials of the (cosine value of the) viewing angle. We leave these further explorations to future work.

%% file: noise_model.tex
\section{Analysis under Probabilistic Models}
\label{sec:noise_model}
In this section, we discuss the benefit of using $A^{(k)}$ with $k > 1$ to identify nearest neighbors when the measurement graph is perturbed by noise. 
To this end, we use the random rewiring model~\cite{singer2011viewing} for the entries of $H^{(1)}$ in Section~\ref{sec:rrm_analysis} and extend it to incorporate small angular perturbation in Section~\ref{sec:angle_perturb}. 
We start by randomly generating $N$ orthonormal frames $x_1, x_2, \dots, x_N$ uniformly sampled from $\SO(3)$ according to the Haar measure. 
Each frame $x_i$ can be represented by a $ 3\times 3$ orthogonal matrix $R_i = [ R_i^1, R_i^2, R_i^3 ]$ and $\det(R_i) = 1$. 
We identify the third column $R_i^3$ as the \emph{viewing angle} $\pi(x_i)$ of the molecule. The first two columns $R^1_i$ and $R^2_j$ form an orthonormal basis for the plane in $\mathbb{R}^3$ perpendicular to the viewing angle $\pi(x_i)$. If the viewing angles for two projection images belong to a small spherical cap with opening angle $\alpha$, then we connect the two points in the graph (i.e.\ $(i, j) \in E$ if $\langle \pi(x_i), \pi(x_j) \rangle > \cos \alpha$). If $x_i$ and $x_j$ are two frames with the same viewing angle, $\pi(x_i) = \pi(x_j)$, then $R_i^1, R_i^2$ and $R_j^1, R_j^2$ are two orthogonal bases for the same plane and the rotation matrix $R_i^{-1}R_j$ has the following form:
\begin{equation}
    \label{eq:align1}
    R_i^{-1}R_j = \left( \begin{array}{ccc}
         \cos \theta_{ij} & -\sin \theta_{ij} & 0 \\
         \sin \theta_{ij} & \cos \theta_{ij} & 0 \\
         0 & 0 & 1 
    \end{array}\right ).
\end{equation}
When the viewing angles are slightly different,~\eqref{eq:align1} holds approximately. The optimal in-plane rotational angle $\theta_{ij}$ provides a good approximation to the angle $\theta_{ij}$ that ``aligns'' the orthonormal bases for the planes $\pi(x_i)^\perp$ and $\pi(x_j)^\perp$. Therefore, if $\langle \pi(x_i), \pi(x_j) \rangle$ is close to 1, the angle $\theta_{ij}$ is given by 
\begin{equation}
    \label{eq:align2}
    \theta_{ij} = \argmin_{\theta \in [0, 2\pi)} \| R_i \rho(\theta) - R_j \|_{\mathrm{F}}, \quad \text{with } \rho(\theta) = \left( \begin{array}{ccc}
         \cos \theta & -\sin \theta & 0 \\
         \sin \theta & \cos \theta & 0 \\
         0 & 0 & 1 
    \end{array}\right ).
\end{equation}
In other words, the ground truth local parallel transport data is computed by aligning the local frames within the connected neighborhood, determined by the entries of the matrix $R_i^{-1}R_j$: 
\begin{align}
\cos \theta_{ij} & = \frac{ \left(R_i^{-1}R_j \right)_{11} + \left(R_i^{-1} R_j\right)_{22} }{\sqrt{\left[ \left(R_i^{-1}R_j \right)_{11} + \left(R_i^{-1} R_j\right)_{22}\right]^2 + \left[ \left(R_i^{-1}R_j \right)_{21} - \left(R_i^{-1} R_j\right)_{12} \right]^2  }}, \nonumber \\
\sin \theta_{ij} & = \frac{ \left(R_i^{-1}R_j \right)_{21} - \left(R_i^{-1} R_j\right)_{12}}{\sqrt{\left[ \left(R_i^{-1}R_j \right)_{11} + \left(R_i^{-1} R_j\right)_{22}\right]^2 + \left[ \left(R_i^{-1}R_j \right)_{21} - \left(R_i^{-1} R_j\right)_{12} \right]^2 }}. 
\label{eq:align}
\end{align}

\subsection{Random Rewiring Model}
\label{sec:rrm_analysis}
Starting from the clean neighborhood graph constructed above, we perturb the graph based on the following process: with probability $p$, we keep the clean edge and the associated transport data $\theta_{ij}$; and with probability $1-p$, we remove the edge $(i, j)$ and randomly rewire $i$ or $j$ with a vertex drawn uniformly at random from the remaining vertices that are not already connected to $i$ or $j$. We assume that if the link between $i$ and $j$ is a random link, then $\theta_{ij} = \phi_{ij}$, which is uniformly distributed over $[0, 2\pi)$. Our model assumes that the underlying graph of links between noisy data points is a small-world graph~\cite{watts1998collective} on the sphere, with edges being randomly rewired with probability $1-p$. The alignments take their correct values for true links and random values for the rewired edges. The parameter $p$ controls the signal to noise ratio of the graph connection where $p = 1$ indicates the clean graph.   

The matrix $H^{(k)}$ is a random matrix under this model with 
\begin{equation}
    \label{eq:Hk}
    H_{ij}^{(k)} = \begin{cases}
    e^{\imath k \theta_{ij}}, & \text{if } (i, j) \in E \text{ and with probability } p,  \\
    e^{\imath k \phi_{ij}} , & \text{if } (i, j) \notin E \text{ and with probability } \frac{(1-p)\bar{D}}{N - \bar{D} }. 
      \end{cases}
\end{equation}
Since the expected value of the random variable $e^{\imath k \phi}$ vanishes for $\phi \sim \mathrm{Uniform}[0, 2\pi)$, the expected value of the matrix $H^{(k)}$ is 
\begin{equation}
\mathbb{E} H^{(k)} = p H^{(k)}_{\text{clean}},
\label{eq:Hk_exp}
\end{equation}
where $H^{(k)}_{\text{clean}}$ is the clean matrix that corresponds to $p = 1$ obtained in the case that all links and angles are set up correctly. At each frequency $k$, the matrix $H^{(k)}$ can be decomposed into 
\begin{equation}
H^{(k)} = p H^{(k)}_{\text{clean}} + R^{(k)},  
\end{equation}
where $R^{(k)}$ is a random matrix whose elements are 
\begin{equation}
    \label{eq:Rk}
    R_{ij}^{(k)} = \begin{cases}
    (1 - p) e^{\imath k \theta_{ij}}, & \text{if } (i, j) \in E \text{ and with probability } p,  \\
      - p e^{\imath k \theta_{ij}}, & \text{if } (i, j) \in E  \text{ and with probability } 1-p, \\
      e^{\imath k \phi_{ij}} , & \text{if } (i, j) \notin E \text{ and with probability } \frac{(1-p)\bar{D}}{N - \bar{D} }, 
      \end{cases}
\end{equation}
and $\bar{D}$ is the average degree of the clean neighborhood graph. The elements in $R^{(k)}$ are independent zero mean random variables with finite moments, since the elements of $R^{(k)}$ are bounded for $1 \leq k \leq k_\text{max}$. 

We use $\| M\|$ to denote the spectral norm of a matrix $M$. Since the underlying graph connectivity for all $R^{(k)}$ is identical and the mean and variance of $R_{ij}^{(k)}$ are identical across $k$, the quantity $\| R^{(k)}\|$ does not change over frequency index $k$. To find an upper bound on $\| R^{(k)} \|$, we take $p = 0$, where the matrix $R^{(k)}$ represents a sparse random graph. Since the surface area of a spherical cap with opening angle $\alpha$ is $ 4\pi \sin^2 \frac{\alpha}{2}$ and $N$ points are uniformly distributed over the sphere, the average degree of the random graph is $N \sin^2 \frac{\alpha}{2}$.    
Adapting~\cite[Theorem 2.1]{khorunzhy2001sparse} to our case, we can show that $\| R^{(k)} \| \leq 2 \sqrt{N} \sin \frac{\alpha}{2} $ with high probability. In Figure~\ref{fig:hist_spec}, we can see that the eigenvalues of $R^{(k)}$ follows Wigner's semicircle law~\cite{wigner1955,wigner1958}.

For the following discussion, we denote $ k_\mathrm{max} = \left \lfloor \frac{1}{h} -1 \right \rfloor$. The ordered eigenvalues for $pH^{(k)}_\text{clean}$ are $\ell^{(k)}_1 \geq \ell^{(k)}_2 \geq \dots \geq \ell^{(k)}_N$, and the ordered eigenvalues for $H^{(k)}$ are $ \tilde{\ell}^{(k)}_1 \geq \tilde{\ell}^{(k)}_2 \geq \dots \geq \tilde{\ell}^{(k)}_{2k+2}, \dots \geq \tilde{\ell}^{(k)}_N$, for $k = 1, \dots, k_\text{max}$. The spectral gap after the $(2k+1)^{\text{th}}$ eigenvalue for $pH_\text{clean}^{(k)}$ is denoted as $\delta_k = \ell^{(k)}_{2k+1} - \ell^{(k)}_{2k+2}$. We note that $ \{\ell_i^{(k)} \}_{i = 1}^{2k+1} \approx pN \lambda_k^{(k)}$ and $ \{\ell_i^{(k)} \}_{i = 2k+2}^{4k+3} \approx pN \lambda_{k+2}^{(k)}$, and $\delta_k \approx pN G^{(k)}$, since $\frac{1}{N} H^{(k)}_{\text{clean}}$ is a discretization of the operator $T_h^{(k)}$. We consider the following three scenarios for the discussion of the stability of the algorithm under noise perturbation: (1) small noise regime ($  \delta_1 \geq 2\| R^{(1)} \| $), (2) medium noise regime ($ \delta_1 < 2\| R^{(1)} \| \leq \delta_{k_\text{max}} $), and  (3) large noise regime ($ \delta_{k_\text{max}} < 2\| R^{(1)} \| $). 

\noindent \textit{$\bullet$ Small noise regime.} This noise regime was previously considered in~\cite{singer2011viewing} to determine the threshold probability $p_c$ for the approximation of the top three eigenvectors of $H^{(1)}$ and the top three eigenvectors of $H^{(1)}_{\text{clean}}$ under the random rewiring model. 
According to Corollary~\ref{cor:2}, the spectral gap gets larger
for higher frequency index $k$.
This implies that the linear space spanned by the first $\left( 2k+1 \right)$ eigenvectors of $H^{\left( k \right)}$ is closer to the top eigenspace of $T_h^{\left( k \right)}$, since the approximation error is inversely proportional to the spectral gap according to the renowned Davis--Kahan theorem~\cite{DK1970,YWS2014}. This also explains the choice of extracting the top $\left( 2k+1 \right)$ eigenvectors of $H^{\left( k \right)}$ in single frequency-$k$ class averaging.

\noindent \textit{$\bullet$ Medium noise regime.} In this situation, we can find a $\tilde{k}$ such that for all $ \tilde{k} \leq k  \leq k_\text{max}$, $ \delta_k > 2 \| R^{(1)} \| = 2\| R^{(k)}\|$. In addition, we can show that $ \ell^{(k)}_{4k+3} \!>\! \| R^{(k)}\|$ for $k = 1, \dots, k_\text{max}$. This is because we have $\lambda_{k+1}^{(k)} > G^{(k)}$ for $k \leq k_\text{max}$ and $\lambda_{k+1}^{(k)}$ decreases as $k$ increases according to Theorem~\ref{thm:eigenvalue-asym} and Theorem~\ref{thm:specgap}. Using the same argument as in the small noise regime, we can justify the benefit of using the top $(2k+1)$ eigenvectors of $H^{(k)}$ at $k > 1$.

\noindent \textit{$\bullet$ Large noise regime. } If we further decrease $p$, the spectral norm of $ R^{(k)}$ becomes larger than the spectral gap $\delta_{k_\text{max}}$. According to Davis-Kahan theorem, it seems impossible to recover the eigenvectors if the eigenvalue perturbation is too large. However, we observe that under this situation, the subspace spanned by the top $2k+1$ eigenvectors of $H^{(k)}$ still has  non-trivial correlation with the subspace spanned by the top $2k+1$ eigenvectors of $H^{(k)}_\text{clean}$, if $ \tilde{\ell}_{2k+1} > \| R_k^{(k)}\|$ and in other words $ \ell_{2k+1} > \frac{1}{2}\| R_k^{(k)} \|$. This phenomenon is similar to the phase transition for eigenvalues and eigenvectors of a low rank matrix under the additive perturbation of a Gaussian  Wigner matrix in~\cite[Section 3.1]{benaych2011eigenvalues}, although our underlying clean matrices $H^{(k)}_{\text{clean}}$ are full rank.
It seems that the eigenvectors of the unperturbed matrix are possible to recover even when the spectral gap is much smaller than that required by Davis-Kahan. In this case, the Davis-Kahan theorem is insufficient to bound the distance between the subspaces since it does not consider the nature of the perturbation. It is useful to use perturbation bounds that take into account the nature of the perturbation such as the upper bound on the entry-wise deviation of the eigenvector in~\cite[Theorem 8]{EBW2017}. The theorem only applies to the situation with $\ell_{2k+1} > \| R^{(k)} \|$. According to the theorem, both $\delta_k$ and $\ell_{2k+1} - \| R^{(k)} \|$ appear in the denominators of the terms in the upper bound for the entry-wise deviation of the eigenvector. As $k$ increases, the spectral gap $\delta_k$ increases, while the term $\ell_{2k+1} - \| R^{(k)} \|$ decreases based on Corollaries~\ref{cor:1},~\ref{cor:2}, and Theorem~\ref{thm:specgap}. This implies that the upper bound of the deviation in~\cite[Theorem 8]{EBW2017} will decrease initially as $k$ increases from 1 because the reduction in the term that contains $\delta_k$ dominates, and then it will increase when the increments in the terms containing $\ell_{2k+1} - \| R^{(k)} \|$ becomes dominant. 
We empirically observe that the accuracy of the affinity measure $A^{(k)}$ increases with increasing $k$ up to a critical cutoff $k_c$ as detailed in Section~\ref{sec:rrm}. We identify $k_c$ as the point when $ \ell^{k_c}_{2k_c+2} \geq \frac{1}{2} \| R^{(k_c)} \|$ and $ \ell^{(k_c+1)}_{2k_c+4} < \frac{1}{2} \| R^{(k_c)} \|$, which corresponds to when $\tilde{\ell}_{2k+2}^{(k)}$ becomes very close to $ \| R^{(k)} \|$. The estimation of the top eigenspace gets less accurate when $k$ increases beyond $k_c$, which will result in worse classification results using $A^{(k)}$ (see Figure~\ref{fig:rrm_vark}). The eigenvector perturbation of a full-rank matrix with additive random matrix is still an open problem and we will provide theoretical justification for our observations in the future. 

Based on the discussions above, we see the benefit of using $A^{(k)}$ for $k > 1$ to select nearest neighbors because the underlying embedding $\Psi^{(k)}$ can be more stable than $\Psi^{(1)}$. Under additive noise perturbation in~\eqref{eq:Hk}, each embedding $\Psi^{(k)}$ is perturbed randomly, but they have non-trivial correlation with the corresponding true eigenspace when the noise is not too large. In addition, $\Psi^{(k)}(i)$ encodes the viewing direction information in terms of the degree-$k$ polynomial of the frame $x_i$ and the underlying information on $x_i$ is perturbed differently at different $k$ even though the noise is not independent. The combined score is able to identify pairs that have consistently high affinities across $k$ and filter out pairs that only have a couple of high scores across $k$.

\subsection{Random Rewiring Model with Small Angular Errors}
\label{sec:angle_perturb}
We extend the random rewiring model in Section~\ref{sec:rrm_analysis} to incorporate small angular errors in the pairwise alignment angles for the correctly connected pairs. Specifically, we consider additive errors in the angle,
\begin{equation}
    \label{eq:true_angle_perturb}
    \tilde{\theta}_{ij} = \theta_{ij} + \varepsilon_{ij},
\end{equation}
where $\varepsilon_{ij}$ are independently drawn from a distribution $\gamma$ on the interval $[0, 2\pi)$. We also assume that $\mathbb{E}(\varepsilon_{ij}) = 0\; \text{mod} \; 2\pi$. We can evaluate $c_k = \mathbb{E}(e^{\imath k \varepsilon})$ for $\varepsilon\sim \gamma([0, 2\pi))$. The matrix $H^{(k)}$ is a random matrix under this model with
\begin{equation}
    \label{eq:Hk2}
    H_{ij}^{(k)} = \begin{cases}
    e^{\imath k \left(\theta_{ij} + \varepsilon_{ij}\right) }, & \text{if } (i, j) \in E \text{ and with probability } p,  \\
    e^{\imath k \phi_{ij}} , & \text{if } (i, j) \notin E \text{ and with probability } \frac{(1-p)\bar{D}}{N - \bar{D} }. 
      \end{cases}
\end{equation}
Since the expected value of the random variable $e^{\imath k \phi}$ vanishes for $\phi \sim \mathrm{Uniform}[0, 2\pi)$, the expected value of the matrix $H^{(k)}$ is 
\begin{equation}
\label{eq:Hk2_exp}
\mathbb{E} H^{(k)} = c_k p H^{(k)}_{\text{clean}},
\end{equation}
where $H^{(k)}_{\text{clean}}$ is the clean matrix that corresponds to $p = 1$ obtained in the case that all links and angles are set up correctly. At each frequency $k$, the matrix $H^{(k)}$ can be decomposed into 
\begin{equation}
H^{(k)} = c_k p H^{(k)}_{\text{clean}} + R^{(k)},  
\end{equation}
where $R^{(k)}$ is a random matrix whose elements are 
\begin{equation}
    \label{eq:Rk2}
    R_{ij}^{(k)} = \begin{cases}
    (e^{\imath k \varepsilon_{ij}} - c_k p) e^{\imath k \theta_{ij}}, & \text{if } (i, j) \in E \text{ and with probability } p,  \\
      - c_k p e^{\imath k \theta_{ij}}, & \text{if } (i, j) \in E  \text{ and with probability } 1-p, \\
      e^{\imath k \phi_{ij}} , & \text{if } (i, j) \notin E \text{ and with probability } \frac{(1-p)\bar{D}}{N - \bar{D} }. 
      \end{cases}
\end{equation}
The analysis follows the steps in Section~\ref{sec:rrm_analysis} and $\| R^{(k)} \| \leq 2 \sqrt{N} \sin \frac{\alpha}{2}$ with high probability. 
Comparing Eq.~\eqref{eq:Hk2_exp} with Eq.~\eqref{eq:Hk_exp}, we find that the main difference is that the eigenvalues of $\mathbb{E} H^{(k)}$ are scaled by $c_k$ at frequency $k$. The condition for the spectral algorithm to work is that the spectral gap $c_k pNG^{(k)}$ and the top eigenvalue $c_k p N \lambda^{(k)}_k$ are sufficiently large compared with $\| R^{(k)} \|$. For a well concentrated distribution $\gamma$, we can first evaluate $c_k$ and then determine the critical cutoff frequency $k_c$ that satisfy the condition. With the same $p$ in the random rewiring model, $k_c$ gets smaller in the presence of additional small angular errors since $c_k < 1$. In Section~\ref{sec:rrm}, we show the performance of the algorithms on a couple of examples where the angular noise follows a von Mises distribution.  

\subsection{Discussions}
In the previous two models, we only consider independent edge noise, i.e., the entries in $R^{(k)}$ for a fixed $k$ are independent. Across different frequencies, the entries $R_k$ are dependent through the relations of the irreducible representations of the angles ($\theta_{ij}$, $\phi_{ij}$, and $\varepsilon_{ij}$) and the graph connectivity. We note that these are simplified models for illustrating the benefits of using $A^{(k)}$ for $k>1$. In the application to cryo-EM 2-D image analysis, the edge perturbations are induced by the independent noise from each image. In this case, for fixed frequency, the entries in $R^{(k)}$ becomes dependent since the edge connections and alignments are affected by the noise in each image node. Still we observe similar benefits of using $A^{(k)}$ for $k > 1$ with the cryo-EM class averaging experiments detailed in Section~\ref{sec:exp_cryo}. We leave the analysis of node level noise to future work.  In addition, the current analysis focuses on data points that are uniformly distributed on the manifold. For non-uniformly distributed data points, different normalization techniques introduced in diffusion maps~\cite{coifman2006diffusion} are needed to compensate for the non-uniform sampling density.

\begin{figure}
	\captionsetup[subfigure]{oneside,margin={0.03cm,0cm}}
    \centering
    \subfloat[$H^{(1)}$, $p = 0.5$]{
    \includegraphics[width = 0.23 \textwidth ]{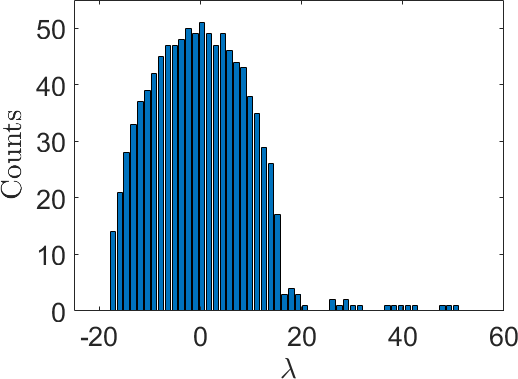}
    \label{fig:hist_spec_k1}
    }
    \subfloat[$H^{(4)}$, $p = 0.5$]{
    \includegraphics[width = 0.23 \textwidth ]{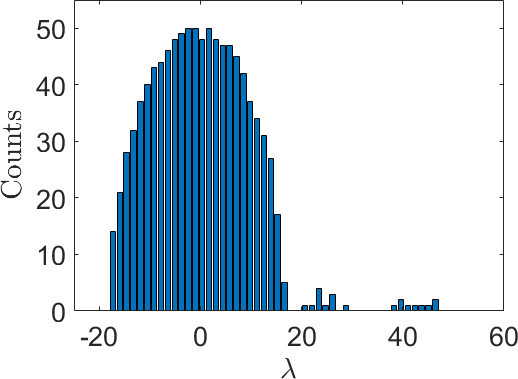}
    \label{fig:hist_spec_k4}
    }
    \subfloat[$H^{(8)}$, $p = 0.5$]{
    \includegraphics[width = 0.23 \textwidth ]{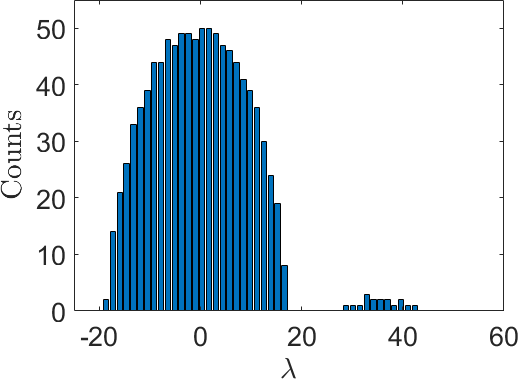}
    \label{fig:hist_spec_k8}
    } 
    \subfloat[$H^{(20)}$, $p = 0.5$]{
    \includegraphics[width = 0.23 \textwidth ]{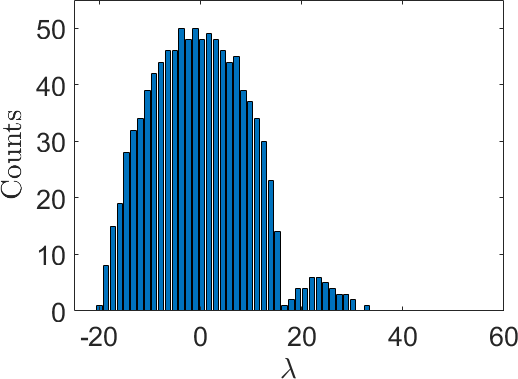}
    \label{fig:hist_spec_k20}
    } \\
    \subfloat[$H^{(1)}$, $p = 0.3$]{
    \includegraphics[width = 0.23 \textwidth ]{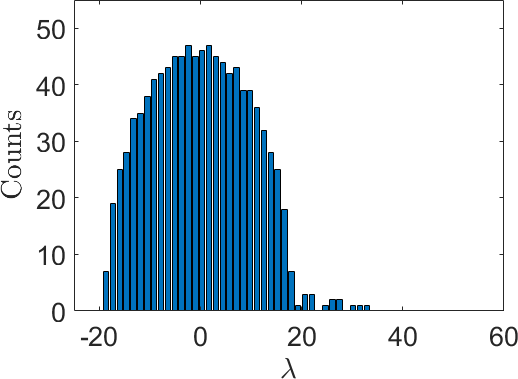}
    \label{fig:hist_spec_k1_p3}
    }
    \subfloat[$H^{(4)}$, $p = 0.3$]{
    \includegraphics[width = 0.23 \textwidth ]{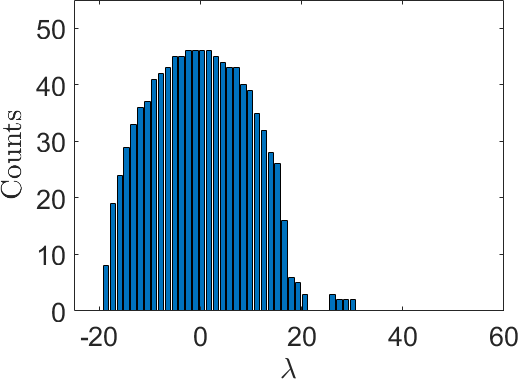}
    \label{fig:hist_spec_k4_p3}
    }
    \subfloat[$H^{(8)}$, $p = 0.3$]{
    \includegraphics[width = 0.23 \textwidth ]{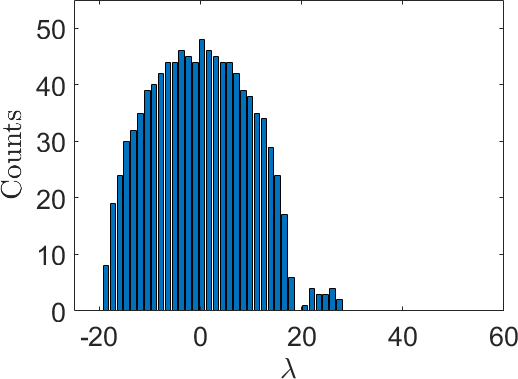}
    \label{fig:hist_spec_k8_p3}
    } 
    \subfloat[$H^{(20)}$, $p = 0.3$]{
    \includegraphics[width = 0.23 \textwidth ]{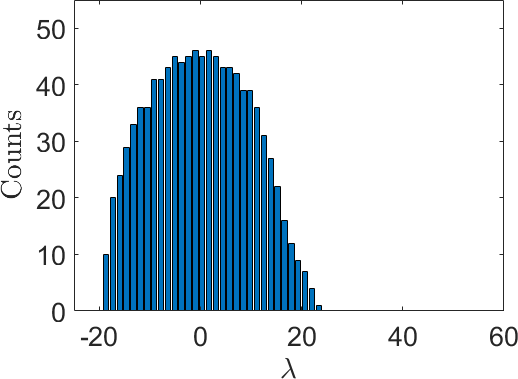}
    \label{fig:hist_spec_k20_p3}
    } \\
    \subfloat[$R^{(1)}$, $p = 0.5$]{
    \includegraphics[width = 0.23 \textwidth ]{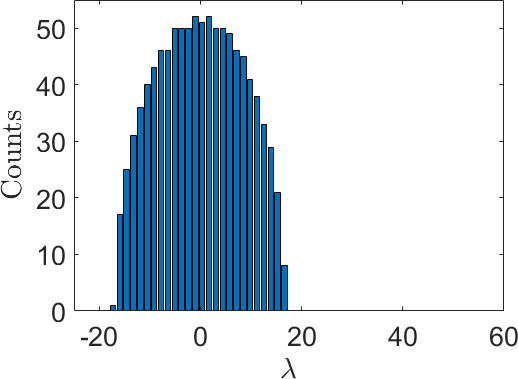}
    \label{fig:hist_spec_Rk1_p5}
    }
    \subfloat[$R^{(20)}$, $p = 0.5$]{
    \includegraphics[width = 0.23 \textwidth ]{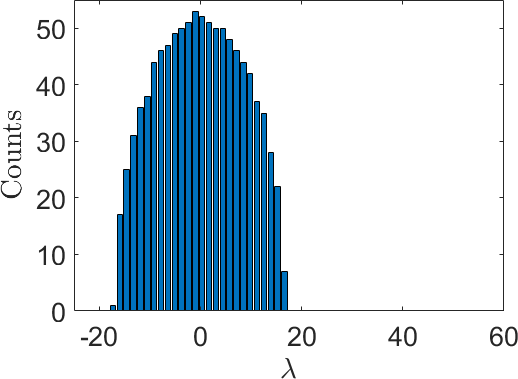}
    \label{fig:hist_spec_Rk20_p5}
    }
    \subfloat[$R^{(1)}$, $p = 0.3$]{
    \includegraphics[width = 0.23 \textwidth ]{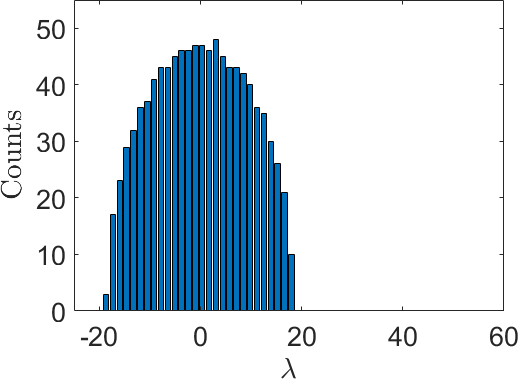}
    \label{fig:hist_spec_Rk1_p3}
    }
    \subfloat[$R^{(20)}$, $p = 0.3$]{
    \includegraphics[width = 0.23 \textwidth ]{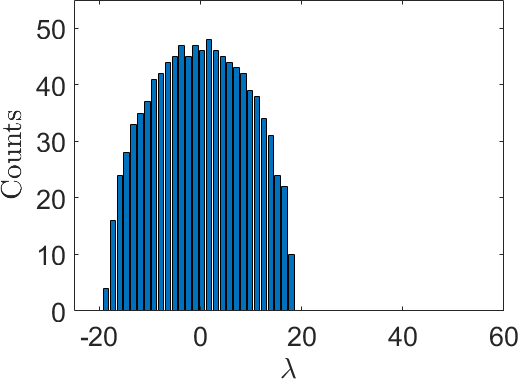}
    \label{fig:hist_spec_Rk20_p3}
    }
    \caption{\small Histograms of the eigenvalues of $H^{(k)}$ and $R^{(k)}$ in~\eqref{eq:Hk} for data generated from random rewiring model with $N = 1000$, $p = 0.5$, and $p = 0.3$. }
    \label{fig:hist_spec}
\end{figure}

%% file: numerical.tex
\begin{figure}
	\captionsetup[subfigure]{oneside,margin={0.03cm,0cm}}
    \begin{center}
    \subfloat[$p = 0.5$]{
    \includegraphics[width = 0.3\textwidth]{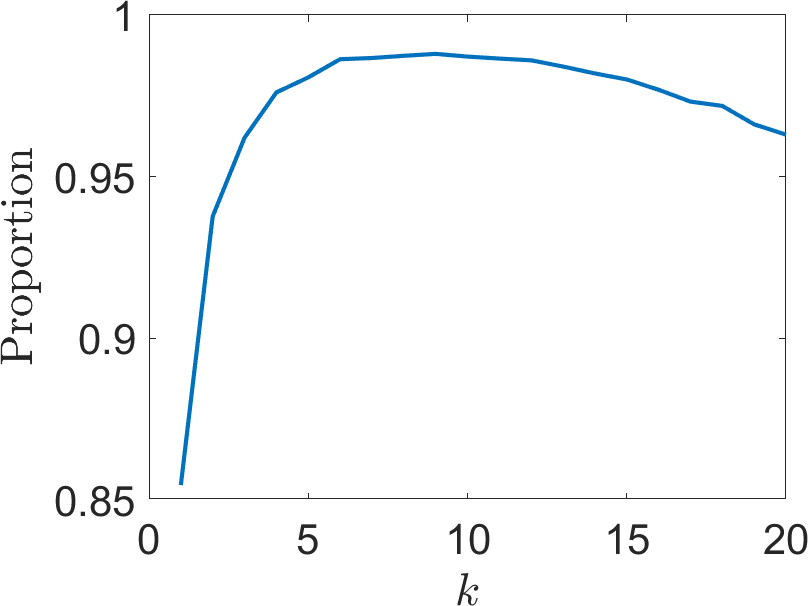}
    \label{fig:rrm_n1000_p5}
    }
    \subfloat[$p = 0.3$]{
    \includegraphics[width = 0.3\textwidth]{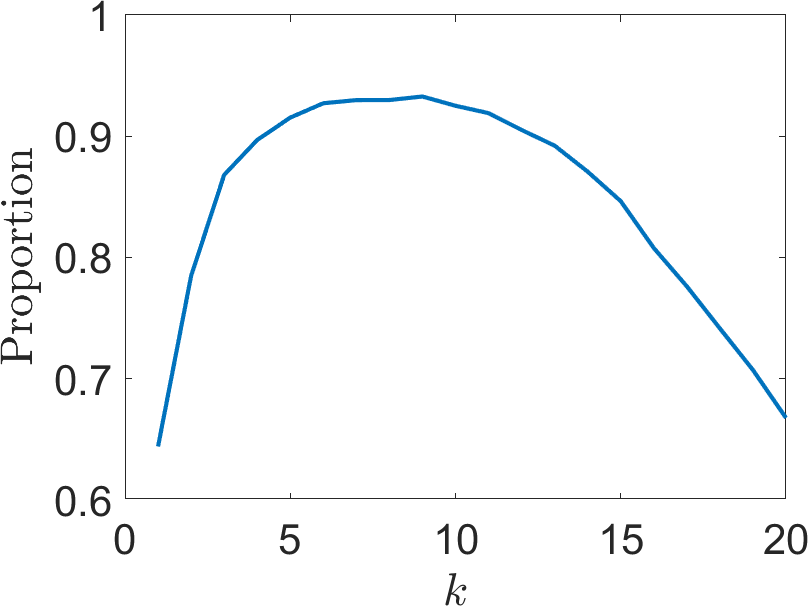} 
    \quad
    \label{fig:rrm_n1000_p3}
    }
    \subfloat[$p = 0.15$]{
    \includegraphics[width = 0.3\textwidth]{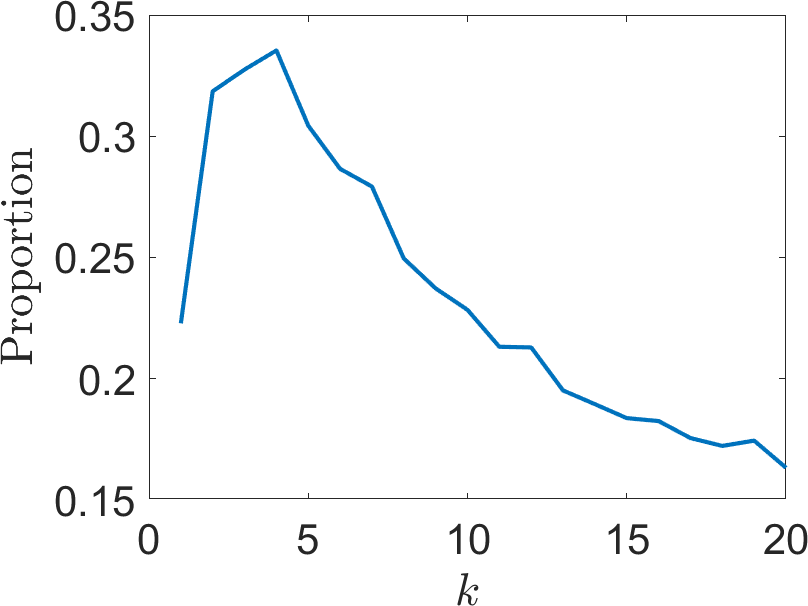} \quad
    \label{fig:rrm_n1000_p15}
    }
    \end{center}
    \caption{\small Proportion of the estimated nearest neighbors that satisfy $\langle \pi(x_i), \pi(x_j)\rangle \! >\! 0.85$ for $p\! =\! 0.5$, $0.3$, and $0.15$. The number of frames $N = 1000$ and the number of nearest neighbors is 50.}
    \label{fig:rrm_n1000_CA}
\end{figure}

\begin{figure}[t!]
\centering
\setlength\tabcolsep{1.0pt}
\renewcommand{\arraystretch}{0.6}
\begin{tabular}{p{0.05\textwidth}<{\centering} p{0.3\textwidth}<{\centering} p{0.3\textwidth}<{\centering} p{0.3\textwidth}<{\centering}}
& $\boldsymbol{k = 1}$ & $\boldsymbol{k = 3}$ & $\boldsymbol{k = 5}$ \\
\raisebox{1.4cm}{\rotatebox{90}{$\boldsymbol{p = 1}$}} &   \includegraphics[width = 0.3\textwidth]{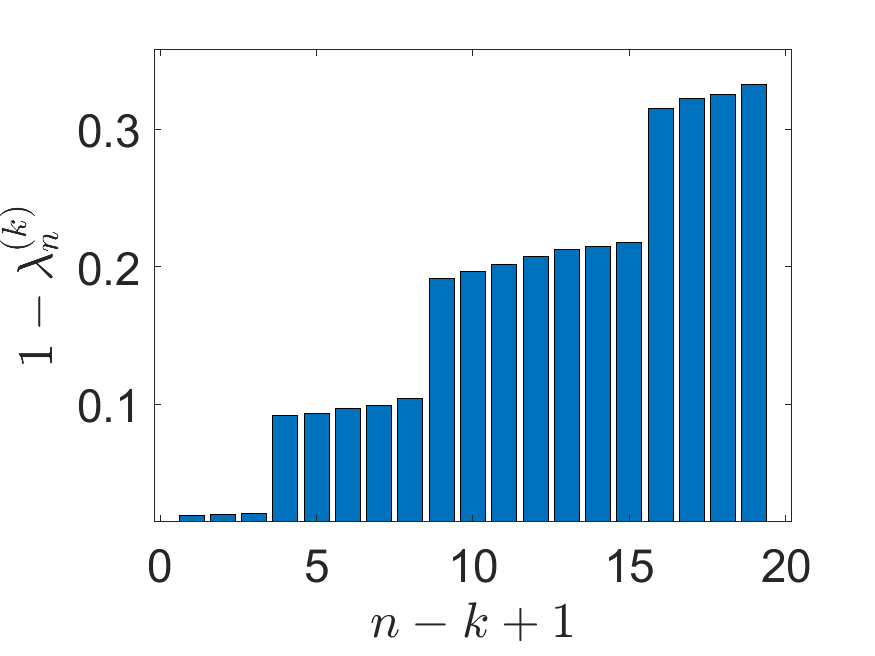} & \includegraphics[width = 0.3\textwidth]{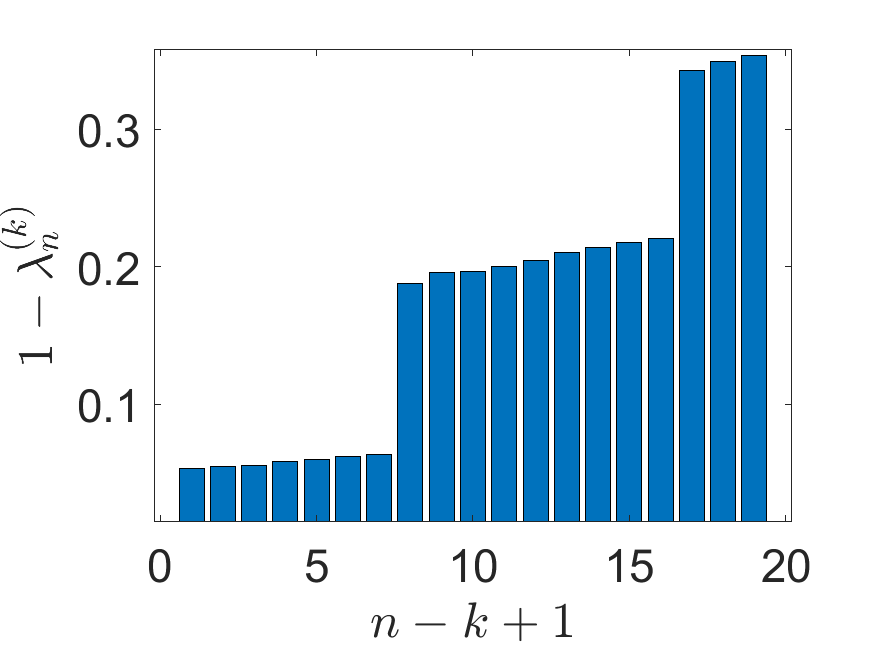}& \includegraphics[width = 0.3\textwidth]{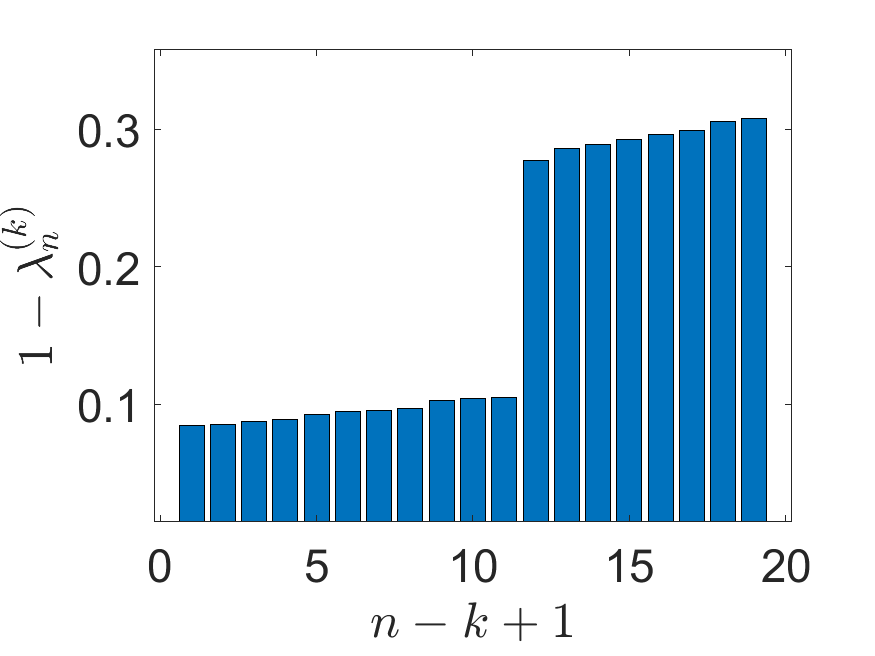} \\
\raisebox{1.3cm}{\rotatebox{90}{$\boldsymbol{p = 0.2}$}}  &   \includegraphics[width = 0.3\textwidth]{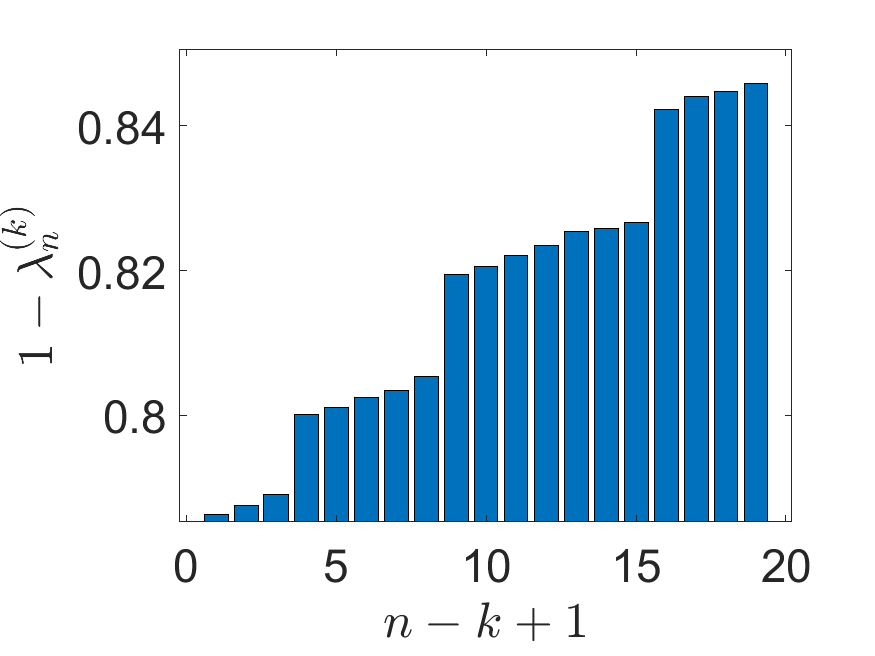} & \includegraphics[width = 0.3\textwidth]{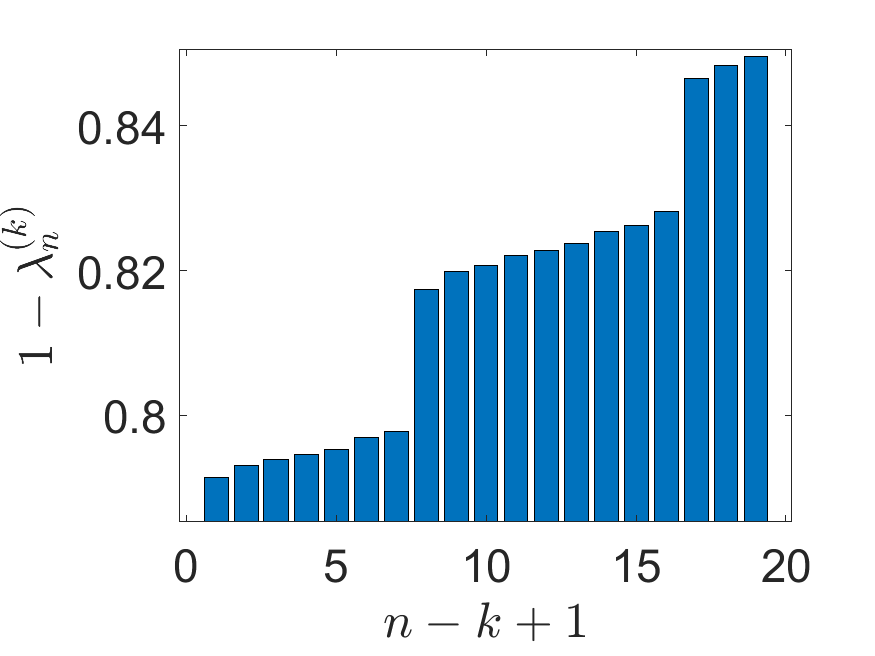}& \includegraphics[width = 0.3\textwidth]{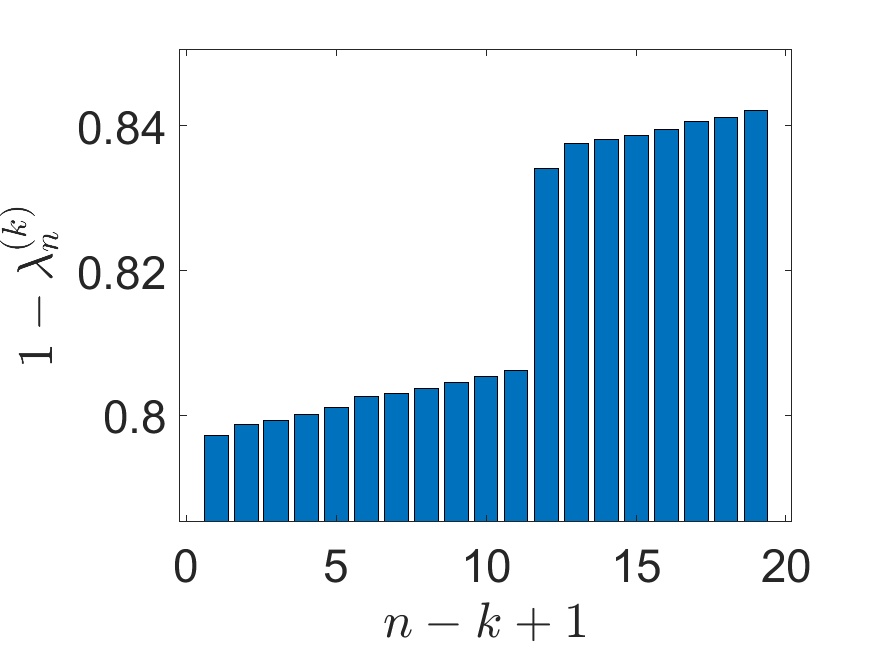} \\
\raisebox{1.3cm}{\rotatebox{90}{$\boldsymbol{p = 0.1}$}}  &   \includegraphics[width = 0.3\textwidth]{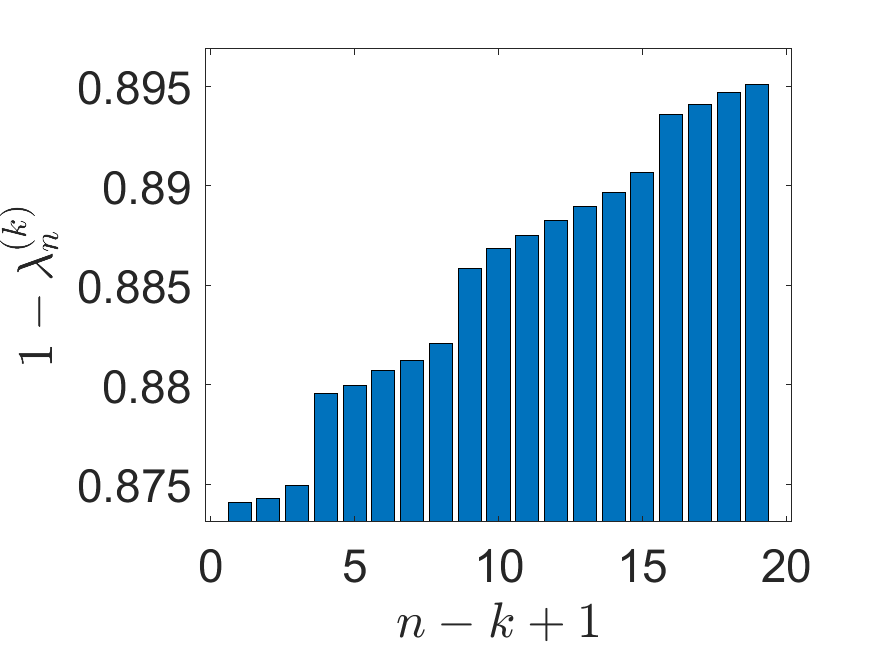} & \includegraphics[width = 0.3\textwidth]{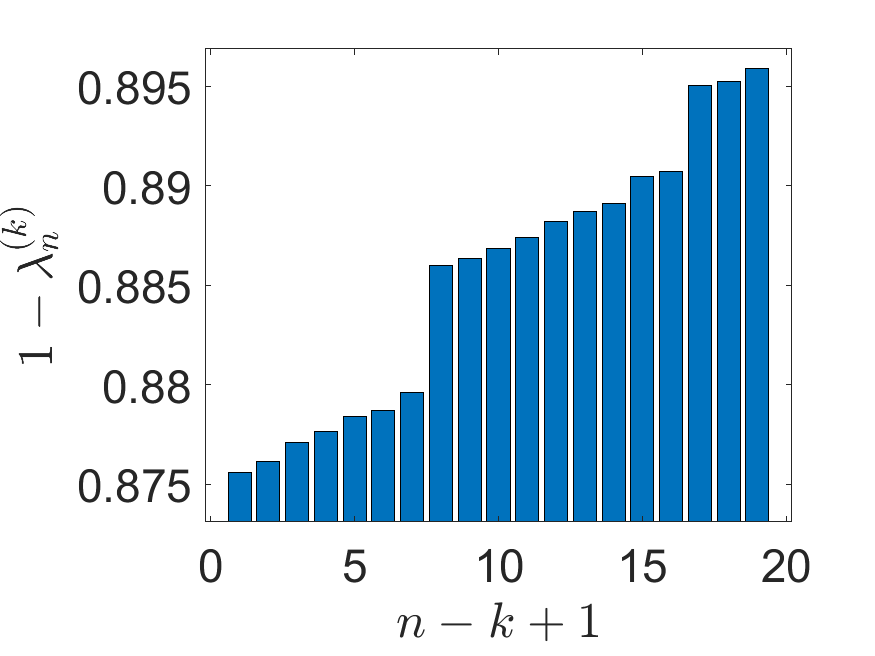}& \includegraphics[width = 0.3\textwidth]{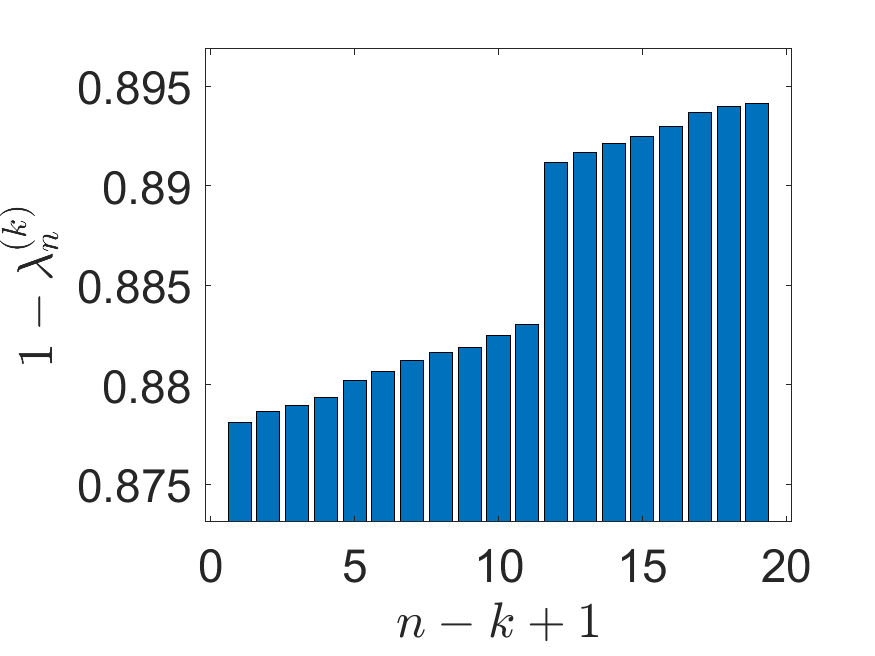} \\
\raisebox{1.2cm}{\rotatebox{90}{$\boldsymbol{p = 0.08}$}}  &   \includegraphics[width = 0.3\textwidth]{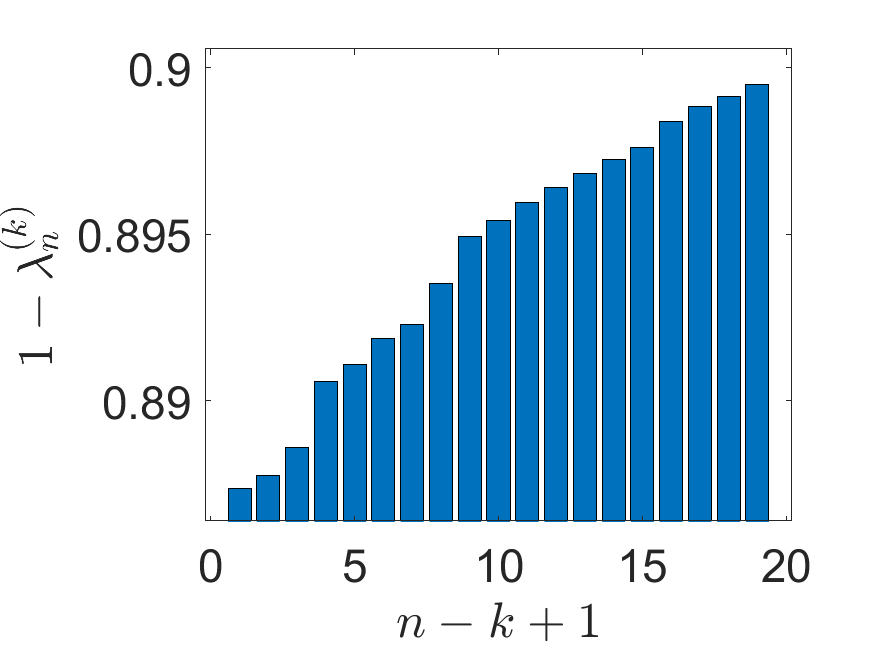} & \includegraphics[width = 0.3\textwidth]{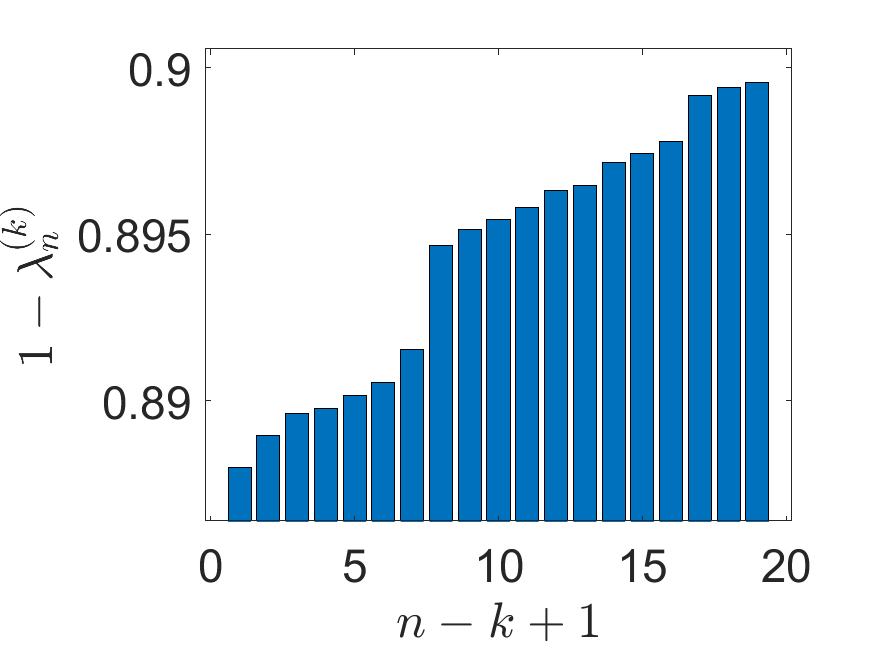}& \includegraphics[width = 0.3\textwidth]{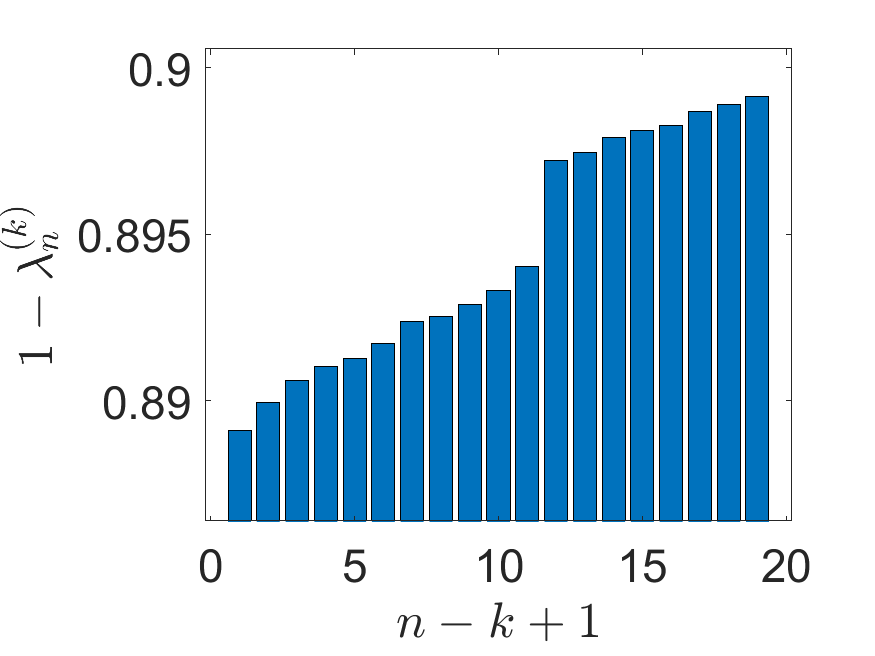} \\
\end{tabular}
    \caption{\small Bar plots of the 19 largest eigenvalues of the $\widetilde{H}^{(k)}$ at different $k$ and $p$ values.}
    \label{fig:prob_eval_bar}
\end{figure}

\section{Numerical Results}
\label{sec:numerics}

We conducted two sets of numerical experiments. The first
set involves simulations of the probabilistic model introduced in ~\cite{singer2011viewing}. The second set applies the proposed algorithm on the noisy simulated projection images of a 3-D volume of 70S ribosome. We point out that there is no direct way to compare the performance of classification algorithms on real microscope images, since their viewing directions are unknown. 
The only way to compare classification algorithms on real data is indirectly, by evaluating the resulting 3-D reconstructions. Here we conduct only numerical experiments from which conclusions can be drawn directly for 2-D images.  All experiments in this section were executed on a Linux machine with 16 Intel Xeon 2.5GHz cores and 512GB of RAM. 

\begin{figure}[t!]
\centering
\setlength\tabcolsep{1.0pt}
\renewcommand{\arraystretch}{0.6}
\begin{tabular}{p{0.05\textwidth}<{\centering} p{0.3\textwidth}<{\centering} p{0.3\textwidth}<{\centering} p{0.3\textwidth}<{\centering}}
& $\boldsymbol{k = 1}$ & $\boldsymbol{k = 5}$ & $\boldsymbol{k = 10}$ \\
\raisebox{1.4cm}{\rotatebox{90}{$\boldsymbol{p = 1}$}} &   \includegraphics[width = 0.3\textwidth]{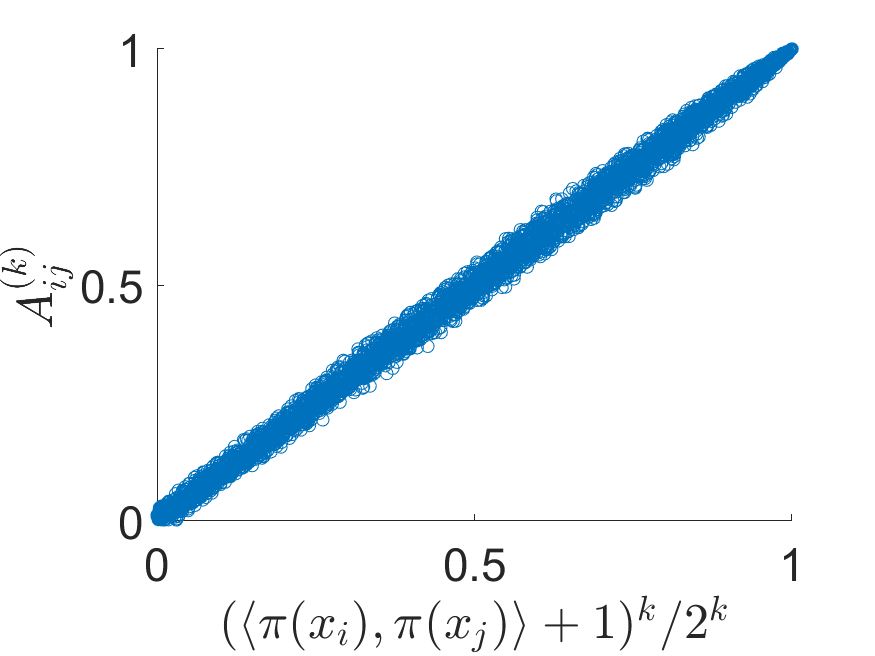} & \includegraphics[width = 0.3\textwidth]{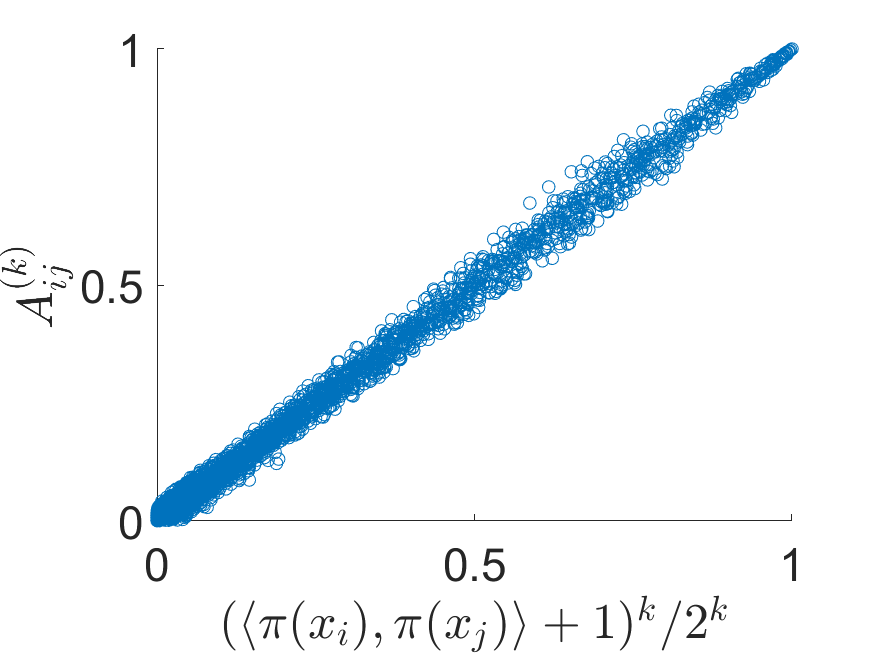}& \includegraphics[width = 0.3\textwidth]{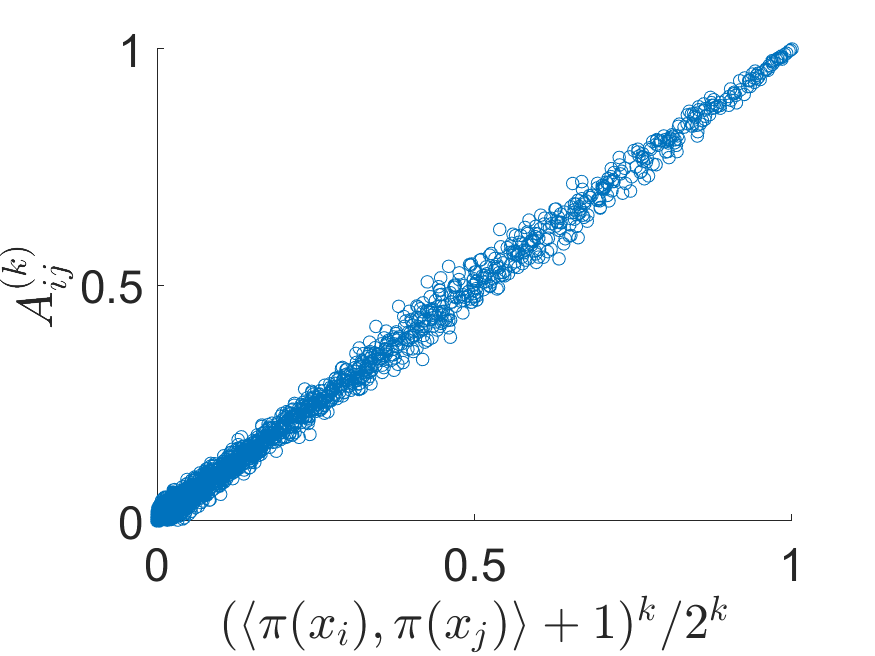} \\\hline
\raisebox{1.3cm}{\rotatebox{90}{$\boldsymbol{p = 0.2}$}}   &   \includegraphics[width = 0.3\textwidth]{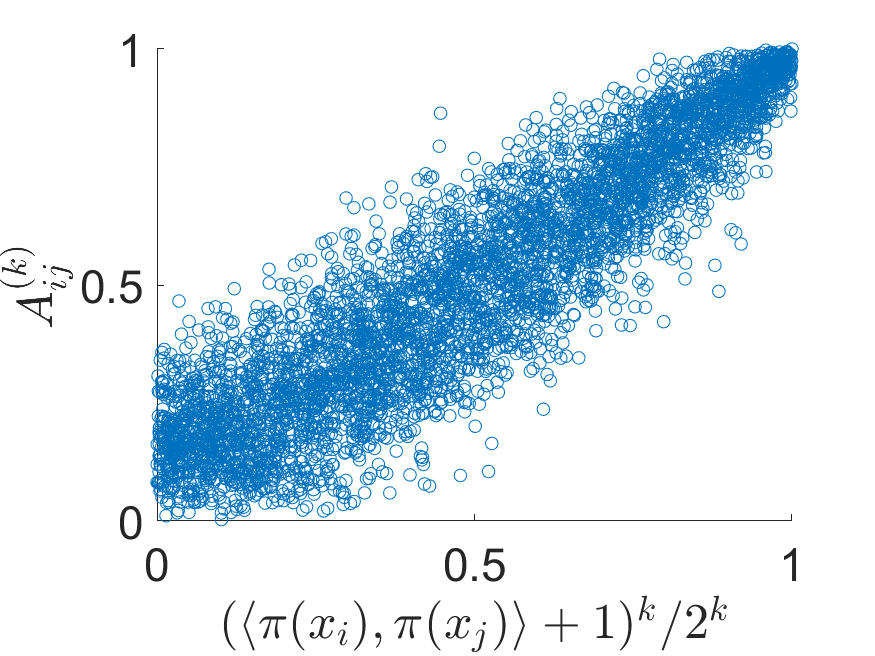} & \includegraphics[width = 0.3\textwidth]{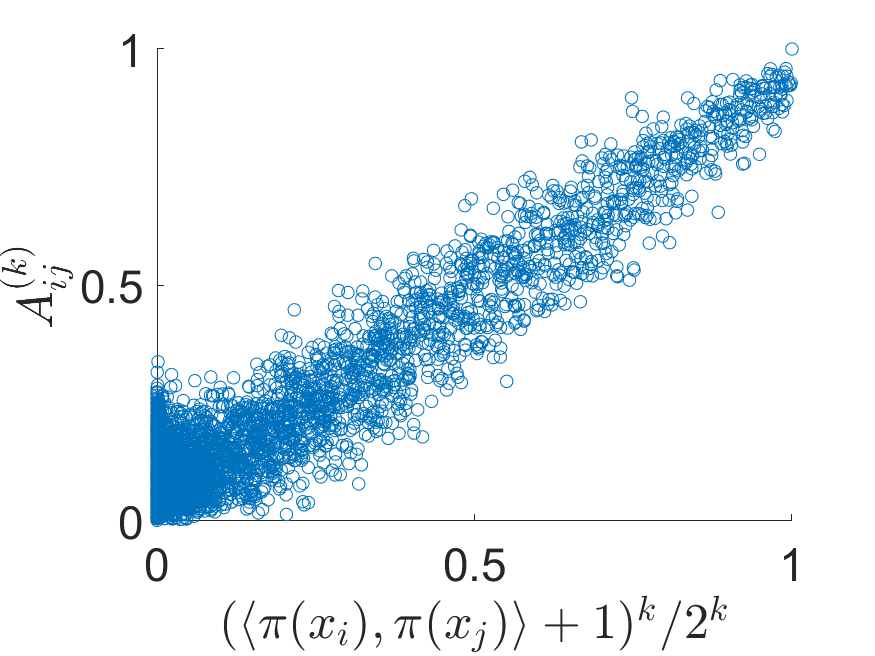}& \includegraphics[width = 0.3\textwidth]{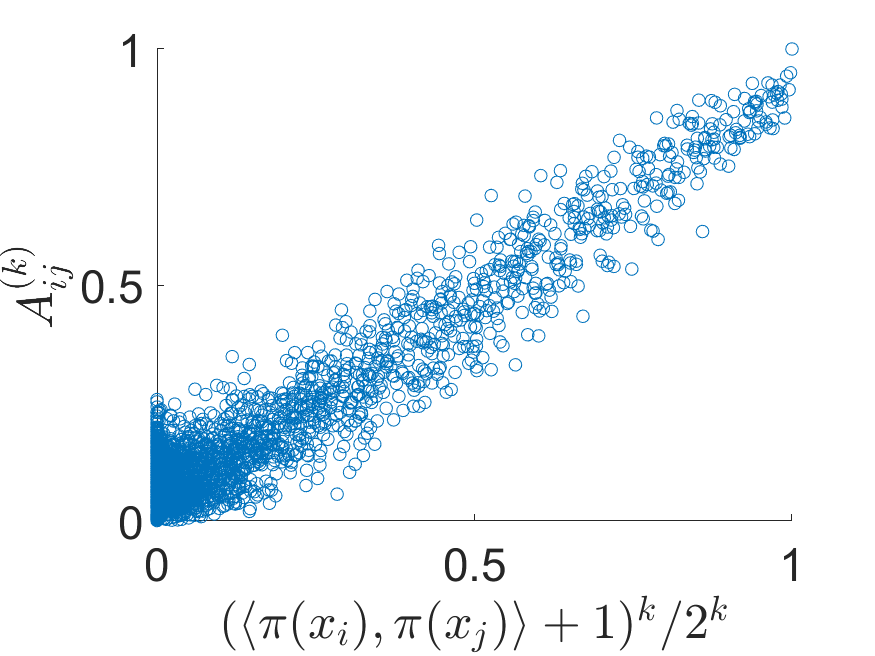} \\\hline
\raisebox{1.3cm}{\rotatebox{90}{$\boldsymbol{p = 0.1}$}}  &   \includegraphics[width = 0.3\textwidth]{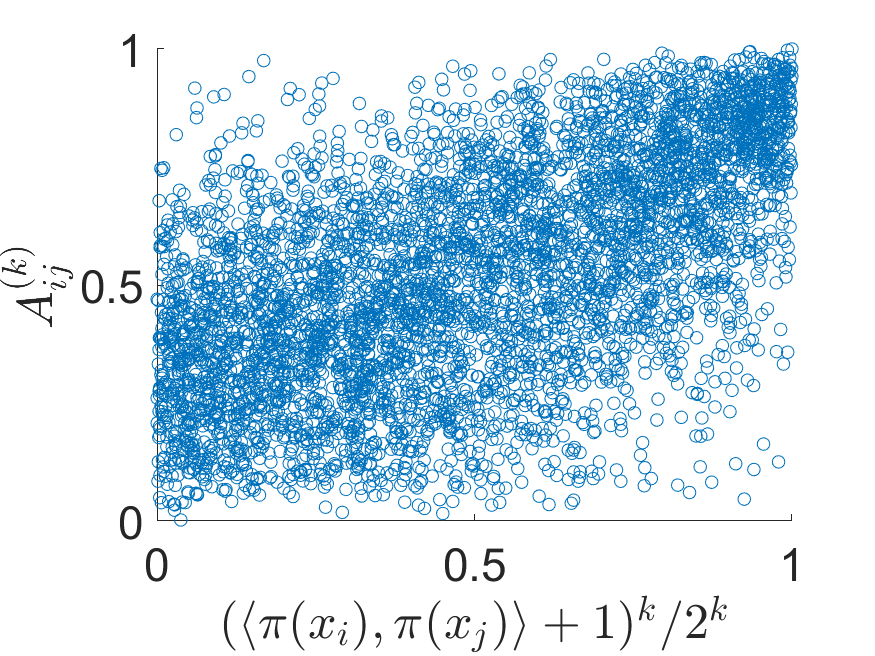} & \includegraphics[width = 0.3\textwidth]{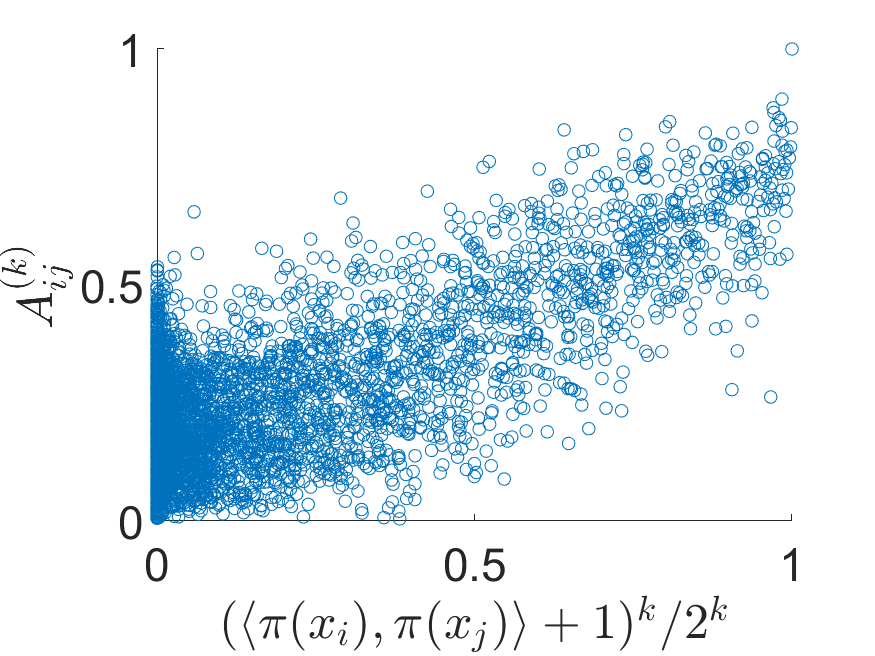}& \includegraphics[width = 0.3\textwidth]{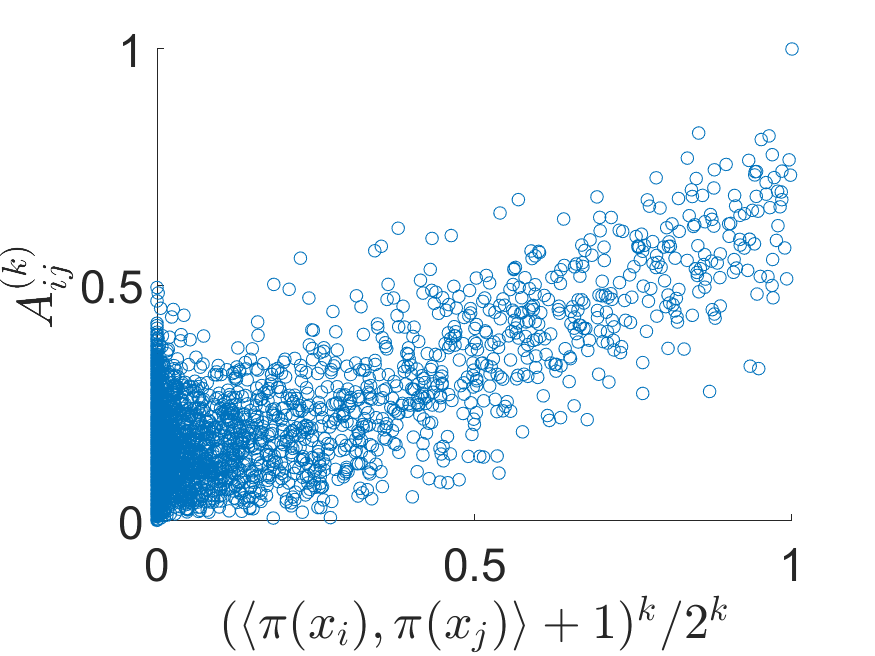} \\\hline
\raisebox{1.2cm}{\rotatebox{90}{$\boldsymbol{p = 0.08}$}}  &   \includegraphics[width = 0.3\textwidth]{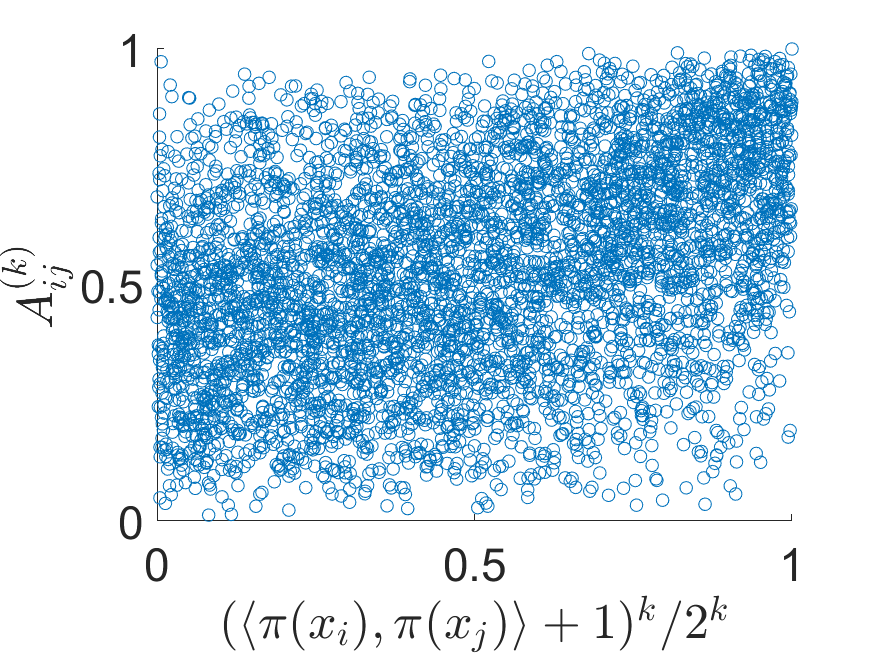} & \includegraphics[width = 0.3\textwidth]{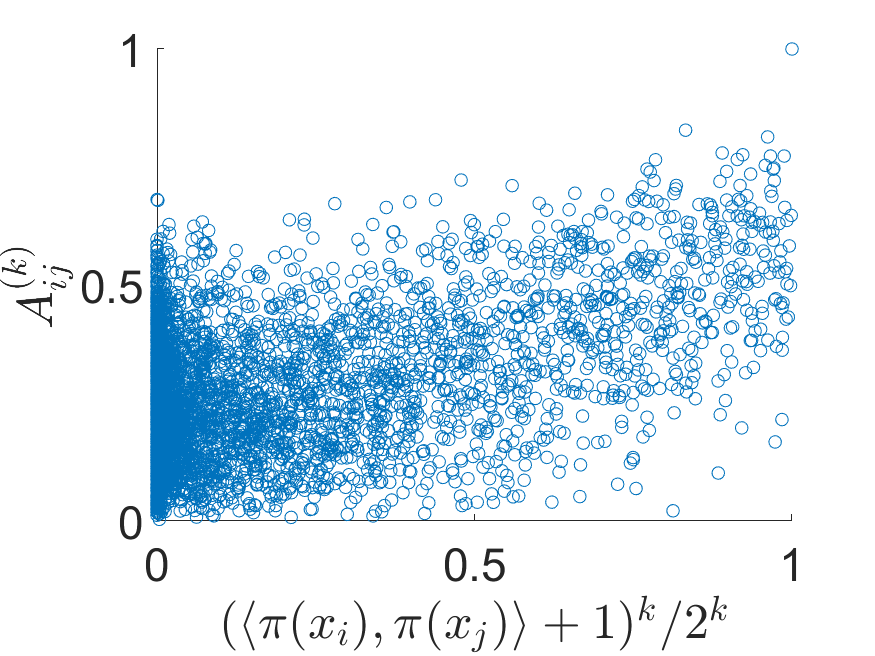}& \includegraphics[width = 0.3\textwidth]{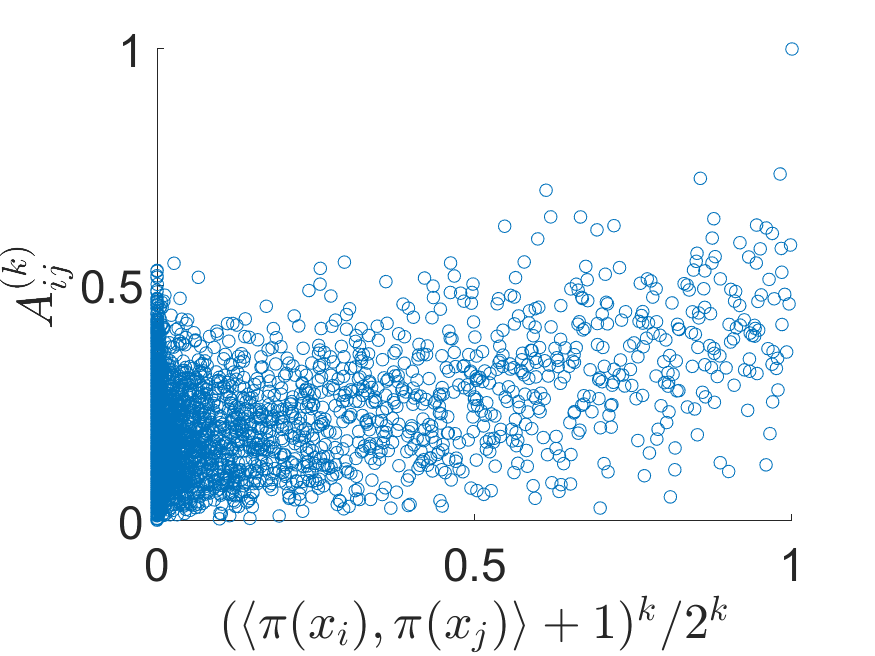} \\
\end{tabular}
    \caption{\small Scatter plots of $A^{(k)}_{ij}$ against $\left(\langle \pi(x_i), \pi(x_j) \rangle + 1 \right)^k / 2^k$ at $p = 1, 0.2, 0.1$ and $0.08$ and $k = 1, 5, \, \text{and }10$. The robustness of the approximation \eqref{eq:approx-identity} is considerably more robust for larger values of $k$.}
    \label{fig:prob_emb_scatter}
\end{figure}

\begin{figure}
    \centering
    \subfloat[$p = 1$]{
    \includegraphics[width = 0.23\textwidth]{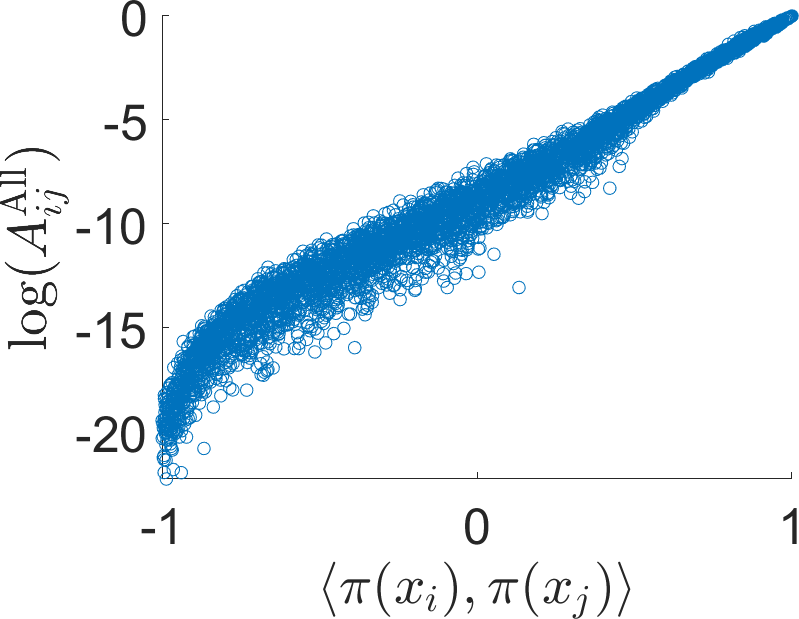}
    \label{fig:mfca_scatter_p100}
    }
   \subfloat[$p = 0.2$]{
    \includegraphics[width = 0.23\textwidth]{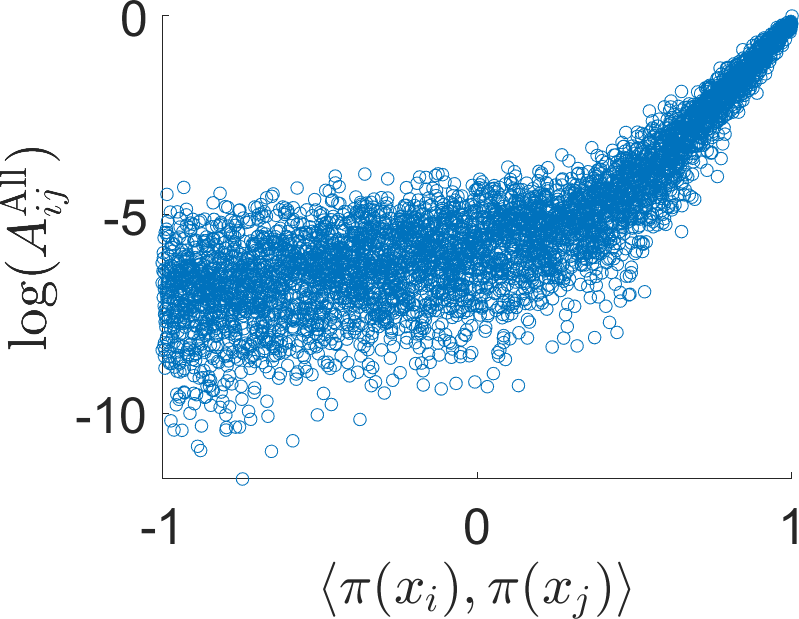}
    \label{fig:mfca_scatter_p20}
    }
    \subfloat[$p = 0.1$]{
    \includegraphics[width = 0.23\textwidth]{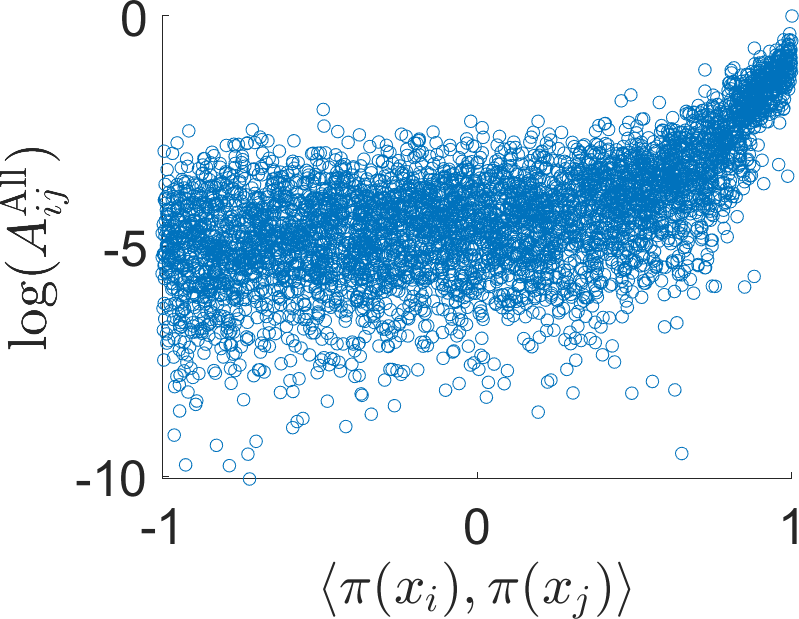}
    \label{fig:mfca_scatter_p10}
    }
    \subfloat[$p = 0.08$]{
    \includegraphics[width = 0.23\textwidth]{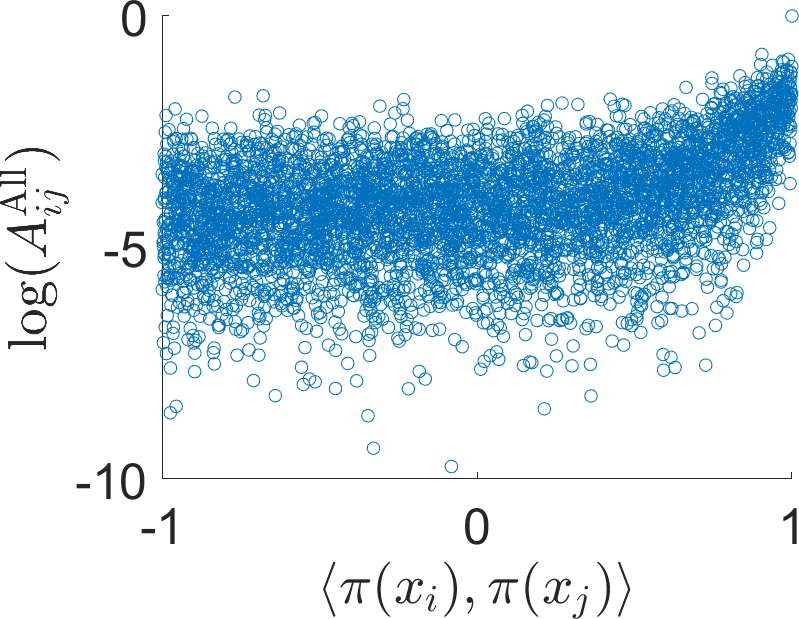}
    \label{fig:mfca_scatter_p8}
    }
    \caption{\small Scatter plots for log multi-frequency class averaging affinity $\log A^{\text{all}}_{ij}$ against $\langle \pi(x_i), \pi(x_j) \rangle $ at $p = 1,\, 0.2, \, 0.1$ and $0.08$.}
    \label{fig:mfca_scatter}
\end{figure}

\subsection{Experiments with Random Rewiring Model}
\label{sec:rrm}
We generate $N = 10,000$ orthonormal frames $x_1, \dots, x_N$ in $\mathbb{R}^3$, uniformly sampled from $\SO(3)$ with respect to the normalized Haar measure. 
To generate the noisy graph under the probabilistic model introduced in~\cite{singer2011viewing}, we keep the correct edge in the neighborhood graph with probability $p$, and use the ground truth local parallel transport data $e^{\imath k \theta_{ij}}$ in~\eqref{eq:align}.  With probability $1 - p$, we rewire the edge such that the node $i$ is connected to a randomly selected node that is not connected with $i$. For the rewired edge, the optimal in-plane rotational alignment angle is replaced with an angle uniformly sampled from $0$ to $2 \pi$. 

In the first experiment, we use a small dataset with $N = 1000$ frames in order to visualize all eigenvalues of $H^{(k)}$. The clean geometric neighborhood is constructed by connecting points where $\langle \pi(x_i) , \pi(x_j) \rangle > 0.8$ (the opening angle $\alpha = 36.9^\circ$) to make sure that the graph is well connected. We vary $p$ and compute all the eigenvalues of $H^{(k)}$ to illustrate the analysis in Section~\ref{sec:noise_model}. Figure~\ref{fig:hist_spec} shows the histograms of the eigenvalues of the matrices $H^{(k)}$ and $R^{(k)}$. We observe that the top eigenvalue of $H^{(k)}$ decreases as $k$ decreases which is consistent with Corollary~\ref{cor:1}. The upper bound for $\| R^{(k)} \|$ as discussed in Section~\ref{sec:noise_model} is $2 \sqrt{N} \sin \frac{\alpha}{2} = 20$, which is consistent with the results shown in the bottom row of Figure~\ref{fig:hist_spec}. In addition, the same figure shows that, $\| R^{(k)} \|$ does not vary with frequency index $k$ under the random rewiring model. Comparing Figure~\ref{fig:hist_spec_k1} with Figure~\ref{fig:hist_spec_k4}, we see that the spectral gap between $(2k+1)^\text{th}$ and $(2k+2)^\text{th}$ eigenvalues increases. Increasing $k$ further, we observe that the $(2k+2)^\text{th}$ eigenvalue of $H^{(k)}$, i.e. $\tilde{\ell}^{(k)}_{2k+2}$, becomes very close to the right edge of the semi-circle as shown in Figure~\ref{fig:hist_spec_k8}.  
Figure~\ref{fig:rrm_n1000_CA} shows the proportion of the estimated 50 nearest neighbors for each frame that satisfy  $\langle \pi(x_i), \pi(x_j) \rangle > 0.85$. 
The proportion reaches the maximum at $k = 9$ for $p = 0.5$ and $p = 0.3$. 

In the second experiment, we use $10,000$ frames to show the spectral properties and the performance of the MFCA algorithm for large sample size. The clean geometric neighborhood graph is constructed by connecting points where $\langle \pi(x_i) , \pi(x_j) \rangle > 0.92$ (within $23.1^\circ$ opening angle). We compute the eigenvalues and eigenvectors of the normalized Hermitian matrix, $\widetilde{H}^{(k)} = D^{-1/2}H D^{-1/2}$. Figure~\ref{fig:prob_eval_bar} shows the top eigenvalues of $\widetilde{H}^{(k)}$. The multiplicities $2k + 1, 2k + 3, 2k + 5, \dots$ of the top eigenvalues are clearly demonstrated in the bar plots for $p = 1$ (the first row in Figure~\ref{fig:prob_eval_bar}). As $p$ decreases, the top spectral gap gets smaller and when $p = 0.1$, it is hard to identify the spectral gap for $k = 1$, whereas the top spectral gap at $k = 5$ is still noticeable. This is consistent with our expectation for improved spectral stability for larger $k$.

\begin{figure}
    \centering
    \subfloat[$p = 1$]{
    \includegraphics[width = 0.23\textwidth]{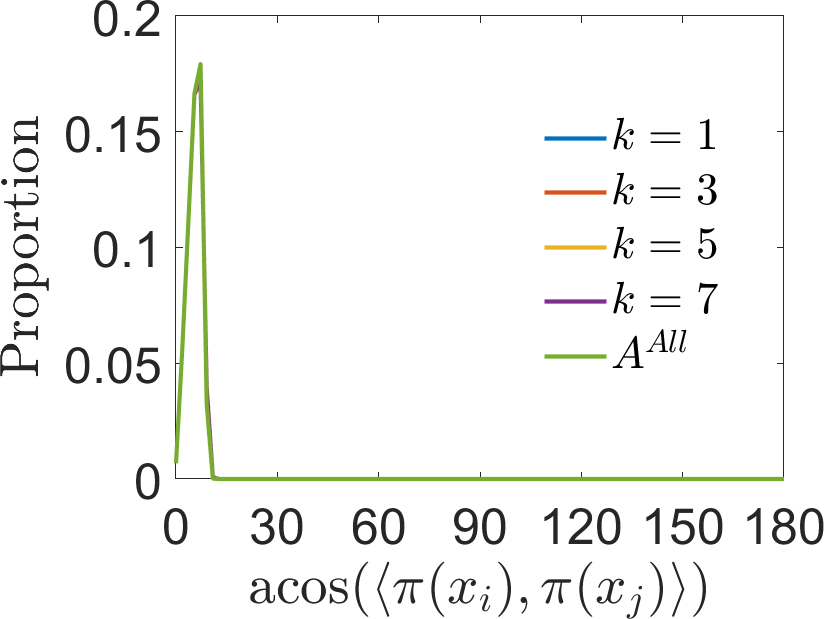}
    \label{fig:mfca_hist_p100}
    }
   \subfloat[$p = 0.2$]{
    \includegraphics[width = 0.23\textwidth]{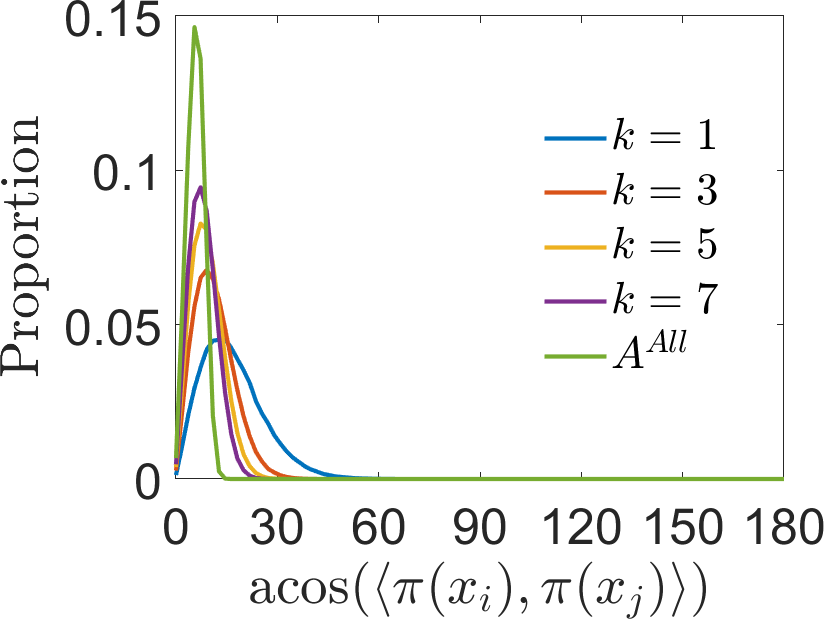}
    \label{fig:mfca_hist_p20}
    }
    \subfloat[$p = 0.1$]{
    \includegraphics[width = 0.23\textwidth]{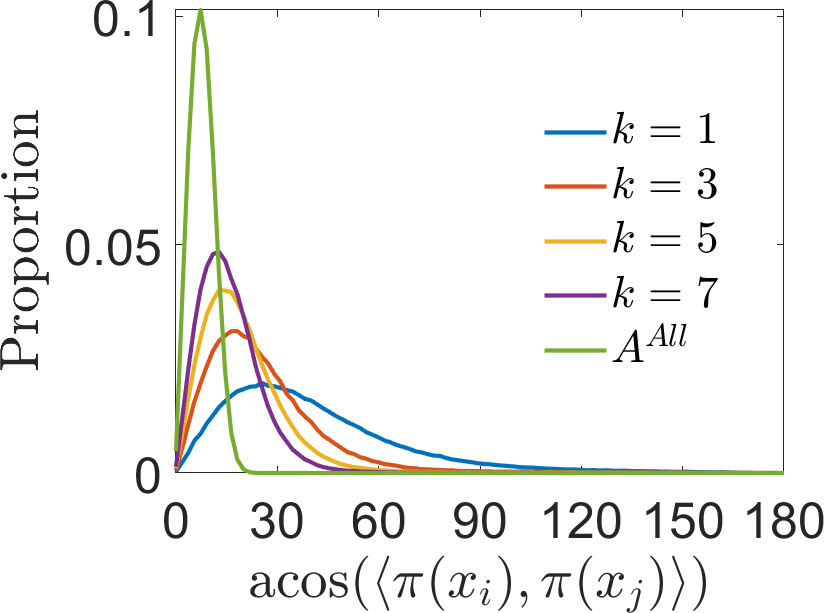}
    \label{fig:mfca_hist_p10}
    }
    \subfloat[$p = 0.08$]{
    \includegraphics[width = 0.23\textwidth]{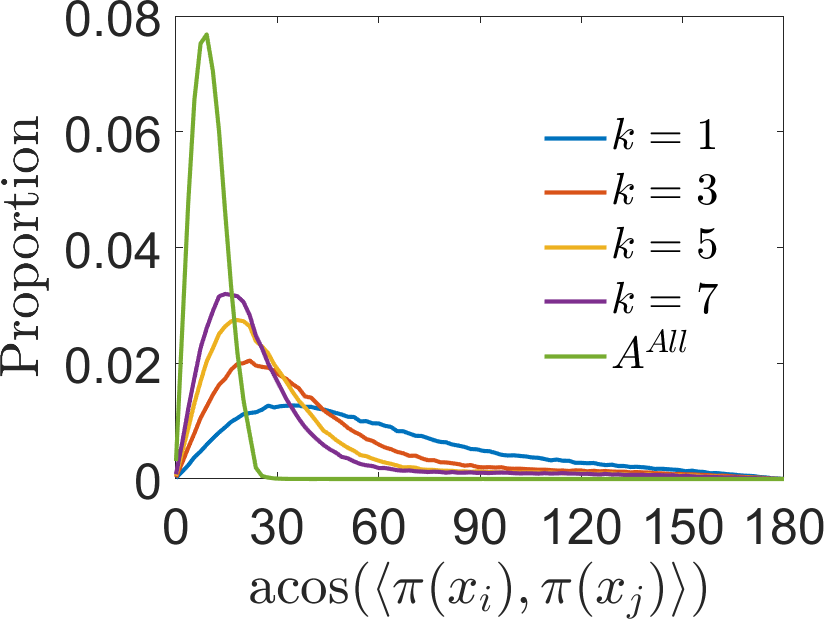}
    \label{fig:mfca_hist_p8}
    }
    \caption{\small Histogram of the angles ($x$-axis, in degrees) between the viewing directions of 10,000 simulated frames and its 50 neighboring points at $p = 1,\, 0.2,\, 0.1,\,\text{and } 0.08$. For $A^\text{All}$, we use $k_\text{max} = 20$.  }
    \label{fig:mfca_hist}
\end{figure}
The estimated $A_{ij}^{(k)}$'s provide good approximations to $\left(\langle \pi(x_i), \pi(x_j) \rangle + 1 \right )^k/2^k $ (see the top row of Figure~\ref{fig:prob_emb_scatter}). This approximation deteriorates as $p$ decreases.  The lower left sub-figure of Figure~\ref{fig:prob_emb_scatter} shows that the original single frequency class averaging nearest neighbor search algorithm fails at $p = 0.08$. Figure~\ref{fig:mfca_scatter} shows the scatter plots of the combined affinity against the dot products $\langle \pi(x_i), \pi(x_j) \rangle $ between the true viewing angles at varying $p$. Even at $p = 0.08$, the combined affinity $A_{ij}^\text{All}$ is still able to identify frames of similar viewing directions. 
\begin{figure}
	\captionsetup[subfigure]{oneside,margin={0.03cm,0cm}}
    \centering
     \subfloat[$p = 0.2$]{
    \includegraphics[height = 0.25 \textwidth]{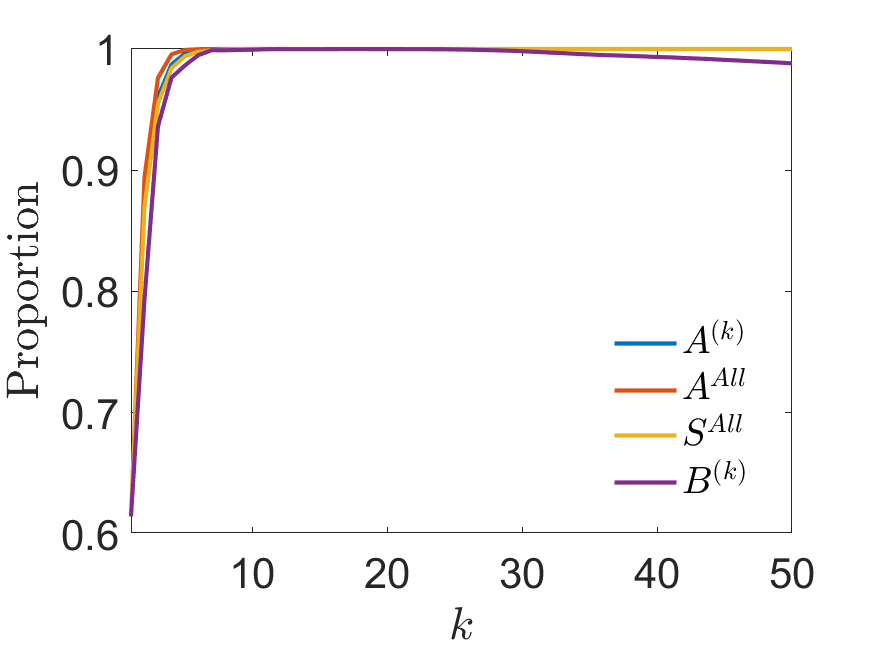}
    \label{fig:rrm_vark_p20}
    } 
    \subfloat[$p = 0.1$]{
    \includegraphics[height = 0.25 \textwidth]{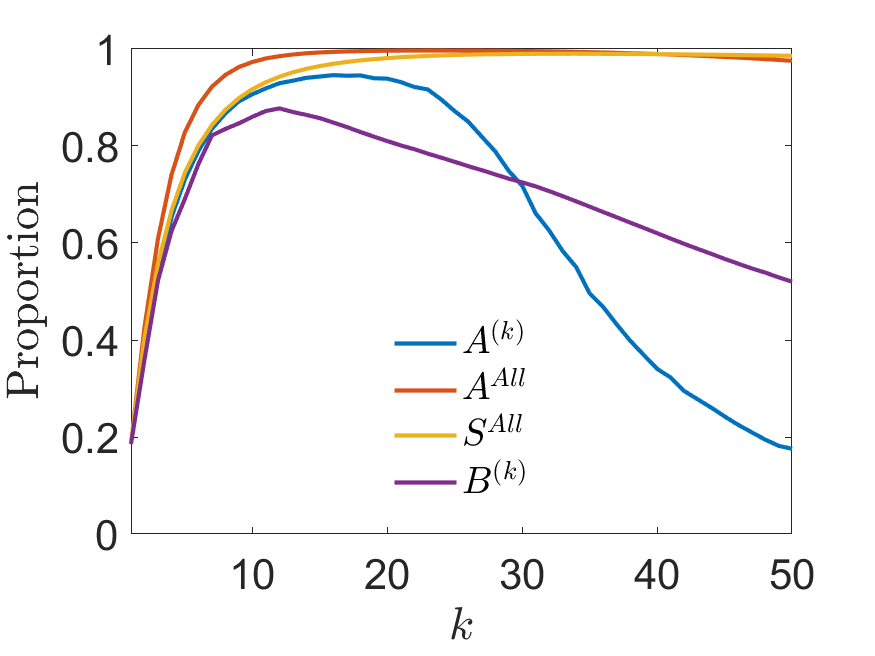}
    \label{fig:rrm_vark_p10}
    } 
    \subfloat[$p = 0.08$]{
    \includegraphics[height = 0.25 \textwidth]{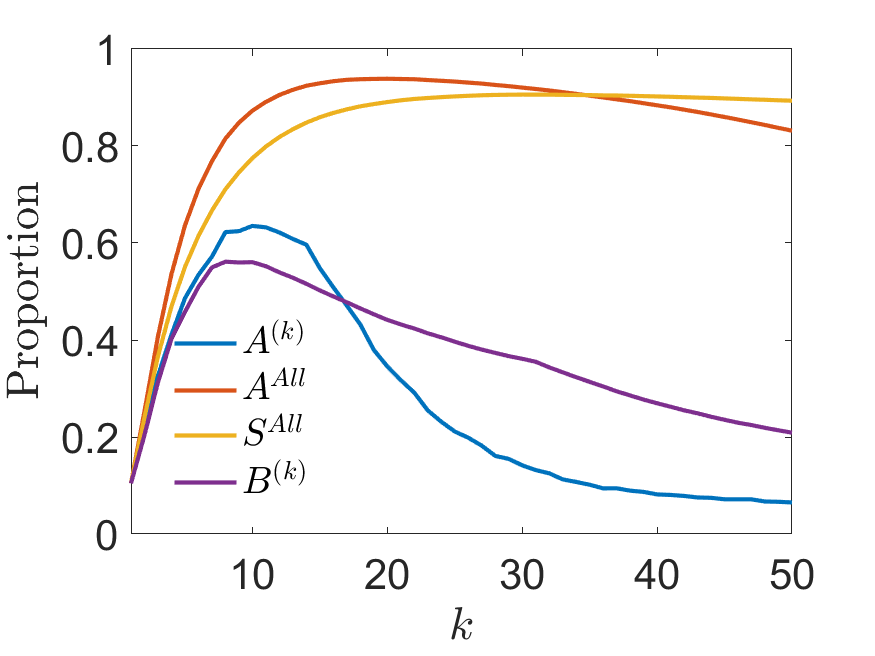}
    \label{fig:rrm_vark_p8}
    }
    \caption{\small Comparing the performance of different affinities according to $A^{(k)}$, $A^{\text{All}}$, $S^{\text{All}}$, and $B^{(k)}$. We evaluate the proportion of the estimated nearest neighbors that satisfy $\langle \pi(x_i), \pi(x_j)\rangle > 0.95$. }
    \label{fig:rrm_vark}
\end{figure}

We evaluate the performance of the proposed algorithms on the nearest neighbor search by 
inspecting the magnitudes of the angles between the viewing directions of frames identified as neighbors by the algorithm. We identify for each frame $50$ nearest neighbors with respect to the affinity measure, and plot in Figure~\ref{fig:mfca_hist} the histogram of the angles between the viewing directions of neighboring frames for varying rewiring probabilities $p=1,0.2,0.1,0.08$. From Figure~\ref{fig:mfca_hist}, we observe that using the affinity $A^{(k)}$ in~\eqref{eq:affinity-freq-k-prac} at higher frequency helps improve the performance of the single-frequency class averaging nearest neighbor search algorithm, especially for the noisy graph at $p = 0.08$ (i.e., $92\%$ of the true edges are corrupted). Moreover, combining the measures at different $k$'s according to~\eqref{eq:all-frequency-affinity-prac} further improves the classification results with significant reduction of outliers at $p = 0.08$ compared to the single frequency nearest neighbor identification results. 
\begin{figure}
	\captionsetup[subfigure]{oneside,margin={0.03cm,0cm}}
	\centering
	\subfloat[$\langle \pi(x_i), \pi(x_j) \rangle = -0.20$]{\includegraphics[width= 0.3\textwidth]{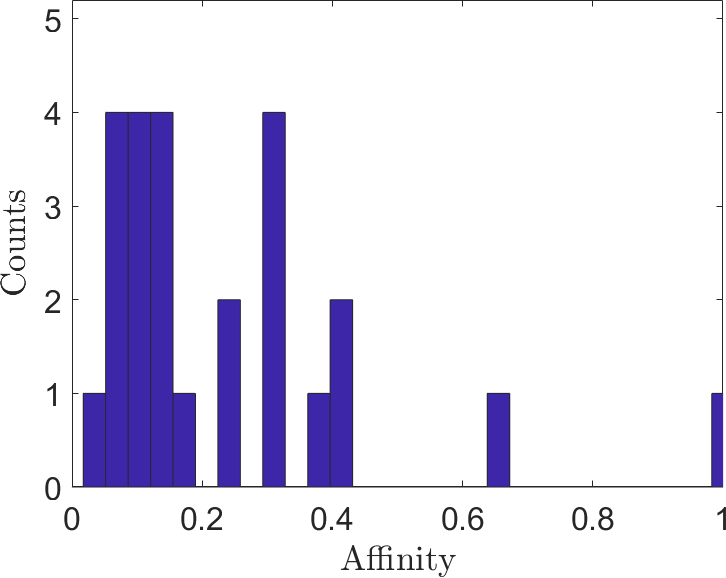}
	\label{fig:hist_single_pair_a}}\quad
	\subfloat[$\langle \pi(x_i), \pi(x_j) \rangle = 0.99$]{\includegraphics[width= 0.3\textwidth]{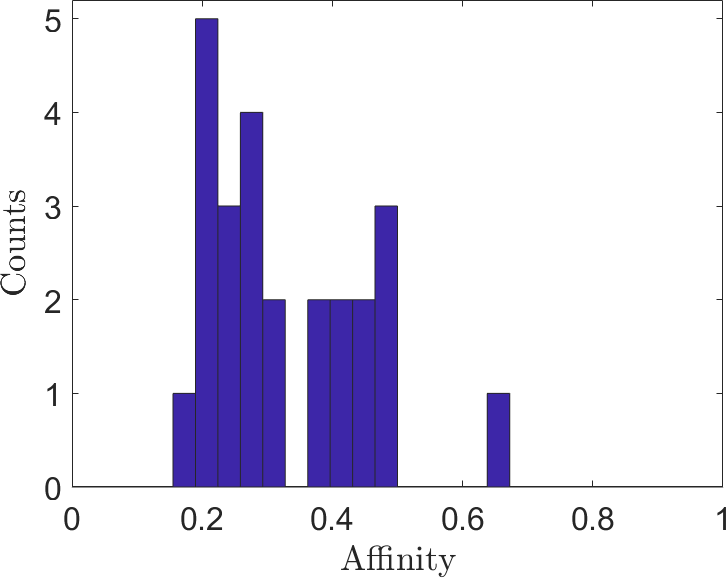}
		\label{fig:hist_single_pair_b}
		}
	\caption{\small Histograms of the affinities $A^{(k)}_{ij}$ with $k = 1, \dots, 25$ for (a) a pair of wrongly identified nearest neighbors by $A^{(1)}$ and (b) a good nearest neighbor pair identified by $A^\text{All}$, but not by any $A^{(k)}$. The data is generated under random rewiring model with $p = 0.08$.}
	\label{fig:hist_single_pair}
\end{figure}

Singer et al. proposed to use more than top 3 eigenvectors from $\widetilde{H}^{(1)}$ for nearest neighbor classification in~\cite[Section 7]{singer2011viewing}. We include it as an additional baseline for comparison here to illustrate the benefit of using the eigenvectors of $\widetilde{H}^{(k)}$ for $k > 1$. Specifically, using the top $2k+1$ eigenvectors of $\widetilde{H}^{(1)}$, we define the affinity $B^{(k)}$ as,
\begin{equation}
    \label{eq:B}
    \widetilde{\Psi}^{(1)}_k(i) = \left(\psi^{(1)}_1(i), \psi^{(1)}_2(i), \dots, \psi^{(1)}_{2k + 1}(i)\right), \quad 
    B^{(k)} = \frac{\left | \langle \widetilde{\Psi}^{(1)}_k(i), \widetilde{\Psi}^{(1)}_k(j) \rangle \right |}{\|\widetilde{\Psi}^{(1)}_k(i) \| \| \widetilde{\Psi}^{(1)}_k(j) \|}. 
\end{equation}
We compare the performance of the algorithms in terms of the proportion of estimated nearest neighbors that satisfy $ \langle \pi (x_i), \pi (x_j) \rangle > 0.95$. Figure~\ref{fig:rrm_vark_p10} shows that under large noise regimes, where $90\%$ of the clean edges are randomly rewired, using $A^{(k)}$ at $k = 16$ outperforms the previous class averaging algorithm that uses only the eigenvectors from $\widetilde{H}^{(1)}$. 
As shown in Figure~\ref{fig:rrm_vark_p8}, combining the information from different frequency channels can significantly boost the performance in finding true nearest neighbors. For $S^\text{All}$, the proportion reaches the maximum value $0.90$ at $k = 30$. For $A^\text{All}$, the proportion reaches the maximum value $0.94$ at $k = 20$. Because the higher-order terms $A^{(k)}$ get much smaller than 1 and become less informative, incorporating more $A^{(k)}$ components deteriorates the performance of the combined score $A^{\text{All}}$ when $k > 20$. The combined affinity $S^{\text{All}}$ is more stable at large $k$. 

To understand why the combined affinities can significantly improve the classification results at $p = 0.08$, we check the values of $A^{(k)}$ for $k = 1, \dots, 25$ for pairs of frames $x_i$ and $x_j$ that satisfy $\langle \pi(x)_i, \pi(x)_j \rangle < 0.95$, but are still identified as nearest neighbors by $A^{(1)}$. We observe that although the corresponding affinities at frequency 1 are above 0.97, $A^{(k)}_{ij}$ at other frequency indices are below 0.7 and concentrated on the interval $(0, 0.2]$ (see the example in Figure~\ref{fig:hist_single_pair_a}). Therefore, the combined affinity $A^{\text{All}}$ is very small and such pair will be removed from the nearest neighbor list. In contrast, for a pair of true nearest neighbors that does not appear in any nearest neighbor list by $A^{(k)}$ for $k = 1, \dots 25$, we observe that although the affinities are lower than 0.7, all individual affinities lie between 0.2 and 0.5 (see Figure~\ref{fig:hist_single_pair_b}). Thus the combined affinity $A^\text{All}$ is higher for the pair in Figure~\ref{fig:hist_single_pair_b} than the pair in Figure~\ref{fig:hist_single_pair_a}. In summary, $A^\text{All}$ is able to not only reject wrongly identified nearest neighbors by $A^{(k)}$, but also find new correct nearest neighbors that are missed by $A^{(k)}$.  
 
In the third experiment, we incorporate the small angular perturbation into the random rewiring model according to Eq.~\eqref{eq:Hk2}. Specifically, we assume that the distribution of the angular error follows the von Mises distribution,
\begin{equation}
    \label{eq:vonMises}
    \gamma(\varepsilon) = \frac{e^{\kappa \cos(\varepsilon)}}{2 \pi I_0(\kappa)},
\end{equation}
where $I_0(\kappa)$ is the modified Bessel function of order $0$. The parameter $\kappa$ controls the concentration of the distribution. For this particular distribution, $c_k = \mathbb{E}(e^{\imath k \epsilon}) = \frac{I_k(\kappa)}{I_0(\kappa)}$, where $I_k(\kappa)$ is the modified Bessel function of order $k$ for $k>0$. 
The clean geometric neighborhood graph is constructed by connecting points where $\langle \pi(x_i) , \pi(x_j) \rangle > 0.7$ with 10,000 frames. We fix $p = 0.08$ ($92\%$ of the clean edges are randomly rewired) and vary the parameter $\kappa$ in von Mises distribution. We show the accuracy of the 50-nearest neighbor identification in Figure~\ref{fig:result_ang_perturb}. Figure~\ref{fig:ang_pdf} depicts the distribution of the angle $\varepsilon$ with $\kappa = 500$ and $\kappa = 64$. 
\begin{figure}
    \captionsetup[subfigure]{oneside,margin={0.03cm,0cm}}
    \centering
    \subfloat[$\gamma(\varepsilon)$]{
    \includegraphics[width = 0.24 \textwidth]{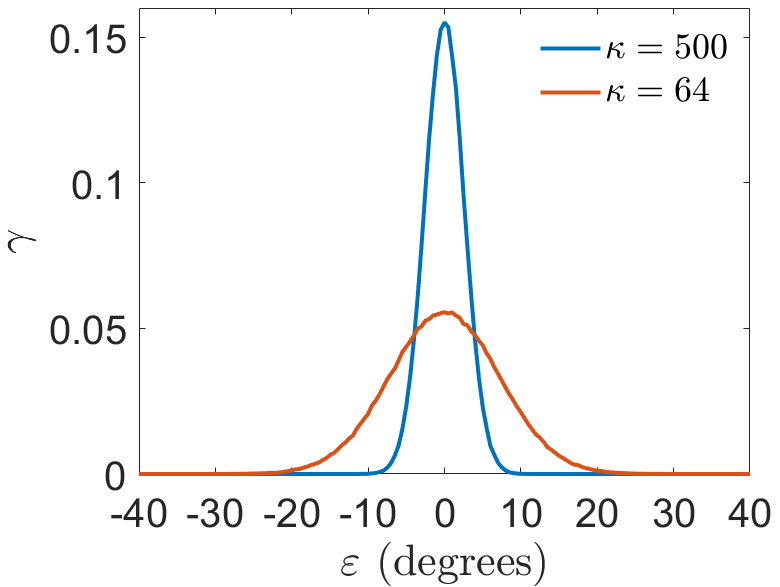}
    \label{fig:ang_pdf}
    }
    \subfloat[$\kappa \rightarrow \infty$] {
    \includegraphics[width = 0.24 \textwidth]{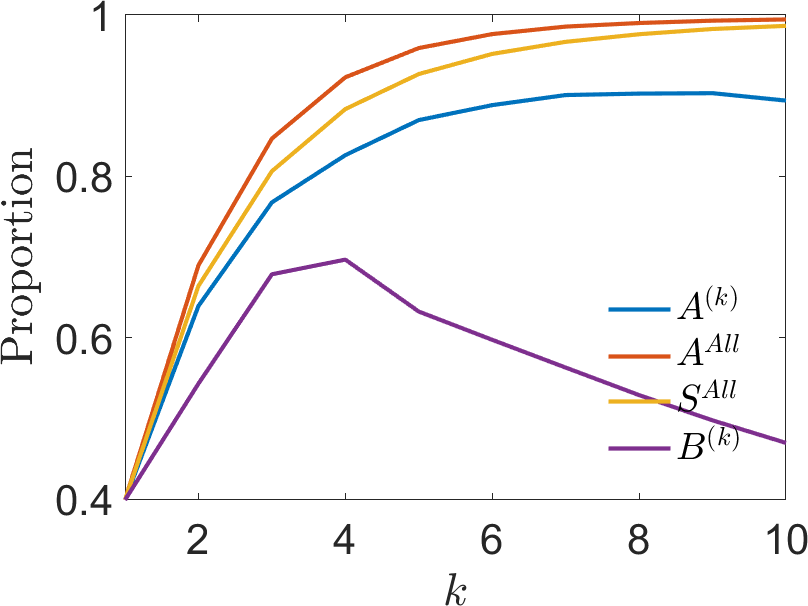}
        \label{fig:ang_beta5000}
    }
    \subfloat[$\kappa = 500$]{
    \includegraphics[width = 0.24 \textwidth]{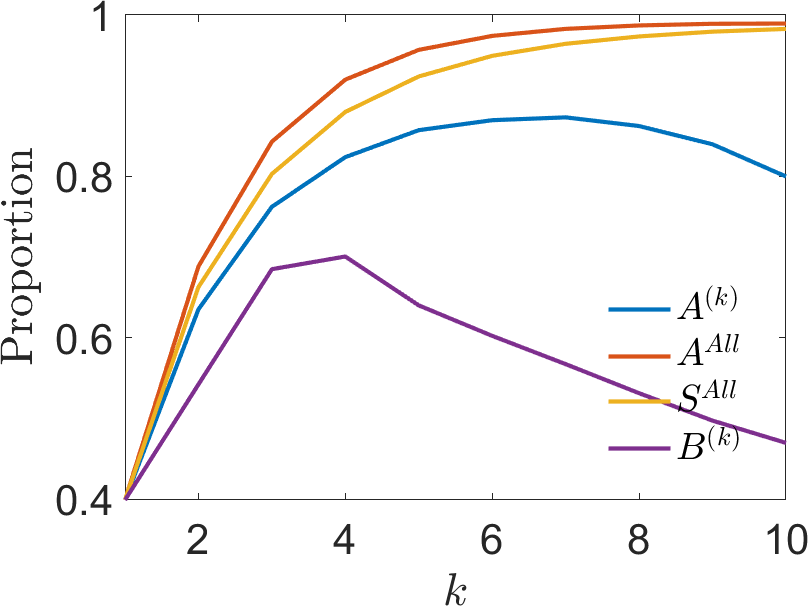}
    \label{fig:ang_beta500}
    }
    \subfloat[$\kappa = 64$]{
    \includegraphics[width = 0.24 \textwidth]{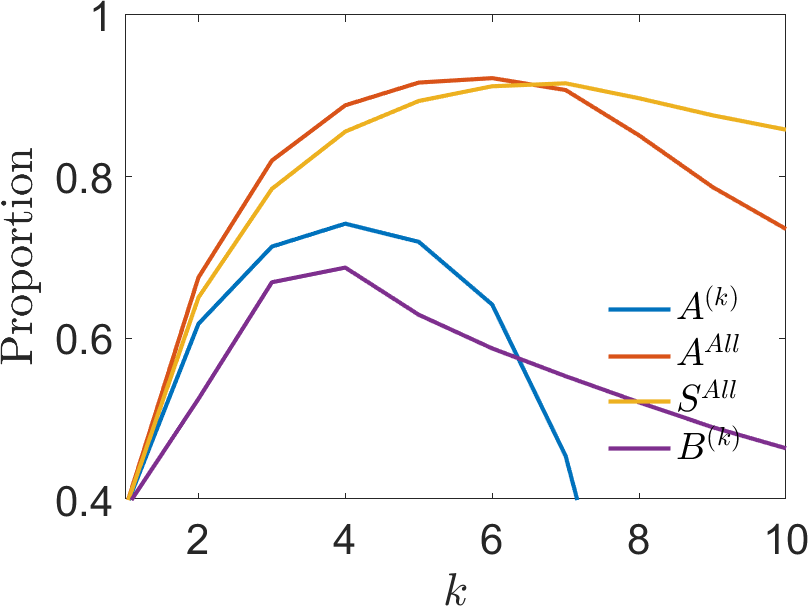}
    \label{fig:ang_beta64}
    }
	\caption{\small Comparing the performance of different affinities according to $A^{(k)}$, $A^{\text{All}}$, $S^{\text{All}}$, and $B^{(k)}$. \protect \subref{fig:ang_beta5000}--\protect \subref{fig:ang_beta64} The proportion of the estimated nearest neighbors that satisfy $\langle \pi(x_i), \pi(x_j)\rangle > 0.95$ with various $\kappa$. \protect\subref{fig:ang_pdf} Distributions of $\varepsilon_{ij}$ for $\kappa = 500$ and $\kappa = 64$. }
    \label{fig:result_ang_perturb}
\end{figure}
Figure~\ref{fig:ang_beta5000} shows the results for random rewiring model without angular perturbation and the performance of $A^{(k)}$ is consistently better than $B^{(k)}$. Comparing Figure~\ref{fig:ang_beta5000} with Figure~\ref{fig:rrm_vark_p8}, we find that we achieve higher accuracy in the nearest neighbor identification from a more densely connected graph in all approaches. From Figures~\ref{fig:ang_beta5000}--\ref{fig:ang_beta64}, we find the performance of $B^{(k)}$ is stable over small angular perturbation. In comparison, the performance of single frequency affinity $A^{(k)}$ deteriorates as $\kappa$ increases. This is due to the fact that $c_k$ gets smaller as $\kappa$ increases and both the top spectral gap and top eigenvalue of $\mathbb{E}H^{(k)}$ depend on $c_k$. Despite this, the combined scores still achieve higher accuracy than $A^{(k)}$ and $B^{(k)}$. 

\begin{figure}[t!]
	\captionsetup[subfigure]{oneside,margin={0.03cm,0cm}}
	\captionsetup[subfigure]{labelformat=empty}
	\centering
	\subfloat[Clean Projections]{\includegraphics[width= 0.2\textwidth]{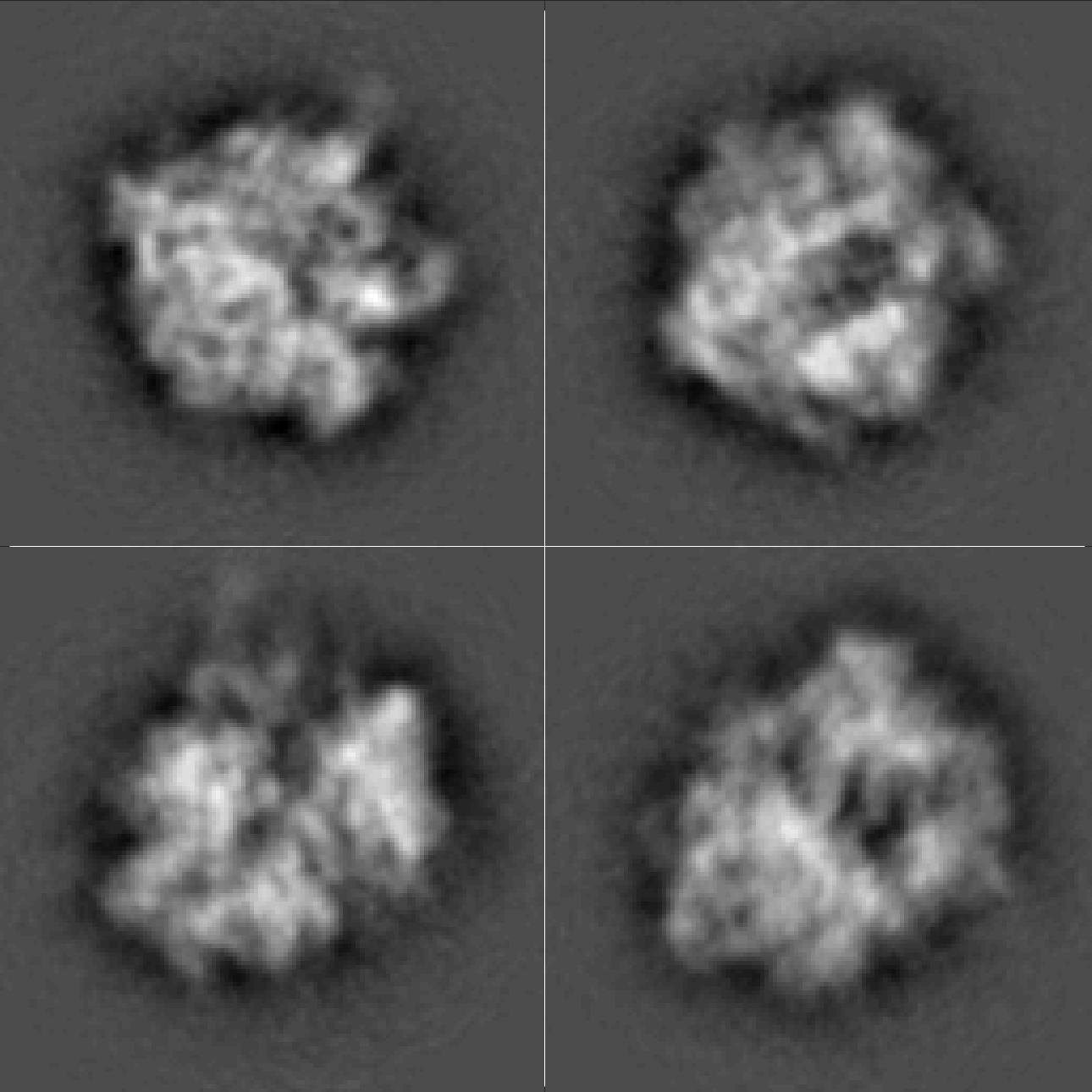}}\;
	\subfloat[$\text{SNR} = 0.05$]{\includegraphics[width= 0.2\textwidth]{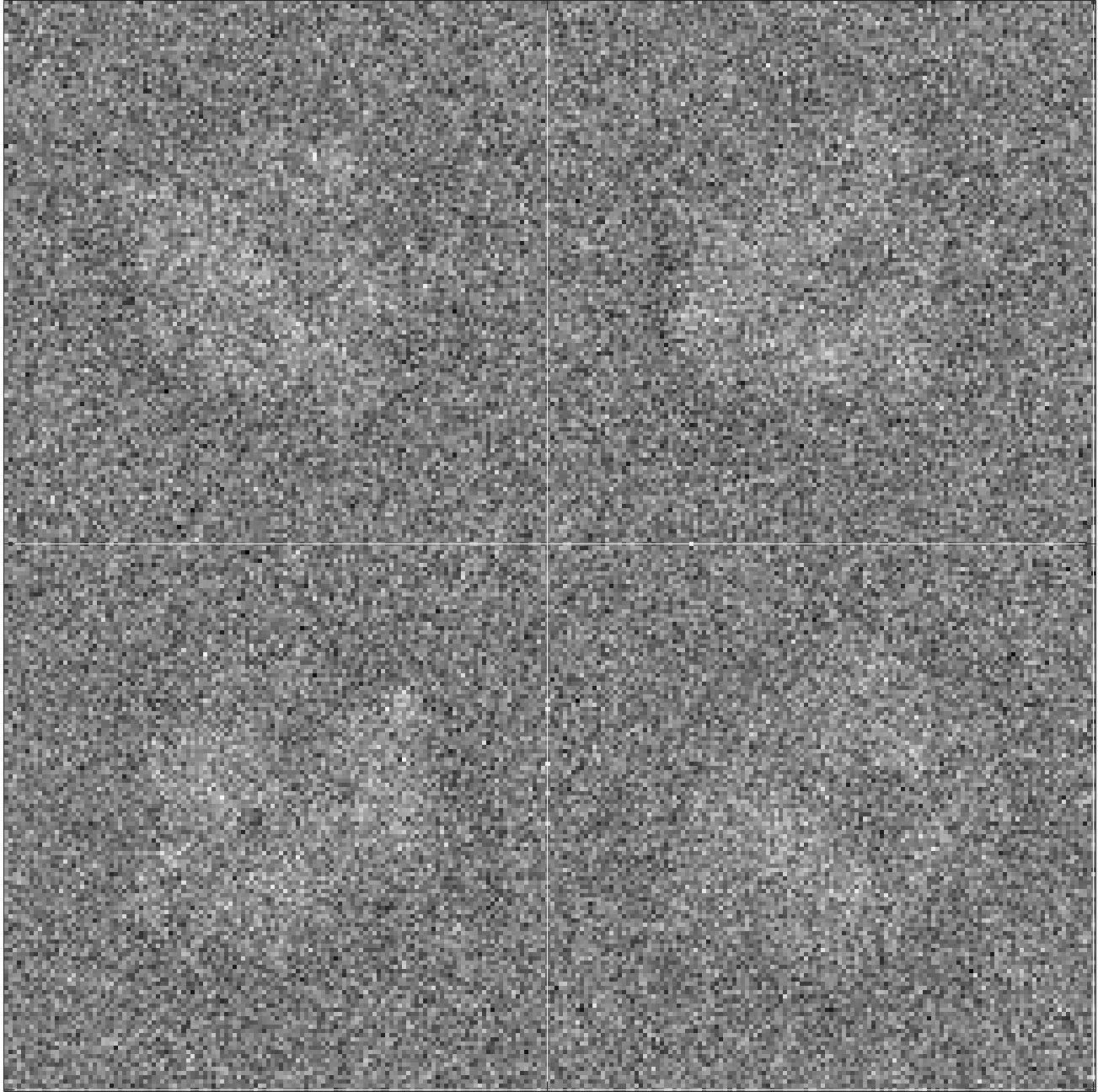}}\;
	\subfloat[$\text{SNR} = 0.01$]{\includegraphics[width= 0.2\textwidth]{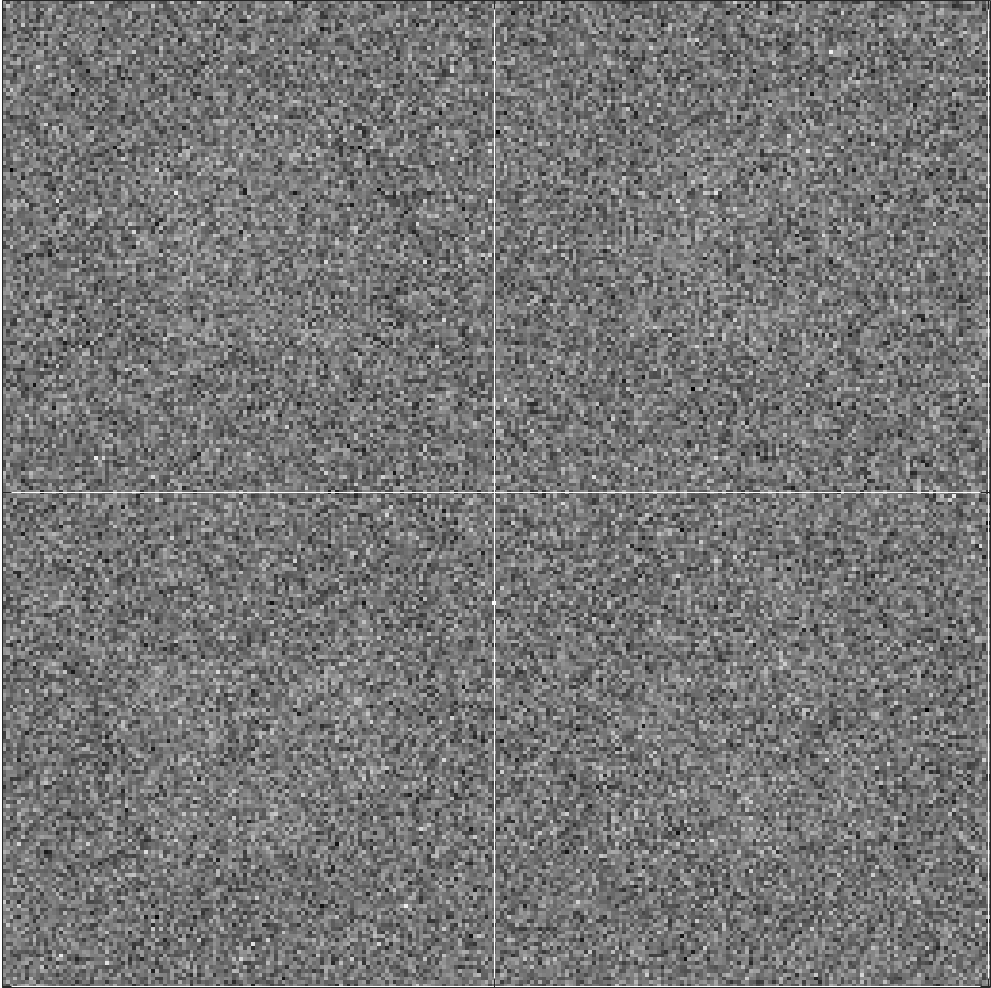}}\;
	\subfloat[$\text{SNR} = 0.008$]{\includegraphics[width= 0.2\textwidth]{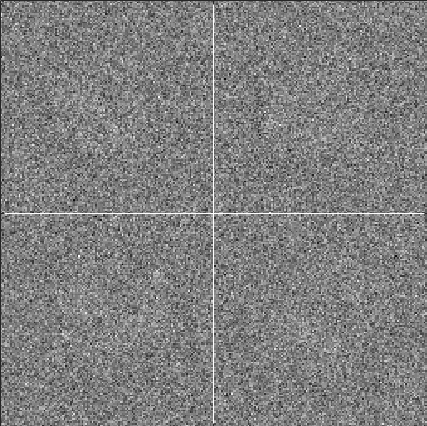}}
\caption{\small Samples of simulated projection images on 70S ribosome. From left to right: Clean projection images, images contaminated by additive white Gaussian noise with signal to noise ratio SNR$ = 0.05,\, 0.01, \text{ and }0.008$.}
\label{fig:cryo_sample}
\vspace{-0.4cm}
\end{figure}

\begin{figure}[t!]
	\centering
	\setlength\tabcolsep{1.0pt}
	\renewcommand{\arraystretch}{0.6}
	\begin{tabular}{p{0.05\textwidth}<{\centering} p{0.3\textwidth}<{\centering} p{0.3\textwidth}<{\centering} p{0.3\textwidth}<{\centering}}
		& $\boldsymbol{k = 1}$ & $\boldsymbol{k = 3}$ & $\boldsymbol{k = 5}$ \\
		\raisebox{1.2cm}{\rotatebox{90}{\textbf{Clean}}} 
		&\includegraphics[width = 0.3\textwidth]{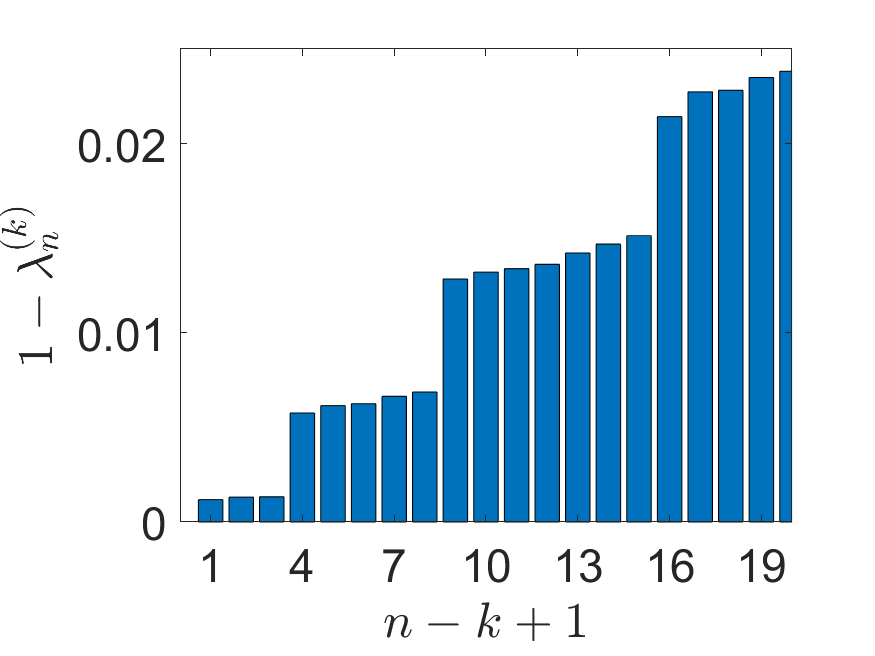} 
		&\includegraphics[width = 0.3\textwidth]{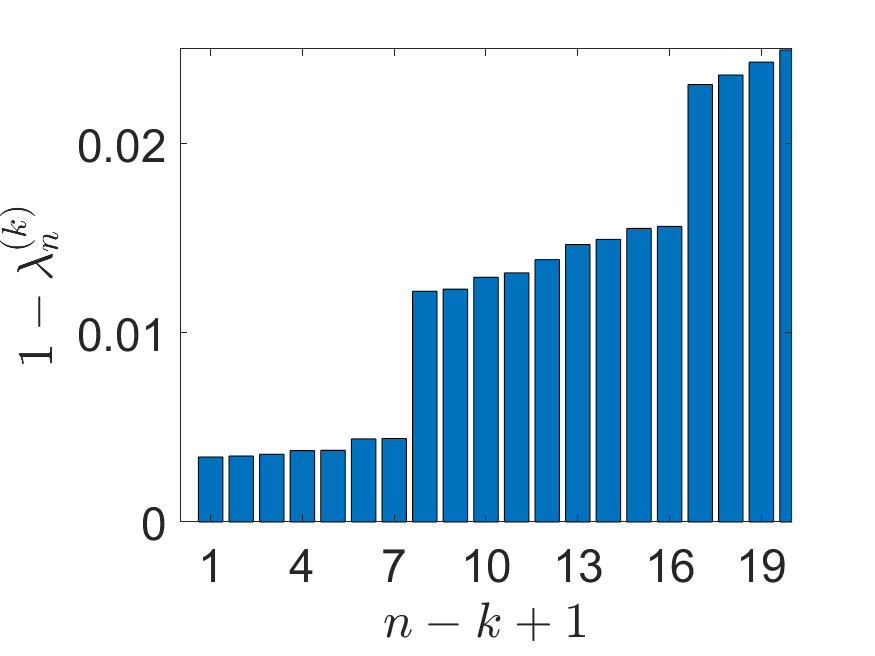} 
		&\includegraphics[width = 0.3\textwidth]{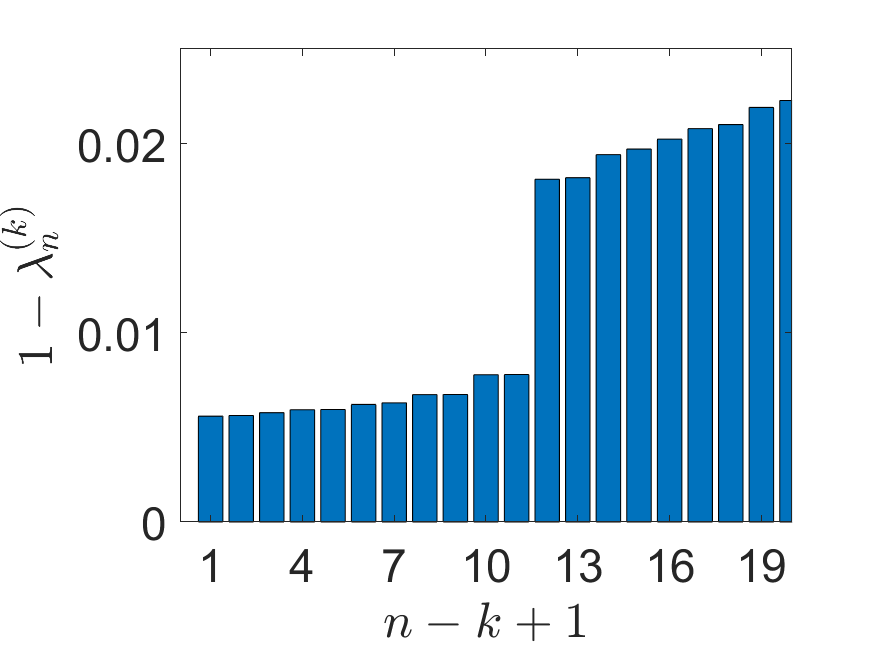} \\
		\raisebox{1.2cm}{\rotatebox{90}{${\textbf{SNR}\bm{ = 0.05}}$}} 
		&\includegraphics[width = 0.3\textwidth]{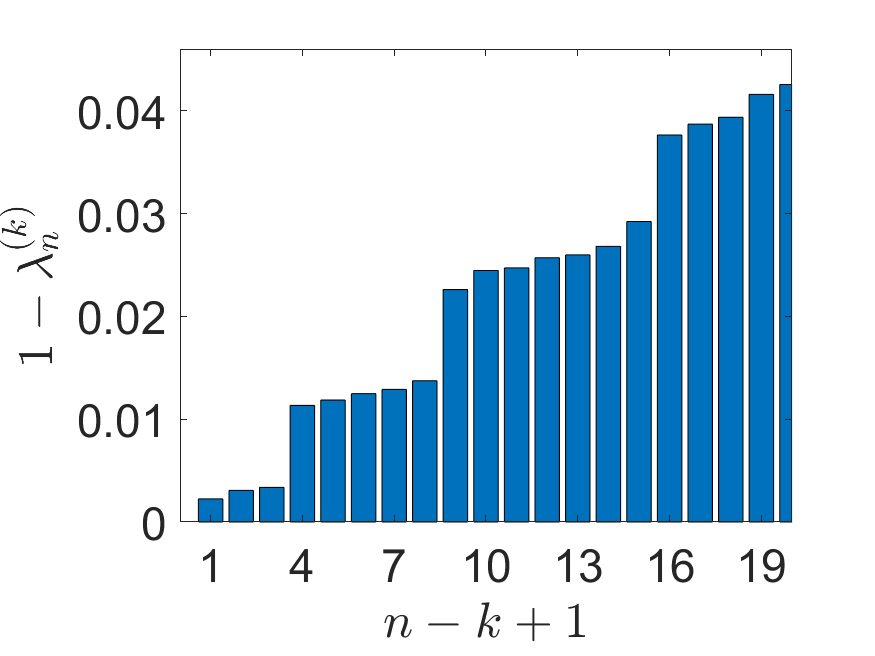} 
		&\includegraphics[width = 0.3\textwidth]{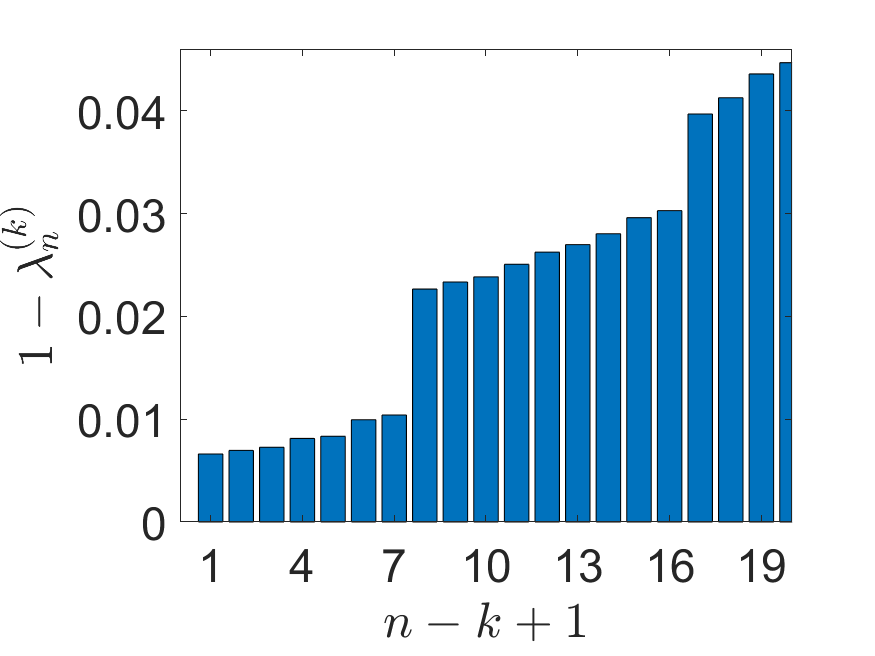} 
		&\includegraphics[width = 0.3\textwidth]{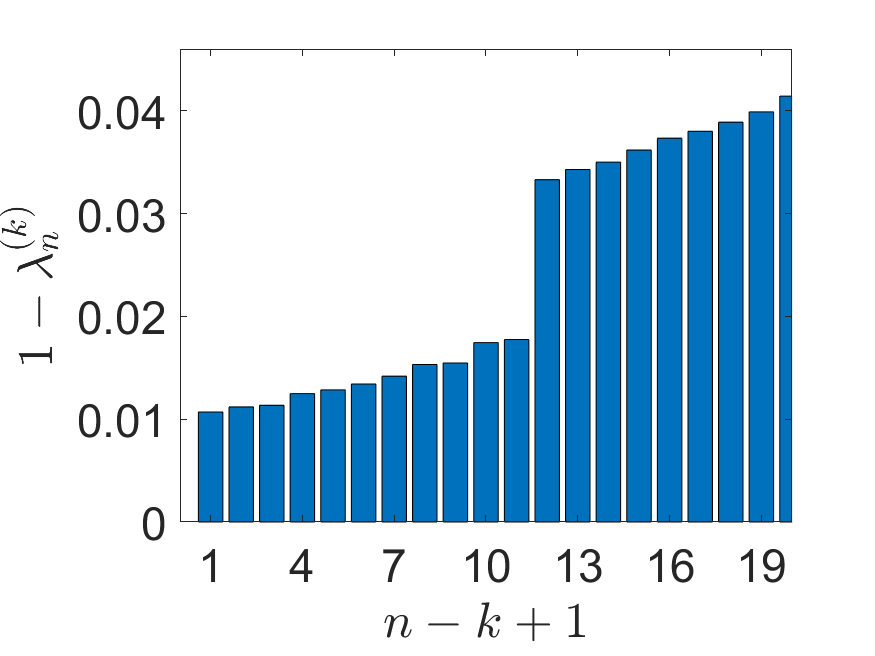} \\
		\raisebox{0.9cm}{\rotatebox{90}{$\boldsymbol{\textbf{SNR} = 0.01}$}} 
		&\includegraphics[width = 0.3\textwidth]{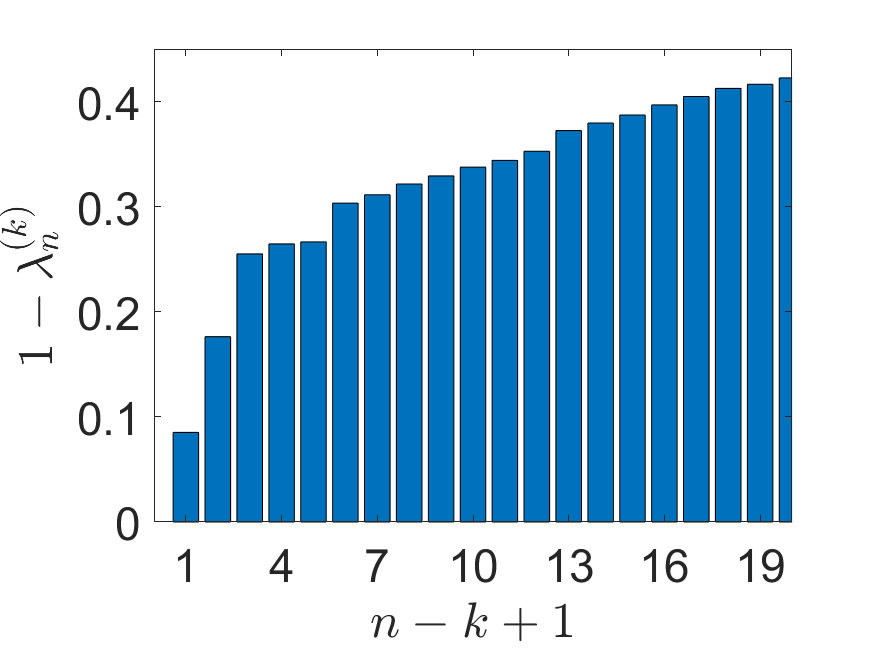} 
		&\includegraphics[width = 0.3\textwidth]{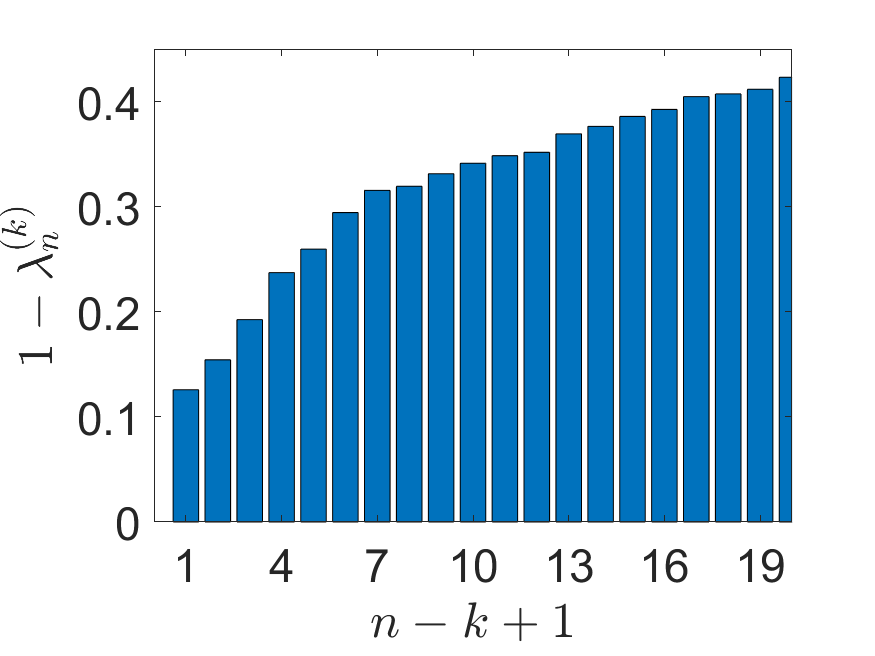} 
		&\includegraphics[width = 0.3\textwidth]{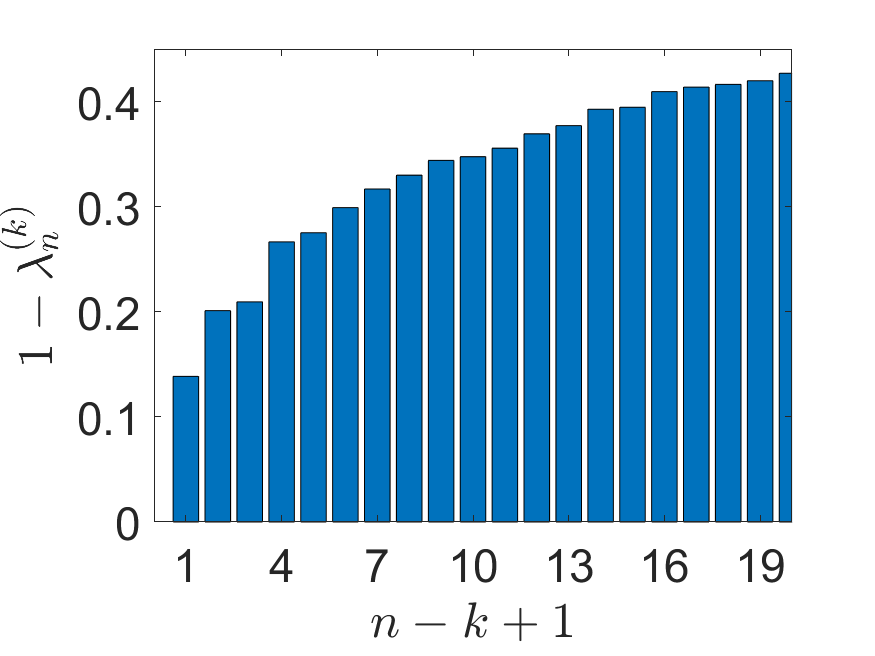} \\
        \raisebox{0.8cm}{\rotatebox{90}{$\boldsymbol{\textbf{SNR} = 0.008}$}} 
		&\includegraphics[width = 0.3\textwidth]{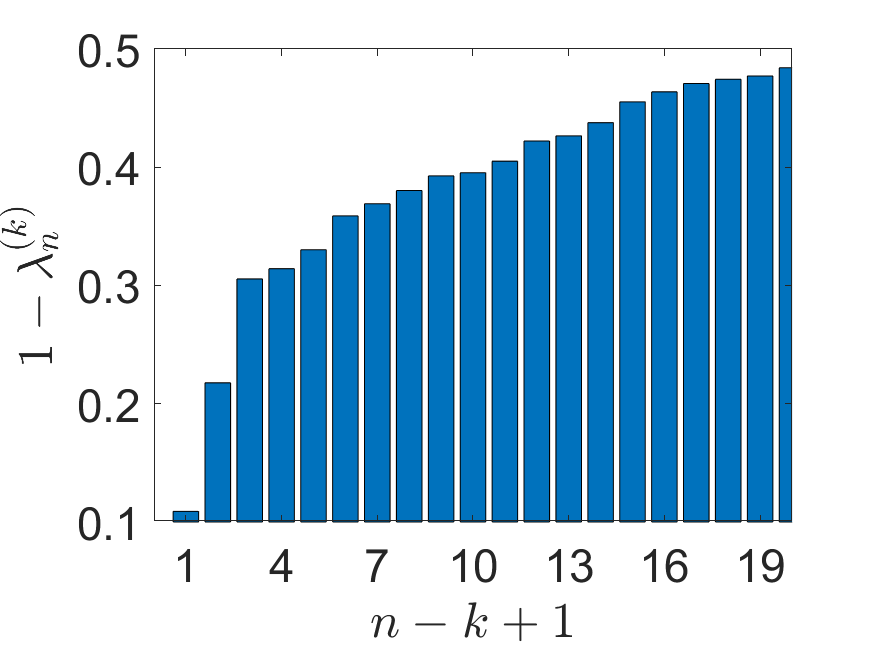} 
		&\includegraphics[width = 0.3\textwidth]{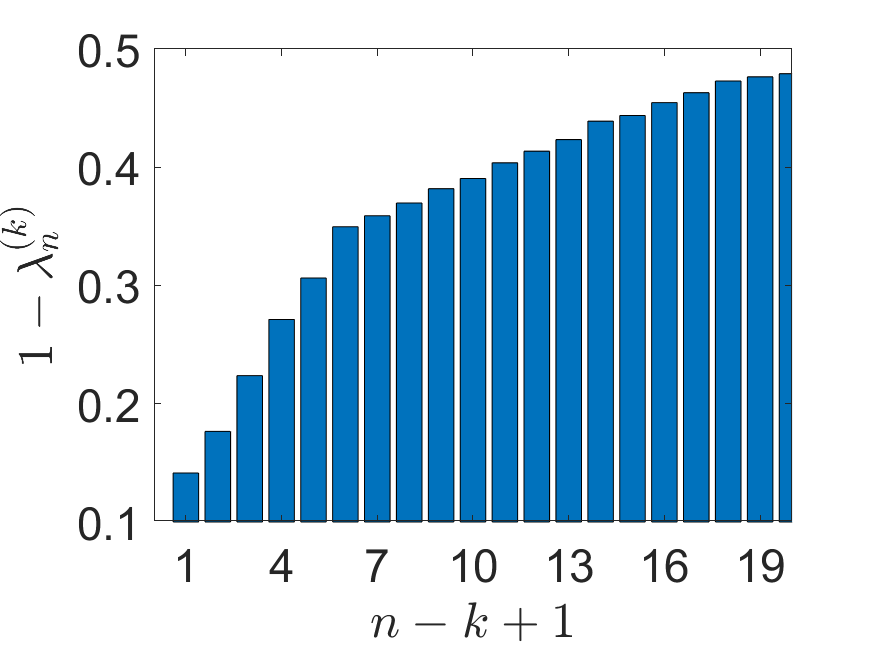} 
		&\includegraphics[width = 0.3\textwidth]{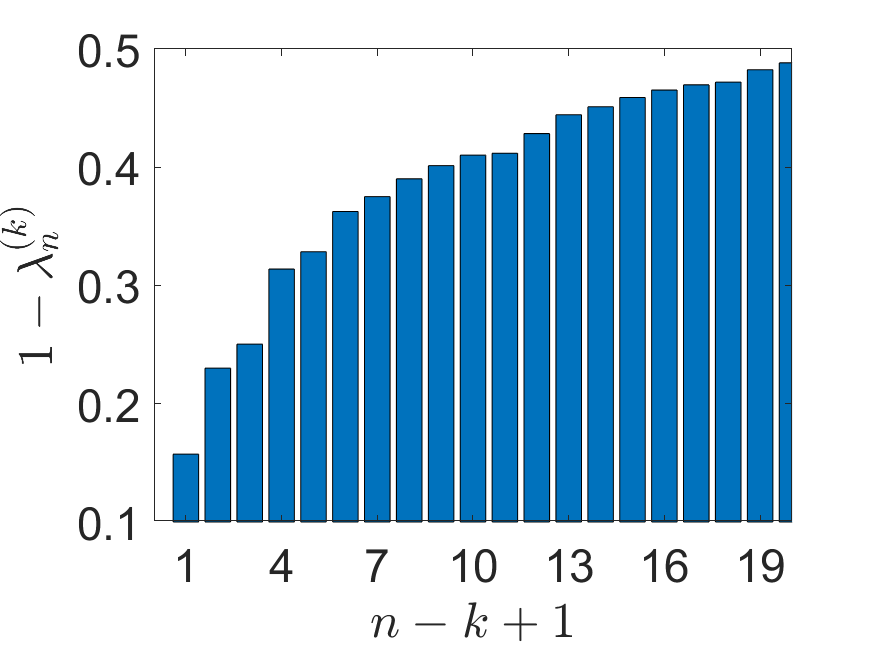} \\
	\end{tabular}
	\caption{\small Bar plots of the top 20 eigenvalues at different frequency $k$ and signal to noise ratio (SNR) for simulated cryo-EM projection images.}
	\label{fig:cryo_spectrum}
	\vspace{-0.4cm}
\end{figure}

\subsection{Experiments with Simulated Cryo-EM Images}
\label{sec:exp_cryo}

\begin{figure}[t!]
	\centering
	\setlength\tabcolsep{1.0pt}
	\renewcommand{\arraystretch}{0.6}
	\begin{tabular}{p{0.05\textwidth}<{\centering} p{0.3\textwidth}<{\centering} p{0.3\textwidth}<{\centering} p{0.3\textwidth}<{\centering}}
		& $\boldsymbol{k = 1}$ & $\boldsymbol{k = 5}$ & $\boldsymbol{k = 10}$ \\
		\raisebox{1.5cm}{\rotatebox{90}{\textbf{Clean}}} 
		&\includegraphics[width = 0.3\textwidth]{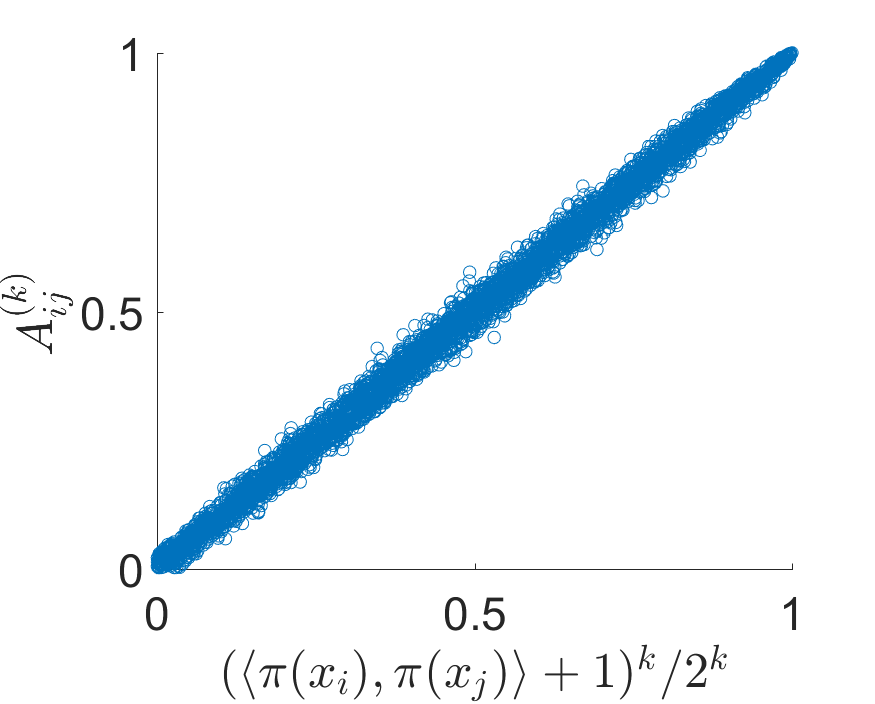} 
		&\includegraphics[width = 0.3\textwidth]{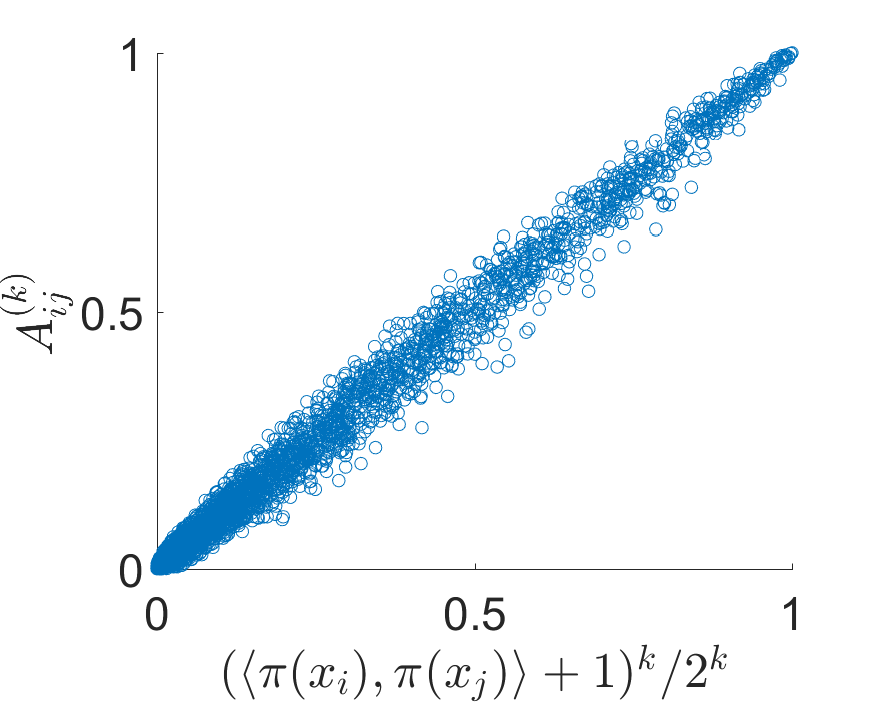} 
		&\includegraphics[width = 0.3\textwidth]{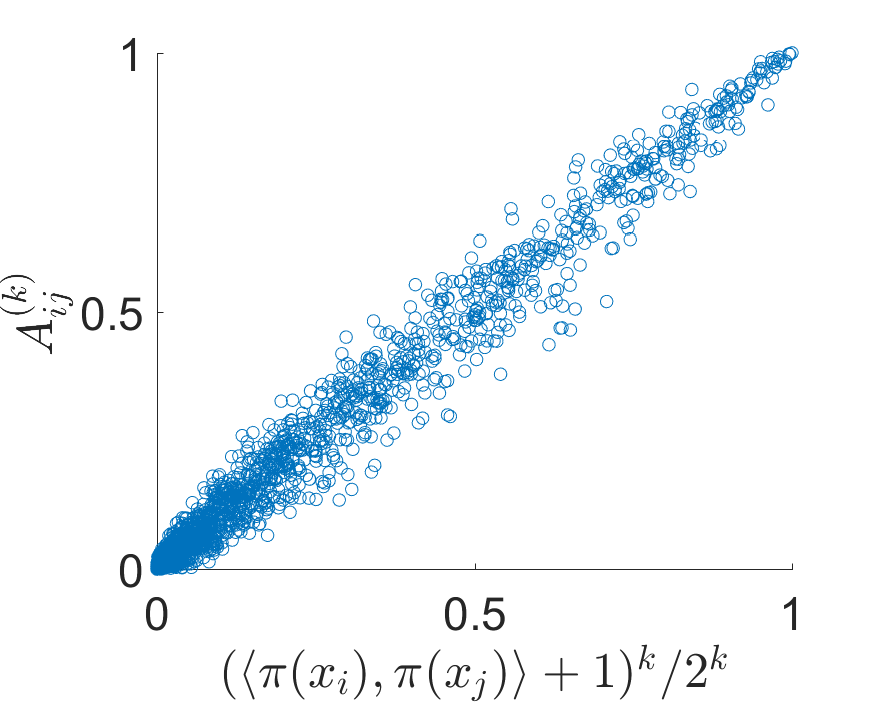}\\\hline
		\raisebox{1.0cm}{\rotatebox{90}{$\bm{\textbf{SNR} = 0.05}$}} 
		&\includegraphics[width = 0.3\textwidth]{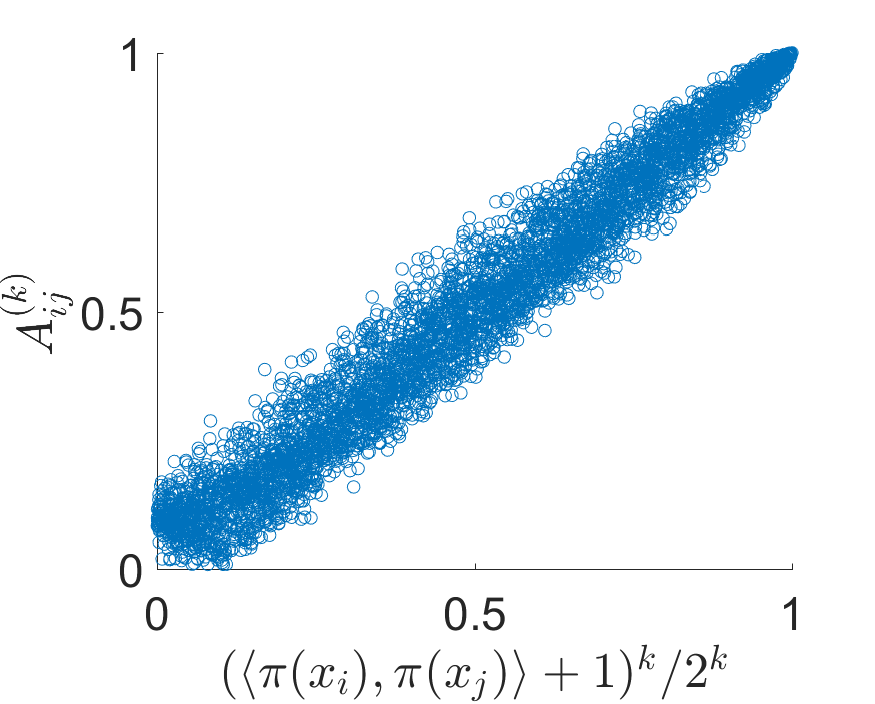} 
		&\includegraphics[width = 0.3\textwidth]{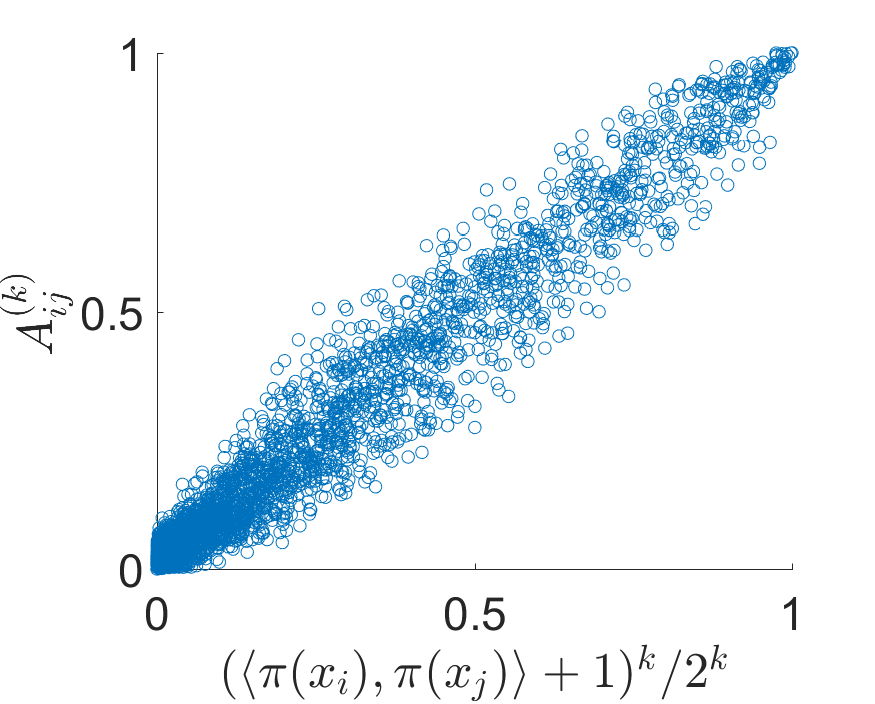} 
		&\includegraphics[width = 0.3\textwidth]{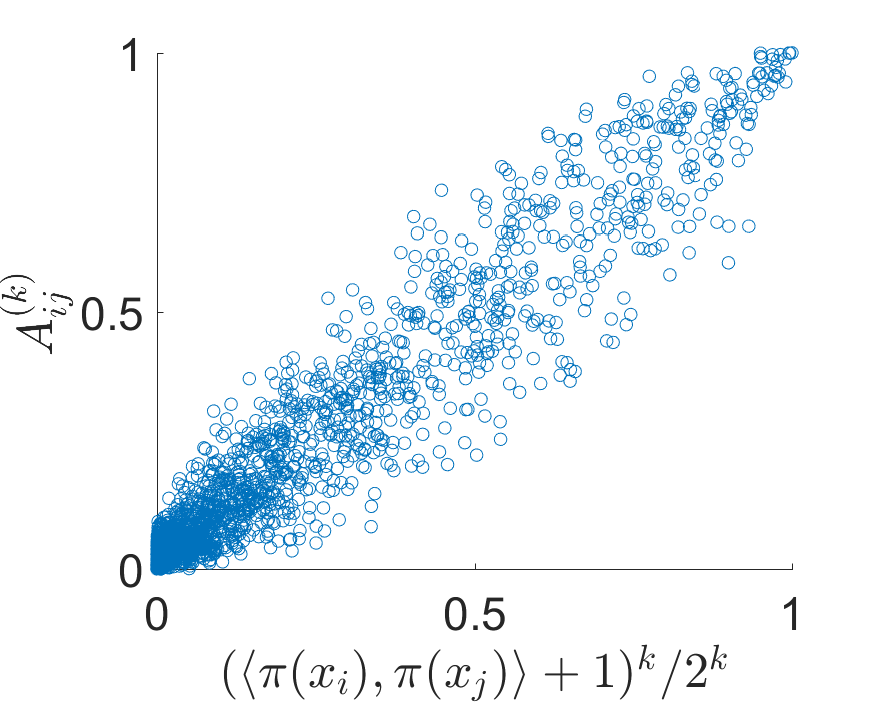}\\\hline
		\raisebox{1.0cm}{\rotatebox{90}{$\boldsymbol{\textbf{SNR} = 0.01}$}} 
		&\includegraphics[width = 0.3\textwidth]{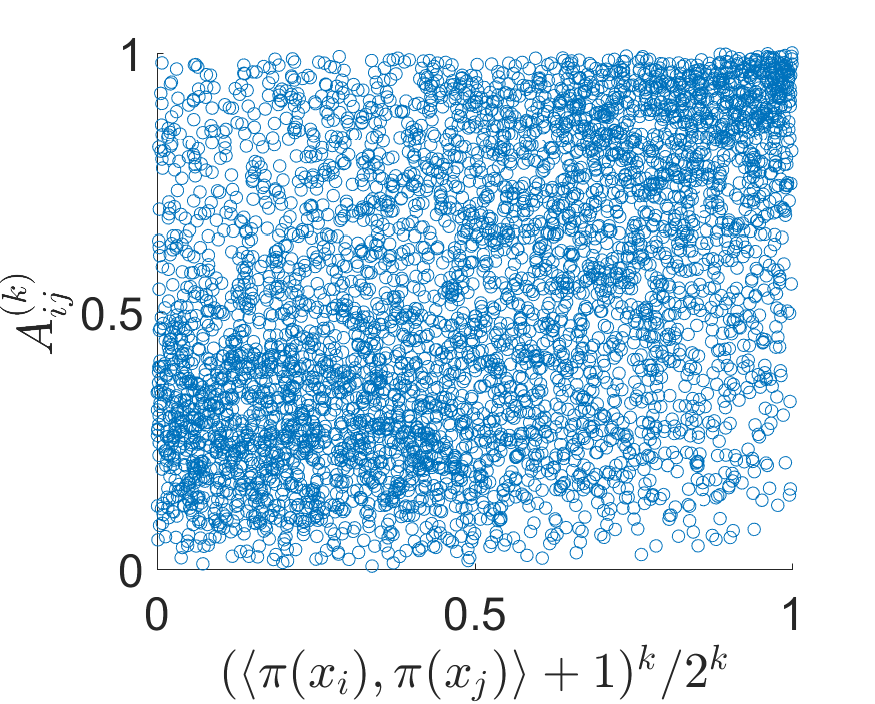} 
		&\includegraphics[width = 0.3\textwidth]{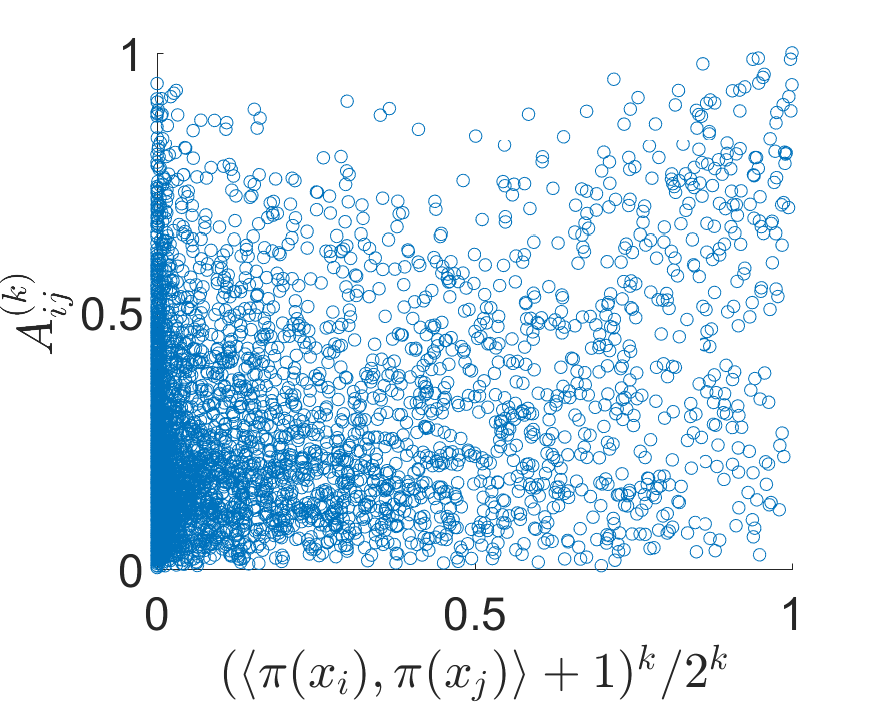} 
		&\includegraphics[width = 0.3\textwidth]{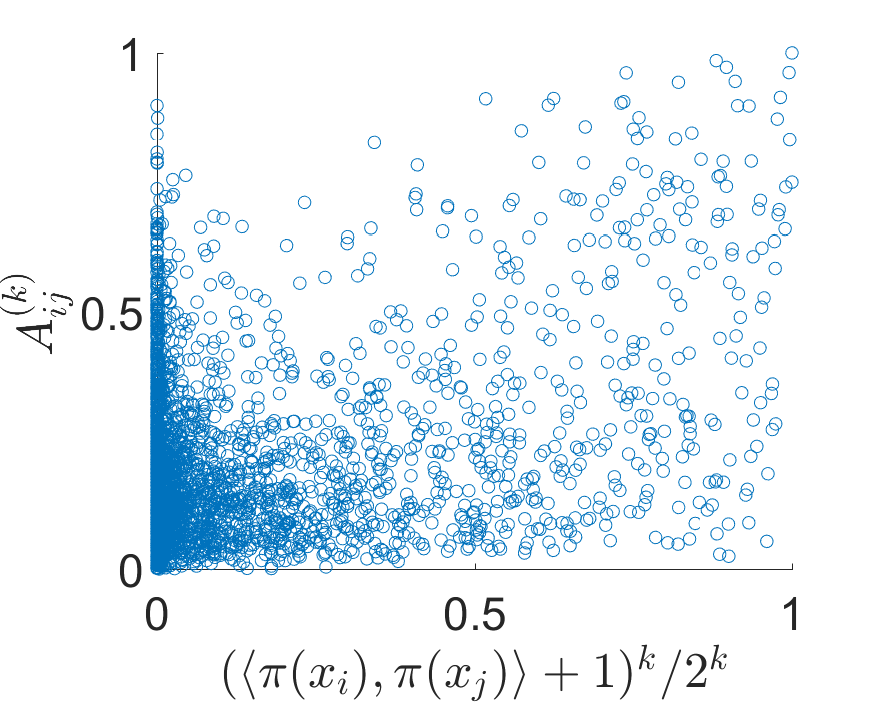}\\ \hline
	  \raisebox{0.9cm}{\rotatebox{90}{$\boldsymbol{\textbf{SNR} = 0.008}$}} 
		&\includegraphics[width = 0.3\textwidth]{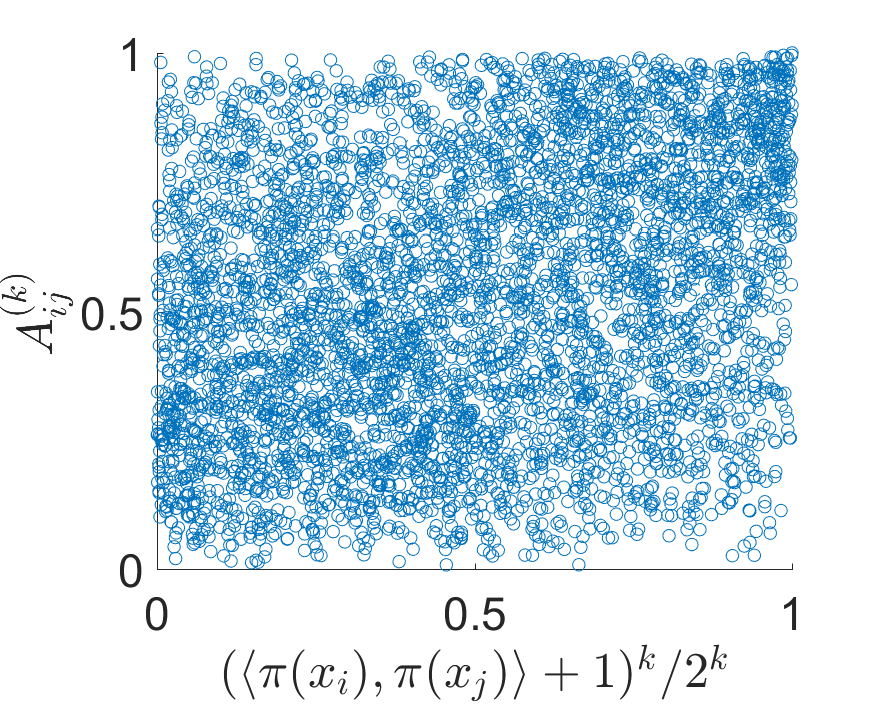} 
		&\includegraphics[width = 0.3\textwidth]{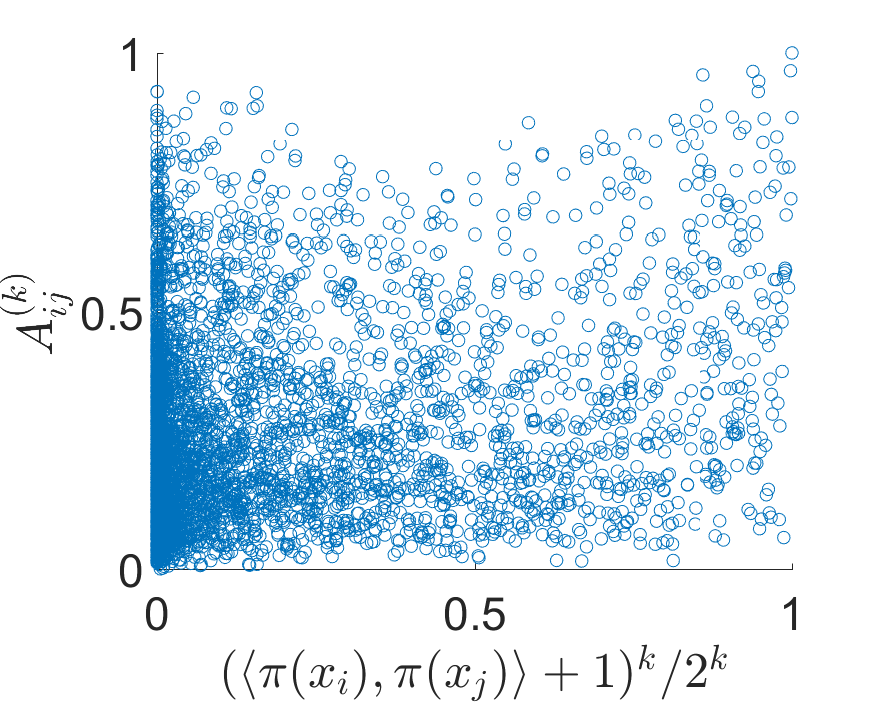} 
		&\includegraphics[width = 0.3\textwidth]{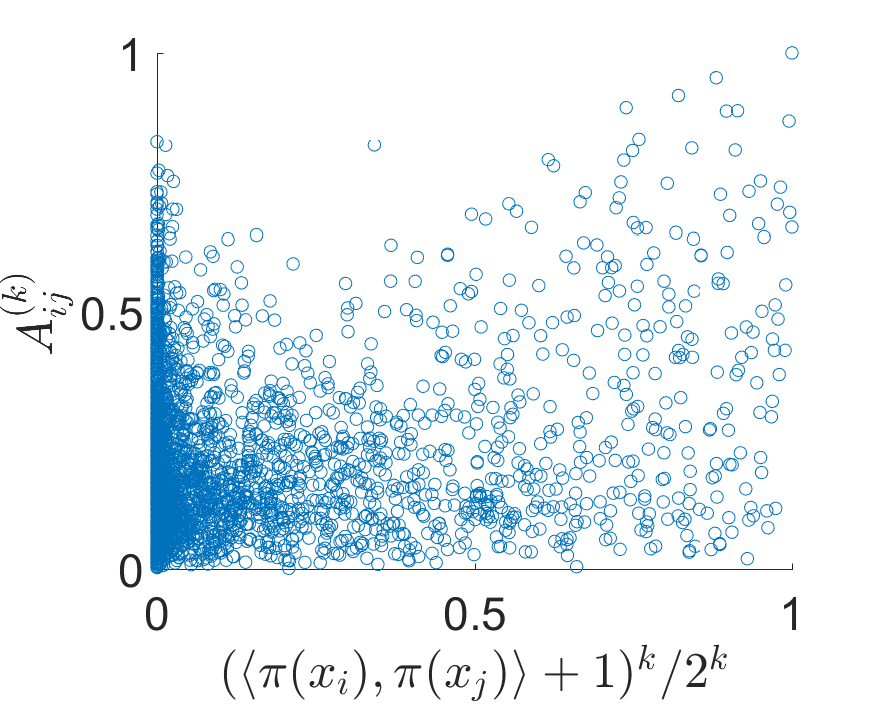}\\
	\end{tabular}
	\caption{\small Scatter plots of $A^{(k)}_{ij}$ against $\left(\langle \pi(x_i), \pi(x_j) \rangle + 1 \right)^k / 2^k$ for clean and noisy projection images with $\text{SNR} = 0.05, \, 0.01, \, 0.008$, at frequency $k = 1, 5, \, \text{and }10$. }
	\label{fig:cryo_scatter}
	\vspace{-0.4cm}
\end{figure}

\begin{figure}
	\centering
	\subfloat[Clean Projections]{\includegraphics[width= 0.23\textwidth]{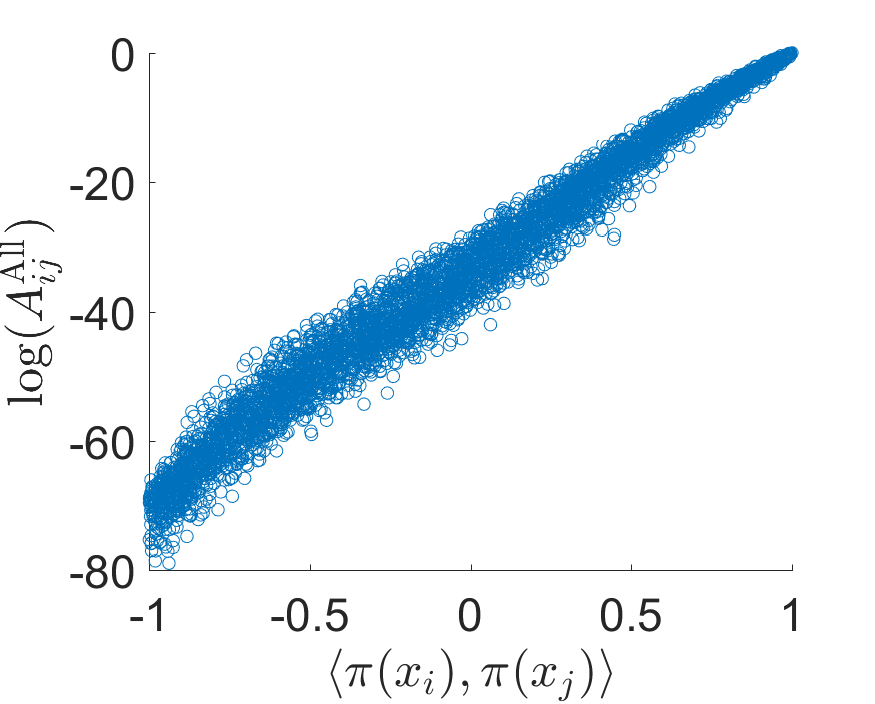}}
	\subfloat[$\text{SNR} = 0.05$]{
		\includegraphics[width = 0.23\textwidth]{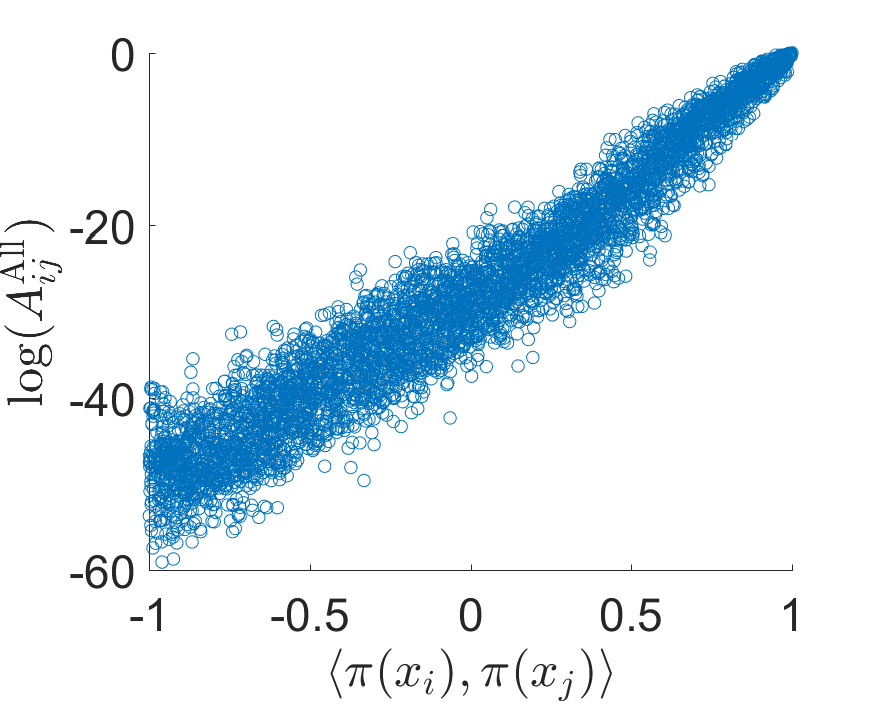}}
	\subfloat[$\text{SNR} = 0.01$]{
		\includegraphics[width = 0.23\textwidth]{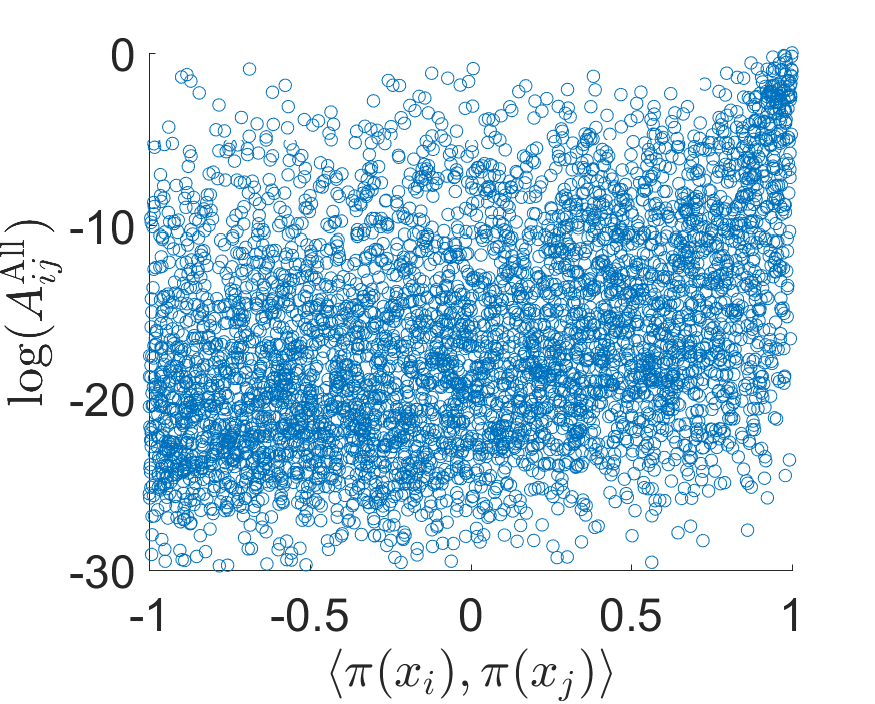}}
	\subfloat[$\text{SNR} = 0.008$]{
		\includegraphics[width = 0.23\textwidth]{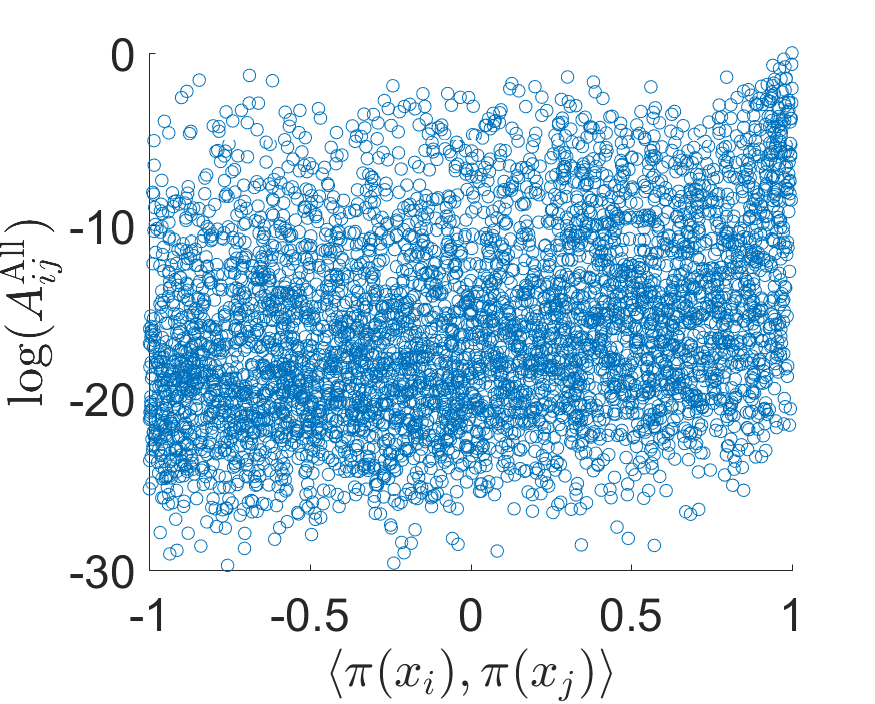}}
	\caption{\small Scatter plots for log multi-frequency class averaging affinity $\log A^{\text{All}}_{ij}$ against $\langle \pi(x_i), \pi(x_j) \rangle $ for clean and noisy projection images with $\text{SNR} = 0.05,\, 0.01,\text{ and } 0.008$.}
	\label{fig:cryo_scatter_log}
\end{figure}

In this section, we apply multi-frequency class averaging on simulated cryo-EM projection images. For each image, the goal is to identify projection images with similar viewing directions. We simulate $N = 10,000$ clean projection images of size $129 \times 129$ pixels from a 3-D electron density map of the 70S ribosome. The orientations for the projection images are uniformly distributed over $\SO(3)$. The clean images are contaminated by additive white Gaussian noise with different signal to noise ratios (SNRs). Sample images are presented in Figure~\ref{fig:cryo_sample}. Here, we do not consider the effects of contrast transfer functions (CTFs) on the images. In order to initially identify similar images and the corresponding rotational alignments, we first expand each image on steerable basis, and denoise the images by using \emph{steerable PCA} (sPCA)~\cite{zhao2016fast}. Then we generate the rotationally invariant features~\cite{zhao2014rotationally} from the filtered expansion coefficients to efficiently identify nearest neighbors without performing all pairwise alignments. The optimal alignment parameters are estimated between initial nearest neighbor pairs. The initial nearest neighbor list and alignment parameters are used to construct the initial graph. For clean images, the initial graph corresponds to the true neighborhood graph. For the extremely noisy images illustrated in Figure~\ref{fig:cryo_sample}, the initial similarity measure is corrupted by noise and images of totally different views can be misidentified as nearest neighbors.

In Figure~\ref{fig:cryo_spectrum}, we present the spectral patterns of the top eigenvalues of $\widetilde{H}^{(k)}$ built from our initial neighborhood identification and rotational alignment. At high SNR, such as $\text{SNR} \geq 0.05$, we can clearly observe the multiplicities $2k+1, 2k+3, 2k+5,\ldots$ and the spectral gaps. As the SNR decreases, such spectral patterns deteriorate. 
\begin{figure}
	\centering
	\subfloat[Clean Projections]{
		\includegraphics[width = 0.23\textwidth]{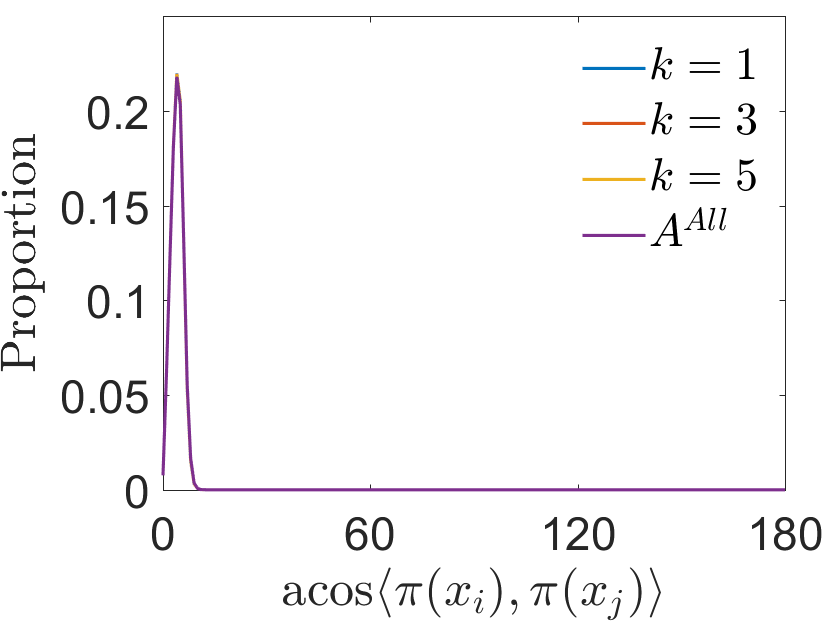}
		\label{fig:mfca_cryo_clean}
	}
	\subfloat[$\text{SNR} = 0.05$]{
		\includegraphics[width = 0.23\textwidth]{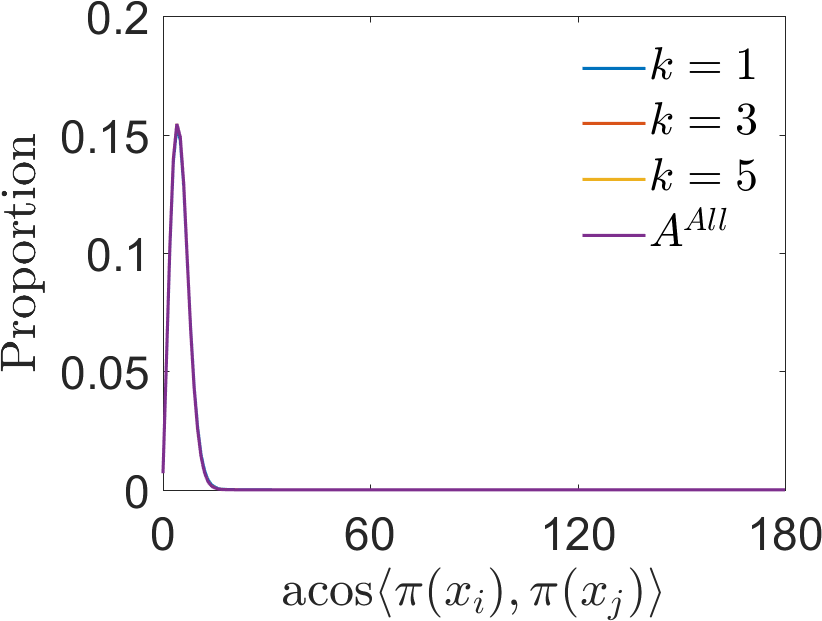}
		\label{fig:mfca_cryo_0_05}
	}
	\subfloat[$\text{SNR} = 0.01$]{
		\includegraphics[width = 0.23\textwidth]{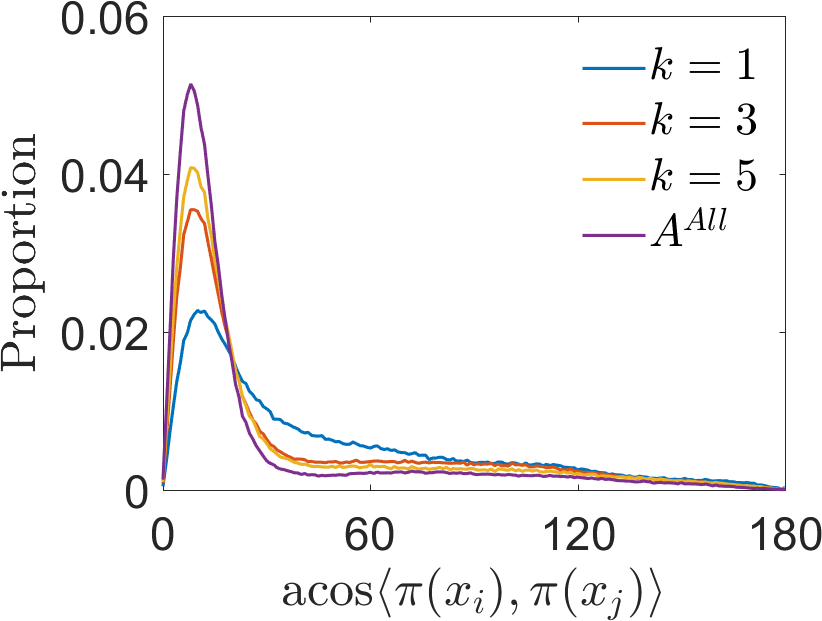}
		\label{fig:mfca_cryo_0_01}
		}
	\subfloat[$\text{SNR} = 0.008$]{
		\includegraphics[width = 0.23\textwidth]{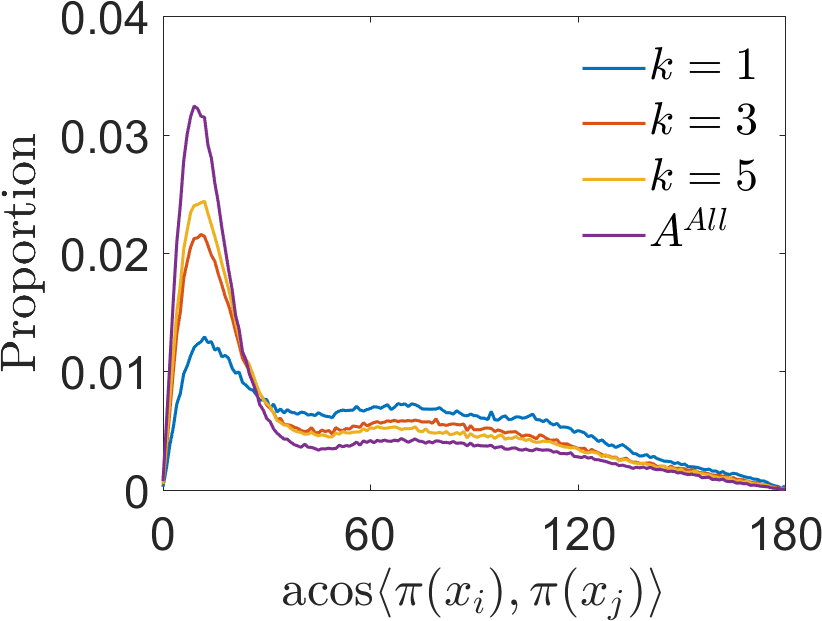}
		\label{fig:mfca_cryo_0_008}
	}
	\caption{\small Histogram of the angles ($x$-axis, in degrees) between the viewing directions of $10,000$ simulated cryo-EM projection images and its 50 neighboring projection images, with different SNRs, from \textit{left} to \textit{right}: clean projection images, $\text{SNR} = 0.05,\, 0.01,\, 0.008$. Here we set the maximum frequency $k_\text{max} = 20$.}
	\label{fig:mfca_cryo_hist}
		\vspace{-0.5cm}
\end{figure}
\begin{figure}
    \centering
    \subfloat[SNR$=0.05$]{
    \includegraphics[height = 0.23\textwidth]{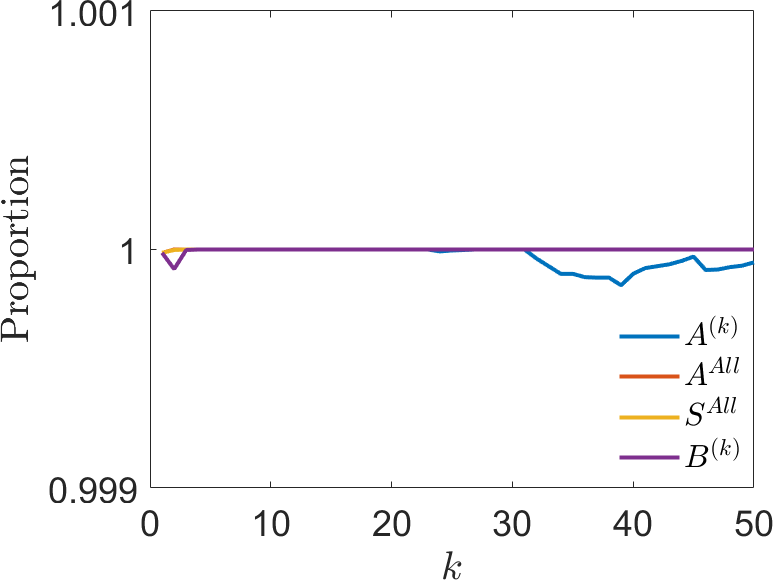}
    \label{fig:cryo_EM_vark_snr005}
    }
    \subfloat[SNR$=0.01$]{
    \includegraphics[height = 0.23\textwidth]{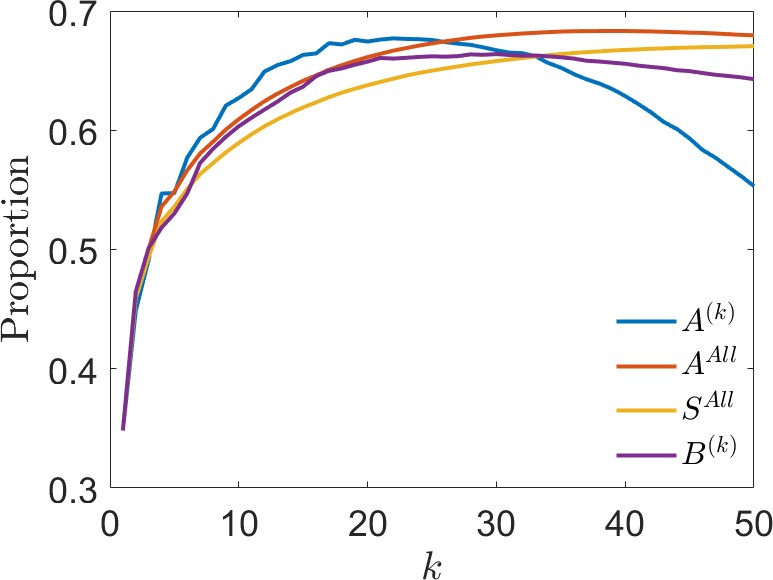}
    \label{fig:cryo_EM_vark_snr001}
    }
    \subfloat[SNR$=0.008$]{
    \includegraphics[height = 0.23\textwidth]{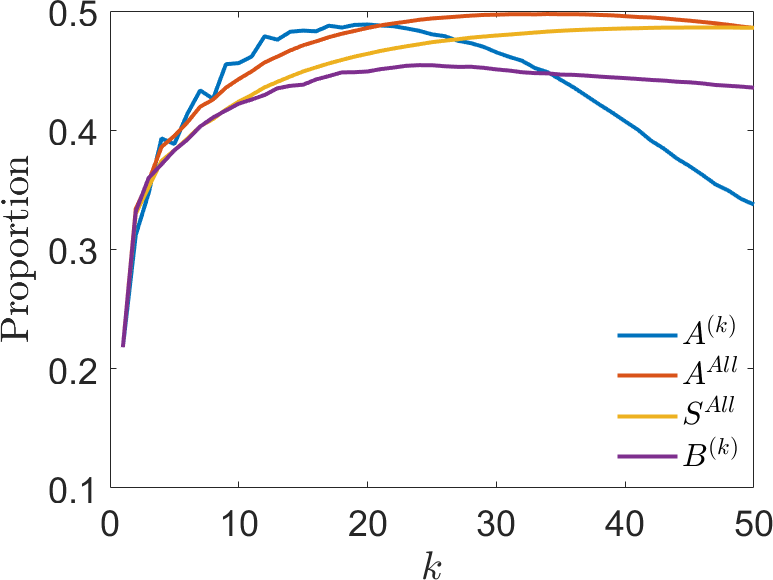}
    \label{fig:cryo_EM_vark_snr0008}
    }
    \caption{\small Comparing the performance of different affinities according to $A^{(k)}$, $A^{\text{All}}$, $S^{\text{All}}$, and $B^{(k)}$ for varying $k$ on noisy cryo-EM images. We evaluate the proportion of the estimated nearest neighbors that satisfy $\langle \pi(x_i), \pi(x_j)\rangle > 0.9$. }
    \label{fig:cryo_EM_vark}
    	\vspace{-0.5cm}
\end{figure}

In Figure~\ref{fig:cryo_scatter}, we present the scatter plots of $A_{ij}^{(k)}$ against $\left(\langle \pi(x_i), \pi(x_j) \rangle + 1 \right )^k/2^k $, with different SNRs. Similar to the synthetic dataset, the $A_{ij}^{(k)}$'s at frequency $k = 1$ fail at low SNRs, such as $\text{SNR} = 0.01\, 0.008$, while the $A_{ij}^{(k)}$'s at frequency $k = 5, 10$ are still able to distinguish the images with similar viewing directions (i.e., $\left(\langle \pi(x_i), \pi(x_j) \rangle + 1 \right )^k/2^k \approx 1$). This result indicates that better neighborhood image identification can be attained using higher frequency $k$. Moreover, Figure~\ref{fig:cryo_scatter_log} shows the scatter plots of the combined affinity against the dot products $\langle \pi(x_i), \pi(x_j) \rangle $ between the true viewing angles at varying SNRs. Even at $\text{SNR} = 0.01$, the combined affinity $A_{ij}^\text{All}$ is still able to distinguish projection images that have similar views $\pi(x)$, in contrast to the approximation results in Figure~\ref{fig:cryo_scatter}. 
\begin{figure}
	\captionsetup[subfigure]{oneside,margin={0.03cm,0cm}}
	\centering
	\subfloat[$\langle \pi(x_i), \pi(x_j) \rangle = -0.55$]{\includegraphics[width= 0.28\textwidth]{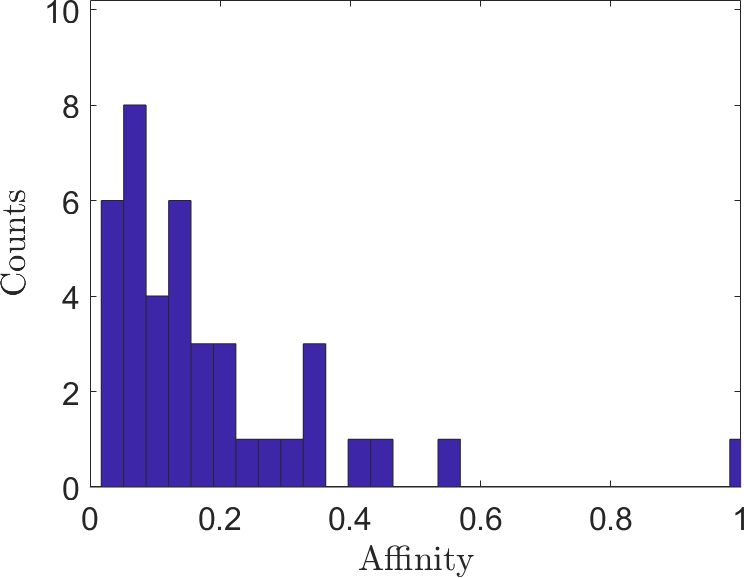}
	\label{fig:hist_score_cryoEM_a}}\,
	\subfloat[$\langle \pi(x_i), \pi(x_j) \rangle = 0.99$]{\includegraphics[width= 0.28\textwidth]{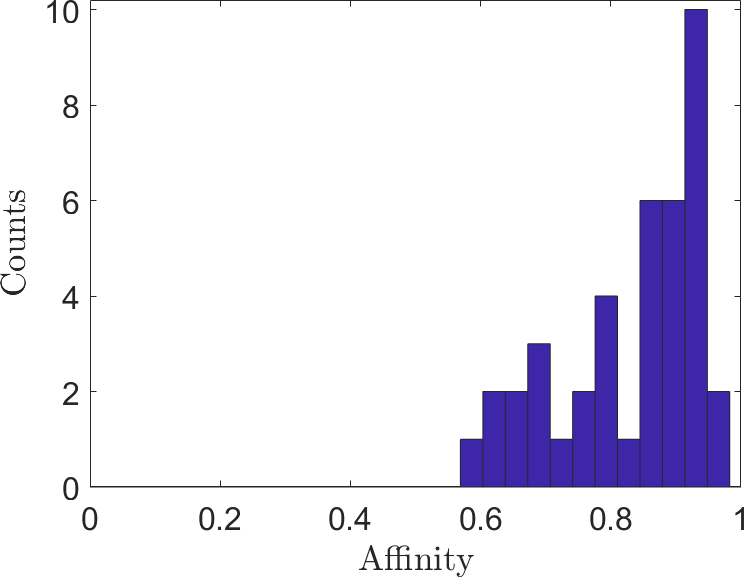}
	\label{fig:hist_score_cryoEM_b}} \,
	\subfloat[$\langle \pi(x_i), \pi(x_j) \rangle = -0.22$]{\includegraphics[width= 0.28\textwidth]{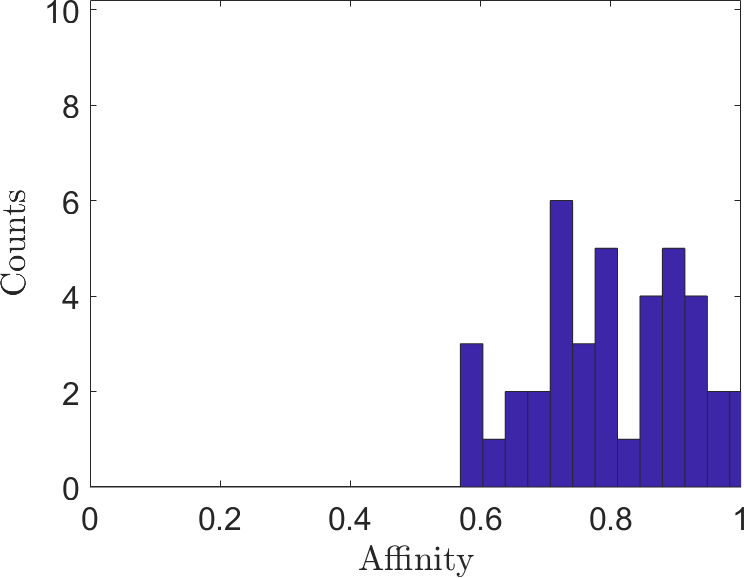}
	\label{fig:hist_score_cryoEM_c}}
	\caption{\small Histograms of the affinities $A^{(k)}_{ij}$ with $k = 1, \dots, 40$ for \protect \subref{fig:hist_score_cryoEM_a} a pair of wrongly identified nearest neighbors by $A^{(1)}$, \protect \subref{fig:hist_score_cryoEM_b} a good nearest neighbor pair identified by $A^\text{All}$, and \protect \subref{fig:hist_score_cryoEM_c} a wrongly identified nearest neighbor by $A^\text{All}$. The simulated cryo-EM images are of SNR$=0.01$.}
	\label{fig:hist_score_cryoEM}
	\vspace{-0.5cm}
\end{figure}
In Figure~\ref{fig:mfca_cryo_hist}, we evaluate the results by plotting the histogram of angels between viewing directions $\arccos (\pi(x_i), \pi(x_j))$ between all identified neighboring images $I_i$ and $I_j$. At high SNR, such as $\text{SNR} = 0.05$, using single frequency information as $k = 1,3,5$ can achieve similar results as combining all the frequencies. At low SNRs, such as $\text{SNR} = 0.01\,\,\text{and}\,\, 0.008$, $A^\text{All}$ which uses all frequencies information up to $k = 20$, outperforms the results obtained from using only a single frequency at $k = 1,\, 3,\, 5$.
 
In Figure~\ref{fig:cryo_EM_vark}, we compare the nearest neighbor classification results using affinities 
$A^{(k)}$, $B^{(k)}$, $A^{\text{All}}$, 
and $S^{\text{All}}$ at various frequency index $k$ for noisy images with SNR$= 0.05$, $0.01$, and $0.008$. The latter two affinities combine $A^{(k')}$ for $k' = 1, \dots, k$. Each image is identified with 50 nearest neighbors and we evaluate the proportion of the estimated nearest neighbors that satisfy $\langle \pi(x_i), \pi(x_j) \rangle > 0.9$. At SNR$=0.05$, all approaches achieve high accuracy (see Figure~\ref{fig:cryo_EM_vark_snr005}). At SNR$= 0.01$, $A^{(k)}$ is able to achieve better classification results than $B^{(k)}$ for $k$ between $4$ and $32$ and the proportion reaches $67.7\%$ for $A^{(k)}$ at $k = 22$. Using $A^{\text{All}}$ can improve the results further at $k = 40$, where the proportion reaches $68.3\%$. The improvement of $A^{(k)}$ and $A^\text{All}$ compared with $B^{(k)}$ gets more prominent at lower SNR (see Figure~\ref{fig:cryo_EM_vark_snr0008} with SNR$=0.008$). 

We note that the construction of the initial graph structure relies on the evaluation of the rotational invariant distance based on the steerable PCA expansion coefficients of the projection images~\cite{zhao2013fourier,zhao2016fast,zhao2014rotationally}. Thus the noise model is different from the probablistic models in Section~\ref{sec:noise_model} and the perturbation at each edge is induced by the noise on the corresponding two nodes. Despite the difference in the noise model, we still observe the benefit of using $A^{(k)}$ with $k>1$. However, the improvement of the combined affinity $A^{\text{All}}$ is not as impressive as the examples shown in Figure~\ref{fig:rrm_vark} and Figure~\ref{fig:result_ang_perturb}. Although we observe that certain miss-classified nearest neighbors by $A^{(1)}$ can be corrected by $A^\text{All}$ as shown in Figure~\ref{fig:hist_score_cryoEM_a}, there are still some wrong nearest neighbors that enjoy consistently high affinities across different $k$'s as shown in Figure~\ref{fig:hist_score_cryoEM_c}.

%% file: future.tex
\section{Conclusion and Future Work}
\label{sec:concl-future-work}

We propose in this paper a novel algorithm, referred to as
multi-frequency class averaging (MFCA), for classifying noisy
projection images in three-dimensional cryo-electron
microscopy by the similarity among viewing directions. The
new algorithm is a generalization of the eigenvector-based
approach of intrinsic classification first appeared in
\cite{singer2011viewing,hadani2011representation2}. We also
extended the representation theoretical framework of
\cite{hadani2011representation,hadani2011representation2} by
means of explicit constructions involving the Wigner
$D$-matrices, which completely characterizes the spectral
information of a generalized localized parallel transport
operator acting on sections of certain complex line bundle
over the two-dimensional unit sphere in $\mathbb{R}^3$;
these theoretical results conceptually establish the
admissibility and (improved) stability of the new MFCA
algorithm.

One intriguing future direction is to investigate into refined and more systematic aggregations of the results obtained from each individual frequency channel. Potential candidates include (1) the harmonic-retrieval-type transformations as in
multi-frequency phase synchronization~\cite{gao2019multi}, (2) cross-frequency invariant features such as bispectrum \cite{Kakarala2012,BBMZS2018}, and (3) tensor-based optimizations for multi-dimensional arrays \cite{KB2009,SFGD+2014,ALGS2017}. The main idea is to further exploit the redundancy in the reconstructed information across different irreducible representations. A direct extension of the MFCA theoretical framework could be a refined geometric interpretation of the multi-frequency vector diffusion maps \cite{fan2019cryo} in terms of aggregating invariant embeddings of the same underlying base manifold from multiple associated vector bundles of a fixed common principal bundle. 

Another future direction of interest is to integrate the mult-frequency methodology into existing algorithmic approaches for tackling the \emph{heterogeneity} problem in cryo-EM imaging analysis and comparative biology \cite{bajaj2018smac,LS2016,GBM2019}. In the context of cryo-EM, this problem occurs when molecules in distinct conformations coexist in solution, and thus images collected in cryo-EM imaging from random orientations should typically be first clustered into subgroups (using e.g. the maximum likelihood classification approaches \cite{SDCS2010,Scheres2012}) before single-particle reconstruction techniques can be applied to each individual subgroup. Recent studies \cite{DSL+2014,FO2016,Frank2018} even provided evidence for a continuous distribution of conformation states to present in a solution, which is far beyond the capability of maximum likelihood classification methods. We expect significant performance boost and sharper theoretical results from extensions of the multi-frequency methodology in these problems.

%% file: app.tex
\section{Basics on Group and Representation Theory}
\label{sec:app_rep}
A group $\mathcal{G}$ is a set with a multiplication
operation: $\mathcal{G} \times \mathcal{G} \mapsto \mathcal{G}$ obeying the following axioms:
\begin{enumerate}
\item For any $x, y \in \mathcal{G}$, $xy \in \mathcal{G}$ (closure);
\item For any $x, y, z \in \mathcal{G}$, $(xy)z = x(yz)$ (associativity);
\item There is a unique element of $\mathcal{G}$ denoted $e$ and called the identity for which $ex = xe = x$ for any $x \in \mathcal{G}$;
\item For any $x \in \mathcal{G}$ there is a corresponding element $x^{-1} \in \mathcal{G}$ called the inverse of $x$, which satisfies
$xx^{-1}=x^{-1}x=e$ for any $x \in \mathcal{G}$.
\end{enumerate}
The group operations may not be commutative, i.e.,
$xy$ is not necessarily equal to $yx$. This is crucial for our present purposes since 3D rotations do not commute. 

We have a group $\mathcal{G}$ acting on a set $X$. This means that each $g \in \mathcal{G}$ has the corresponding transformations based on a left (group) action $L_g: X \rightarrow X$ and a right (group) action $R_g: X \rightarrow X$. A left (group) action of $\mathcal{G}$ on $X$ is a rule for combining elements $g \in \mathcal{G}$ and elements $x \in X$, denoted by $g \vartriangleright x$. We additionally require the following three axioms.
\begin{enumerate}
\item $g \vartriangleright x \in X$ for all $x \in X$ and $g \in \mathcal{G}$. 
\item  $e \vartriangleright x = x$ for all $x \in X$.
\item $g_2 \vartriangleright (g_1 \vartriangleright x) = (g_2g_1)\vartriangleright x$ for all $x \in X$ and $g_1, g_2 \in \mathcal{G}$.
\end{enumerate}

A right (group) action of $\mathcal{G}$ on $X$ is a rule for combining elements $g \in \mathcal{G}$ and elements $x \in X$, denoted by $x \vartriangleleft g$. We additionally require the following three axioms.
\begin{enumerate}
\item $x \vartriangleleft g \in X$ for all $x \in X$ and $g \in \mathcal{G}$. 
\item  $x \vartriangleleft e = x$ for all $x \in X$.
\item $(x\vartriangleleft g_1) \vartriangleleft g_2 = x \vartriangleleft (g_1g_2)$ for all $x \in X$ and $g_1, g_2 \in \mathcal{G}$.
\end{enumerate}
The action of $\mathcal{G}$ on $X$ extends to functions on $X$ as shown in~\eqref{eq:Gactf}. 

In the paper we focus on two groups, namely $\SO(2)$ and $\SO(3)$. Both are compact Lie groups and admit irreducible representations. The group $\SO(2)$ is commutative and thus its irreducible representations are one dimensional complex numbers, $\rho_k(w(\theta)) = e^{\iota k \theta}$, for $w \in\SO(2)$ with a rotational angle $\theta \in [0, 2\pi)$. The irreducible representations of $\SO(3)$ are given by the Wigner $D$-matrices, which will be described in the subsection below. 

\subsection{Wigner's $D$- and $d$-Matrices}
\label{sec:prel-wign-d}

In this section we recall the definition and relevant properties of
the Wigner's $D$- and $d$-matrices, which are used
extensively in the paper for explicit computations related to the irreducible representations of $\SO(3)$. Recall that elements of $\SO(3)$ are realized as rotation matrices parameterized
by \emph{Euler angles} $\left( \varphi, \vartheta,\psi
\right)\in \left[ 0,2\pi \right)\times \left[ 0,\pi
\right]\times \left[ 0,2\psi \right)$: each $x\in \SO \left(
  3 \right)$ can be explicitly written as
\begin{equation}
\label{eq:euler-angle}
  \begin{aligned}
    x = x \left( \varphi, \vartheta,\psi \right)=
  \begin{pmatrix}
    \cos \varphi\cos\psi-\sin\varphi\sin\psi\cos\vartheta &\quad
    -\cos\varphi\sin\psi-\sin\varphi\cos\psi\cos\vartheta &\quad
    \sin\varphi\sin\vartheta\\
    \sin \varphi\cos\psi+\cos\varphi\sin\psi\cos\vartheta &\quad
    -\sin\varphi\sin\psi+\cos\varphi\cos\psi\cos\vartheta &\quad
    -\cos\varphi\sin\vartheta\\
    \sin\psi\sin\vartheta & \cos\psi\cos\vartheta & \cos\vartheta
  \end{pmatrix}.
  \end{aligned}
\end{equation}
Note that this is equivalent to writing $x=R_1 \left(
  \varphi \right)R_2 \left( \vartheta \right)R_3 \left( \psi
\right)$, where
\begin{equation*}
  R_1 \left( \varphi \right)=
  \begin{pmatrix}
    1 & 0 & 0\\
    0 & \cos\varphi & -\sin\varphi\\
    0 & \sin\varphi & \cos\varphi
  \end{pmatrix},\quad 
  R_2 \left( \vartheta \right)=
  \begin{pmatrix}
    \cos\vartheta & 0 & \sin\vartheta\\
    0 & 1 & 0\\
    -\sin\vartheta & 0 & \cos\vartheta
  \end{pmatrix}, \quad 
  R_3 \left( \psi \right)=
  \begin{pmatrix}
    \cos\psi & -\sin\psi & 0\\
    \sin\psi & \cos\psi & 0\\
    0 & 0 & 1
  \end{pmatrix}.
\end{equation*}
The last column in the matrix representation \eqref{eq:euler-angle} is
exactly the \emph{view direction} corresponding to $x\in \SO \left(  3
\right)$. For the simplicity of statements, we denote the
viewing direction of $x\in\SO ( 3  )$ as
\begin{equation*}
  \pi \left( x \right)=\pi \left( x \left( \varphi,\vartheta,\psi \right) \right)=\left( \sin\varphi\sin\vartheta,-\cos\varphi\sin\vartheta,\cos\vartheta \right)^{\top}\in\mathbb{R}^3.
\end{equation*}

For each integer $\ell=0,1,2,\dots$, the \emph{Wigner's
$D$-matrix} $\SO
\left( 3 \right)\ni x\mapsto D^{\ell}\left( x
\right)\in\mathbb{C}^{\left( 2\ell+1 \right)\times \left( 2\ell+1 \right)}$ is
the unique (up to isomorphism)
irreducible matrix representation of $\SO ( 3  )$
of index $\ell$. For each $x\in\SO ( 3  )$,
$D^{\ell}\left( x \right)$ is a $\left( 2\ell+1
\right)$-by-$\left( 2\ell+1 \right)$ complex 
Hermitian matrix, of which the entries we denote by
$D^{\ell}_{mn} \left( x \right)$ ($-\ell\leq m,n\leq
\ell$). As group representations, we have for any
$\ell=0,1,\dots$ and any $x,x'\in\SO ( 3  )$
the multiplicative formula
\begin{equation}
  \label{eq:wigner-D-mult}
  D^{\ell}\left( x' \right)D^{\ell}\left( x
  \right)=D^{\ell}\left( x'\vartriangleright x \right).
\end{equation}
The $ 2\ell+1 $ entries in the central
column of $D^{\ell}$, i.e., $D^{\ell}_{m0}$ ($-\ell\leq
m\leq \ell$), gives rise to the $2\ell+1$ independent
spherical harmonics of degree $\ell$. More generally, the
$2\ell+1$ entries in the $s$th column ($-\ell\leq s\leq
\ell$) of $D^{\ell}$ give rise to the $2\ell+1$ independent
spin-weighted spherical harmonics of degree $\ell$ and
weight $s$
\cite{eastwood1982edth,gelfand2018representations}. Using
the Euler angles, Wigner's $D$-matrices can be written
explicitly as
\begin{equation}
\label{eq:wigner-D-d}
  D^{\ell}_{mn}\left( \varphi,\vartheta,\psi \right)
  :=D^{\ell}_{mn}\left( x \left( \varphi,\vartheta,\psi
    \right) \right)=e^{-\iota m\varphi}d^{\ell}_{mn}\left(
    \vartheta \right)e^{-\iota n \psi},\quad m,n=-\ell,\dots,\ell
\end{equation}
where matrices $d^{\ell} \left( \varphi \right)$ are
known as \emph{Wigner's $d$-matrices}. They are real
$\left( 2\ell+1 \right)$-by-$\left( 2\ell+1 \right)$
matrices with an explicit formula for its $\left( m,n
\right)$th entry as
\begin{equation*}
  \begin{aligned}
    d_{mn}^{\ell}\left( \vartheta \right) \!= \!\left( -1
    \right)^{\ell-n} \left[ \left( \ell+m \right)! \left(
        \ell-m \right)! \left( \ell+n \right)! \left( \ell-n
      \right)! \right]^{1/2}\sum_s \left( -1 \right)^s
    \frac{\displaystyle\left( \cos \frac{\vartheta}{2}
      \right)^{m+n+2s}\left( \sin \frac{\vartheta}{2}
      \right)^{2\ell-m-n-2s}}{\displaystyle s! \left(
        \ell-m-s \right)! \left( \ell-n-s \right)! \left( m+n+s \right)!}
  \end{aligned}
\end{equation*}
with the sum running over all $s\in\mathbb{Z}$ that make sense of the
factorials
\cite[\S3.3.2]{marinucci2011random}. We will only
need the explicit form of $d^{\ell}_{mn}$ for the special
case $m=n=-\ell$: In this case it is straightforward to
verify that the summation consists of only one term
$s=2\ell$, and hence
\begin{equation}
  \label{eq:wigner-d-top-left}
  d^{\ell}_{-\ell,-\ell}\left( \vartheta \right)=\left( \cos
  \frac{\vartheta}{2}\right)^{2\ell}=\left( \cos^2 \frac{\vartheta}{2} \right)^{\ell}=\left( \frac{1+\cos\vartheta}{2} \right)^{\ell}.
\end{equation}
Alternatively,
$d_{mn}^{\ell}$ can also be written explicitly in terms of
Jacobi polynomials as (see
e.g. \cite[\S13.1.1]{marinucci2011random})
\begin{equation}
  \label{eq:wigner-small-d-jacobi}
  d_{mn}^{\ell}\left( \vartheta \right)=2^{-m}\left[
    \frac{\left( \ell-m \right)!\left( \ell+m
      \right)!}{\left( \ell-n \right)!\left( \ell+n
      \right)!} \right]^{\frac{1}{2}}\left( 1-\cos\vartheta
  \right)^{\frac{m-n}{2}}\left( 1+\cos\vartheta
  \right)^{\frac{m+n}{2}}P_{\ell-m}^{\left( m-n,m+n \right)} \left( \cos\vartheta \right)
\end{equation}
where $\left\{ P_n^{\left( a,b \right)}:n=0,1,2,\dots
\right\}$ denote the sequence of Jacobi polynomials with
parameters $a,b$ \cite[\S13.1.1]{marinucci2011random}. This
gives rise to the explicit formula for the
diagonal entries of the Wigner $d$-matrices:
\begin{equation}
  \label{eq:wigner-small-d-diagonal}
  d_{mm}^{\ell}\left( \vartheta \right)=2^{-m}\left( 1+\cos\vartheta \right)^mP_{\ell-m}^{\left( 0,2m \right)}\left( \cos\vartheta \right).
\end{equation}
In particular, we see directly from \eqref{eq:wigner-small-d-jacobi} that
\begin{equation}
  \label{eq:wigner-small-d-at-zero}
  d_{mn}^{\ell}\left( 0
  \right)=\delta_{mn}P_{\ell-m}^{\left( 0,2m \right)}\left(
    1 \right)=\delta_{mn}\cdot{\ell-m\choose \ell-m}=\delta_{mn}
\end{equation}
where $\delta_{mn}$ is the Kronecker delta notation
\begin{equation*}
  \delta_{mn}=
  \begin{cases}
    1 & \textrm{if $m=n$}\\
    0 & \textrm{otherwise.}
  \end{cases}
\end{equation*}
If the Euler angles of $x'$ take the form
$\left( 0,0,\psi \right)$, then by \eqref{eq:wigner-D-mult}
we have
\begin{equation}
  \label{eq:highest-weight}
  D_{mn}^{\ell}\left( \left( 0,0,\psi \right)\vartriangleright x
  \right)=\sum_{s=-\ell}^{\ell}D_{ms}^{\ell}\left( 0,0,\psi \right)D_{sn}^{\ell}\left( x
  \right)\stackrel{\eqref{eq:wigner-D-d}}{=\!=}\sum_{s=-\ell}^{\ell}
  d_{ms}^{\ell}\left( 0 \right) e^{-\iota
    s\psi}D_{sn}^{\ell}\left( x
  \right)\stackrel{\eqref{eq:wigner-small-d-at-zero}}{=\!=}e^{-\iota
  m\psi}D_{mn}^{\ell}\left( x \right).
\end{equation}
We will need this relation in the proof of Theorem~\ref{thm:eval}.

Recall from \cite[pp.21--22]{varshalovich1988quantum} that
Euler angles admit physical interpretations for the rotation
matrix: If we denote the canonical
right-handed orthonormal basis in $\mathbb{R}^3$ by
$\left\{\mathbf{e}_1, \mathbf{e}_2, \mathbf{e}_3\right\}$,
and write $R_{\mathbf{e}_i}\left( \alpha \right)\in\SO
\left( 3 \right)$ for the rotation around axis
$\mathbf{e}_i$ ($i=1,2,3$) by angle $\alpha$, then rotation
by $x \left( \varphi, \vartheta,\psi \right)\in\SO \left( 3
\right)$ is equivalent to i) rotation by angle $\varphi$
around $\mathbf{e}_3$, ii) rotation by angle $\vartheta$
around the new axis $\mathbf{e}_2'=R_{\mathbf{e}_3}\left(
  \varphi \right)\mathbf{e}_2$, and iii) rotation by angle
$\psi$ around the new axis
$\mathbf{e}_3'=R_{\mathbf{e}_3}\left( \varphi
\right)R_{\mathbf{e}_2}\left( \vartheta
\right)\mathbf{e}_3$. From this geometric interpretation, it
is clear that the action of $\SO ( 2  )$ on $\SO ( 3 )$ considered throughout this paper only
affects the Euler angle $\psi$. In other words, under the
canonical identification of $\SO ( 2  )$ with $\SO( 3)$ elements of the form
\begin{equation}
\label{eq:SO2-form}
  g = g \left( \alpha \right)=\begin{pmatrix}
    \cos \alpha & -\sin \alpha & 0\\
    \sin \alpha  &   \cos \alpha & 0\\
    0 & 0 & 1
  \end{pmatrix},\qquad \alpha\in \left[ 0,2\pi \right)
\end{equation}
then $x \left( \varphi, \vartheta, \psi
\right)\vartriangleleft {g_r} \left( \alpha \right)=x \left(
  \varphi, \vartheta, \psi+\alpha \right)$. Together with
\eqref{eq:wigner-D-d}, this implies
\begin{equation}
  \label{eq:wigner-D-SO2}
  \begin{aligned}
    D^{\ell}_{mn} \left( x \left( \varphi, \vartheta, \psi
\right)\vartriangleleft g \left( \alpha
\right)\right)&=D^{\ell}_{mn} \left( x \left( \varphi,
  \vartheta, \psi+\alpha\right)\right)\\
   &=e^{-\iota n\alpha}D^{\ell}_{mn} \left( x \left( \varphi,
  \vartheta, \psi\right)\right)=\rho_n \left(  g^{-1} \right) D^{\ell}_{mn} \left( x \left( \varphi,
  \vartheta, \psi\right)\right)
  \end{aligned}
\end{equation}
where again $\rho_n$ stands for the complex unitary irreducible
representation of $\SO ( 2  )$ of character $n$. 

\section{Spectral Analysis of the Local Generalized
  Parallel Transport Operators}
\label{sec:spectral-analysis}

We prove Theorem~\ref{thm:eval} to Theorem~\ref{thm:inner-product} in this appendix.

\subsection{Proof of Theorem~\ref{thm:eval}}
\label{sec:proof-theor-refthm-asym}

We begin with the isotypic decomposition
\eqref{eq:isotypic-decomp-1},
\eqref{eq:isotypic-decomp-2}. Following \eqref{eq:schur},
our strategy is to find a ``good point'' $x_0\in\SO( 3 )$ and a ``good function'' $u\in\mathcal{H}_{n,-k}$
($n\geq \left| k \right|$) such that $u \left( x_0
\right)\neq 0$, and evaluate
\begin{equation}
  \label{eq:rayleigh-quotient}
  \lambda_n^{\left( k \right)}\left( h \right)=\frac{\displaystyle
    \left(T_h^{\left( k \right)}u\right)\left( x_0
    \right)}{u \left( x_0 \right)}.
\end{equation}
To this end, pick the following basis for the Lie algebra
$\mathfrak{so}\left( 3 \right)$:
\begin{equation*}
  A_1=
  \begin{pmatrix}
    0 & 0 & 0\\
    0 & 0 & -1\\
    0 & 1 & 0
  \end{pmatrix},\quad
  A_2=
  \begin{pmatrix}
    0 & 0 & 1\\
    0 & 0 & 0\\
    -1 & 0 & 0
  \end{pmatrix},\quad
  A_3=
  \begin{pmatrix}
    0 & -1 & 0\\
    1 & 0 & 0\\
    0 & 0 & 0
  \end{pmatrix}.
\end{equation*}
It is straightforward to check that these elements satisfy
the commutator relations
\begin{equation*}
  \left[ A_3, A_1 \right]=A_2,\quad \left[ A_2, A_3
  \right]=A_1,\quad \left[ A_1, A_2 \right]=A_3.
\end{equation*}
We fix $x_0=\mathrm{I}_3$, the canonical standard
orthonormal frame in $\mathbb{R}^3$. We further equip $\SO( 3)$ with standard spherical coordinates ---
the Euler angles --- of the form
\begin{equation*}
  x = x \left( \varphi, \vartheta, \psi
  \right)=x_0\vartriangleleft e^{\varphi A_3}e^{\vartheta
    A_2}e^{\psi A_3}
\end{equation*}
where $\left( \varphi,\vartheta,\psi \right)\in \left(
  0,2\pi \right)\times \left( 0,\pi \right)\times \left(
  0,2\pi \right)$, as in
\cite[\S3.2.1]{hadani2011representation2}. The normalized
Haar measure on $\SO ( 3  )$ is given by the
density
\begin{equation*}
  \frac{\sin\theta}{8\pi^2}\,\mathrm{d}\varphi\,\mathrm{d}\vartheta\,\mathrm{d}\psi.
\end{equation*}

Consider the subgroup $T_{A_3}$ of $\SO ( 3  )$ generated by
the infinitesimal element $A_3$. For every $k\in\mathbb{Z}$
and $n\in\mathbb{N}$ with $n\geq \left| k \right|$, the
Hilbert space $\mathcal{H}_{n,k}$ admits yet another
isotypic decomposition with respect to the left action of
$T_{A_3}$:
\begin{equation}
  \label{eq:isotypic-decomp-3}
  \mathcal{H}_{n,-k}=\bigoplus_{m=-n}^n\mathcal{H}_{n,-k}^m
\end{equation}
where $s\in \mathcal{H}_{n,-k}^m$ if and only if
\begin{equation}
\label{eq:third-isotypic}
  s \left( e^{-tA_3}\vartriangleright x \right)=e^{\iota
    mt}s \left( x \right)\quad\textrm{for every $x\in\SO( 3)$ and $t\in\mathbb{R}$.}
\end{equation}
As pointed out in \cite[\S3.3.1]{hadani2011representation2},
elements of $\mathcal{H}_{n,-k}^m$ are often referred to as
(generalized) \emph{spherical functions}. In the physics
literature, they are also known as \emph{spin-weighted
  spherical functions}, which are closely related with
Wigner $D$-matrices
\cite{boyle2016should,eastwood1982edth,campbell1971tensor,goldberg1967spin,newman1966note}. We
extend the computation in
\cite[\S3]{hadani2011representation2} to $k>1$, by fully leveraging
properties of the Wigner $D$-matrices. In fact, we are going
to fix $m=-k$ and choose the ``good function'' $u$ as
$D^n_{-k,-k}$, the $\left( -k, -k \right)$th 
entry of the Wigner $D$-matrix of weight $n$, for any $n\geq
\left| k \right|$ --- it is clear from
\eqref{eq:wigner-D-SO2} that 
$D_{-k,-k}^n\in\mathcal{H}_{-k}$ for any $n\geq \left| k
\right|$, and from \eqref{eq:highest-weight} we know that
$D_{-k,-k}^n$ satisfies \eqref{eq:third-isotypic} with
$m=-k$. Our goal is to evaluate
\begin{equation}
  \label{eq:specialized-value}
  \lambda_n^{\left( k \right)} \left( h \right)=\frac{\displaystyle \left(
      T_h^{\left( k \right)}D_{-k,-k}^n \right)\left( x_0
    \right)}{D_{-k,-k}^n \left( x_0 \right)}.
\end{equation}
Now, on the one hand we have
\begin{equation}
  \label{eq:denomenator}
  D_{-k,-k}^n \left( x_0 \right)=D_{-k,-k}^n \left( 0,0,0
  \right)\stackrel{\eqref{eq:wigner-D-d}}{=\!=}d_{-k,-k}^n\left(
    0
  \right)\stackrel{\eqref{eq:wigner-small-d-at-zero}}{=\!=}1.
\end{equation}
On the other hand, note that by the invariance and
equivariance of the transport data \eqref{eq:Tk} we have for any $x=x
\left( \varphi,\vartheta,\psi \right)\in \SO ( 3  )$
\begin{align*}
  T^{\left( k \right)} \left( x_0,x \right) &= T^{\left( k
    \right)} \left( x_0,x_0\vartriangleleft e^{\varphi A_3}e^{\vartheta
    A_2}e^{\psi A_3} \right)=T^{\left( k
    \right)} \left( x_0, e^{\varphi A_3}\vartriangleright
  x_0\vartriangleleft e^{\vartheta A_2}e^{\psi A_3}
                                              \right)\\
    &=T^{\left( k \right)}\left( e^{-\varphi
      A_3}\vartriangleright x_0,  x_0\vartriangleleft
      e^{\vartheta A_2}e^{\psi A_3} \right)=T^{\left( k
      \right)}\left( x_0 \vartriangleleft e^{-\varphi
      A_3},  x_0\vartriangleleft
      e^{\vartheta A_2}e^{\psi A_3} \right)\\
      &=  e^{\iota k
      \varphi}T^{\left( k \right)} \left(
      x_0,x_0\vartriangleleft e^{\vartheta A_2}
        \right)e^{\iota k \psi}, 
\end{align*}
and
\begin{align*}
  D_{-k,-k}^n \left( x \right)=D_{-k,-k}^n \left(
  \varphi,\vartheta,\psi \right)=  e^{-\iota
  k \varphi}d_{-k,-k}^n \left( \vartheta \right)e^{-\iota k \psi}=e^{-\iota k
  \varphi}d_{-k,-k}^n \left( x_0\vartriangleleft
  e^{\vartheta A_2} \right)e^{-\iota k \psi}.
\end{align*}
Therefore,
\begin{align*}
  \left(T_h^{\left( k \right)}D_{-k,-k}^n\right)\left( x_0 \right)&=\int_{B \left( x, \alpha \right)}
  T^{\left( k \right)} \left( x_0,x \right)D_{-k,-k}^n \left( x
  \right)\,\mathrm{d}x\\
  &=\int_{B \left( x, \alpha \right)}
  T^{\left( k \right)} \left( x_0,x_0\vartriangleleft e^{\vartheta A_2}
  \right) D_{-k,-k}^n \left( x_0\vartriangleleft
  e^{\vartheta A_2} \right) \,\mathrm{d}x \left(
    \varphi,\vartheta,\psi \right)\\
  &=\int_{B \left( x, \alpha \right)}
  \rho_k\left(T\left( x_0,x_0\vartriangleleft e^{\vartheta A_2}
  \right)\right) D_{-k,-k}^n \left( x_0\vartriangleleft
  e^{\vartheta A_2} \right) \,\mathrm{d}x \left(
    \varphi,\vartheta,\psi \right)\\
   &\stackrel{(*)}{=\!=}\int_{B \left( x, \alpha \right)}
  D_{-k,-k}^n \left( \left(x_0\vartriangleleft
  e^{\vartheta A_2}\right)\vartriangleleft T\left( x_0\vartriangleleft e^{\vartheta
      A_2}, x_0 \right)\right) \,\mathrm{d}x \left(
    \varphi,\vartheta,\psi \right)\\
  &\stackrel{(**)}{=\!=}\int_{B \left( x, \alpha \right)}
  D_{-k,-k}^n \left( x_0\vartriangleleft
  e^{\vartheta A_2}\right) \,\mathrm{d}x \left(
    \varphi,\vartheta,\psi \right),
\end{align*}
where $(*)$ used the fact that
$D_{-k,-k}^n\in\mathcal{H}_{-k}$, and $(**)$ follows from
the definition \eqref{eq:parallel-transport-operator} and
the geometric fact that $x_0\vartriangleleft
  e^{\vartheta A_2}$ is exactly the parallel transport of
  $x_0$ along the unique geodesic connecting $\pi \left( x_0
  \right)$ to $\pi \left( x_0 \vartriangleleft
  e^{\vartheta A_2}\right)$:
\begin{align*}
  \left(x_0\vartriangleleft
  e^{\vartheta A_2}\right)\vartriangleleft T\left( x_0\vartriangleleft e^{\vartheta
      A_2}, x_0 \right)=t_{\pi \left( x_0\vartriangleleft
  e^{\vartheta A_2} \right),\pi \left( x_0 \right)}x_0=x_0\vartriangleleft
  e^{\vartheta A_2}.
\end{align*}
It follows that
\begin{align*}
  \left(T_h^{\left( k \right)}D_{-k,-k}^n\right)\left( x_0 \right)&=\int_{B \left( x, \alpha \right)}
  D_{-k,-k}^n \left( x_0\vartriangleleft
  e^{\vartheta A_2}\right) \,\mathrm{d}x \left(
    \varphi,\vartheta,\psi \right)\\
   &=\frac{1}{\left( 2\pi
     \right)^2}\int_0^{2\pi}\,\mathrm{d}\varphi\int_0^\alpha
     \frac{\sin\vartheta}{2}D_{-k,-k}^n \left( 0,\vartheta,0
     \right)
     \,\mathrm{d}\vartheta=\int_0^\alpha  \frac{\sin\vartheta}{2}d_{-k,-k}^n
     \left( \vartheta \right) \,\mathrm{d}\vartheta.
\end{align*}
Since $d_{-k,-k}^n=d_{k,k}^n$ (see e.g. \cite[formula
(3.16)]{marinucci2011random}), this further implies
\begin{align}
  \left(T_h^{\left( k \right)}D_{-k,-k}^n\right)\left( x_0
  \right)&=\int_0^\alpha\frac{\sin\vartheta}{2}d_{kk}^n\left(
  \vartheta \right)
  \,\mathrm{d}\vartheta\stackrel{\eqref{eq:wigner-small-d-diagonal}}{=\!=}2^{-\left(
  k+1 \right)}\int_0^\alpha\sin\vartheta \left( 1+\cos\vartheta
           \right)^kP_{n-k}^{\left( 0,2k \right)}\left(
           \cos\vartheta \right)\,\mathrm{d}\vartheta\nonumber\\
   &=-2^{-\left( k+1 \right)}\int_0^\alpha \left( 1+\cos\vartheta
           \right)^kP_{n-k}^{\left( 0,2k \right)}\left(
           \cos\vartheta \right) \,\mathrm{d}\cos\vartheta\nonumber\\
    &\stackrel{z:=\cos\vartheta}{=\!=\!=\!=\!=}2^{-\left(
      k+1 \right)}\int_{1-h}^1 \left( 1+z
     \right)^k P_{n-k}^{\left( 0,2k \right)}\left(
           z \right)\,\mathrm{d}z\label{eq:lambda-jacobi-poly-form}
\end{align}
where in the last equality we used $h=1-\cos \alpha $. Using the
explicit form of Jacobi polynomials (see
e.g. \cite[Chap.~IV, formula (4.2.1)]{szego1939orthogonal})
\begin{align*}
  P_{n-k}^{\left( 0,2k \right)}\left( z
  \right)&=\sum_{\nu=0}^{n-k}{n-k\choose n-k-\nu}{n+k\choose
    \nu} \left( \frac{z-1}{2} \right)^{\nu}\left(
           \frac{z+1}{2} \right)^{n-k-\nu}\\
    &=\sum_{\nu=0}^{n-k}{n-k\choose \nu}{n+k\choose
    \nu} \left( \frac{z-1}{2} \right)^{\nu}\left( \frac{z+1}{2} \right)^{n-k-\nu}
\end{align*}
we have
\begin{align*}
  \left(T_h^{\left( k \right)}D_{-k,-k}^n\right)\left( x_0
  \right)&=2^{-\left( k+1 \right)}\cdot
           2^k\sum_{\nu=0}^{n-k}{n-k\choose
           \nu}{n+k\choose \nu}\int_{1-h}^1 \left( \frac{z-1}{2}
           \right)^{\nu}\left( \frac{z+1}{2} \right)^{n-\nu}
           \,\mathrm{d}z\\
           &\stackrel{z:=1-2w}{=\!=\!=\!=\!=\!=}\frac{1}{2}\sum_{\nu=0}^{n-k}
             \left( -1
             \right)^{\nu}{n-k\choose\nu}{n+k\choose
             \nu}\int_0^{\frac{h}{2}} w^{\nu} \left( 1-w
             \right)^{n-\nu} \cdot 2\,\mathrm{d}w\\
           &=\sum_{\nu=0}^{n-k}
             \left( -1 \right)^{\nu}{n-k\choose\nu}{n+k\choose
             \nu}\int_0^{\frac{h}{2}} w^{\nu} \left( 1-w
             \right)^{n-\nu}\,\mathrm{d}w\\
           &=\sum_{\nu=0}^{n-k}
             \left( -1 \right)^{\nu}{n-k\choose\nu}{n+k\choose
             \nu}\mathrm{B}\left( \frac{h}{2};\nu+1,n-\nu+1 \right)
\end{align*}
where $\mathrm{B}\left( x;a,b \right)=\int_0^xw^{a-1}\left( 1-w \right)^{b-1}\,\mathrm{d}w$ is the incomplete
Beta function. It follows that for all $n\geq \left| k \right|$
\begin{align*}
  \lambda_n^{\left( k \right)}\left( h \right)&=\frac{\displaystyle \left(
      T_h^{\left( k \right)}D_{-k,-k}^n \right)\left( x_0
    \right)}{D_{-k,-k}^n \left( x_0 \right)}=\left(
      T_h^{\left( k \right)}D_{-k,-k}^n \right)\left( x_0
    \right)\\
     & =\sum_{\nu=0}^{n-k}
             \left( -1 \right)^{\nu}{n-k\choose\nu}{n+k\choose
             \nu}\mathrm{B}\left( \frac{h}{2};\nu+1,n-\nu+1 \right).
\end{align*}
From the integral form of the incomplete Beta function it is
clear that $\mathrm{B}\left( h/2;\nu+1,n-\nu+1 \right)$ is a
polynomial of degree $\left( n+1 \right)$ in $h$. In particular, by repeatedly applying the recursive relation
\begin{align*}
  \mathrm{B}\left( x;a+1,b \right)=\frac{a}{b}\mathrm{B}\left( x;a,b+1 \right) -\frac{1}{b}x^a \left( 1-x \right)^b,
\end{align*}
we easily obtain
\begin{align}
  \lambda_k^{\left( k \right)}\left( h
  \right)&=\mathrm{B}\left( \frac{h}{2};1,k+1
           \right)=\frac{1-\left( 1-h/2
           \right)^{k+1}}{k+1},\label{eq:lambda-k-k}
\end{align}
\begin{align}
   \lambda_{k+1}^{\left( k \right)}\left( h
  \right)&=\mathrm{B}\left( \frac{h}{2};1,k+2 \right)-\left(
           2k+1 \right)\mathrm{B}\left( \frac{h}{2};2,k+1
           \right)\notag\\
            &=\mathrm{B}\left( \frac{h}{2};1,k+2
              \right)-\frac{2k+1}{k+1}\mathrm{B}\left( \frac{h}{2};1,k+2 \right) +\frac{2k+1}{k+1} \left( \frac{h}{2} \right)\left( 1-\frac{h}{2} \right)^{k+1} \notag\\
            &= \frac{2(k+1)( 1- (1 - h/2)^{k+2})}{k+2} - \frac{(2k+1)(1 - (1-h/2)^{k+1})}{k + 1}, \label{eq:lambda-k-k+1}
\end{align}
\begin{align}
    \lambda_{k+2}^{\left( k \right)}\left( h
  \right)&=\mathrm{B}\left( \frac{h}{2};1,k+3 \right)-4\left(
           k+1 \right)\mathrm{B}\left(
           \frac{h}{2};2,k+2\right)+\left( k+1 \right)\left(
           2k+1 \right)\mathrm{B}\left( \frac{h}{2};3,k+1
           \right) \notag\\
           &=\mathrm{B}\left( \frac{h}{2};1,k+3 \right)-2 \mathrm{B}\left(
           \frac{h}{2};2,k+2\right)-\left( 2k+1
             \right)\left( \frac{h}{2} \right)^2 \left(
             1-\frac{h}{2} \right)^{k+1} \notag\\
           &=\frac{k}{k+2}\mathrm{B}\left( \frac{h}{2};1,k+3 \right)+\frac{2}{k+2}\left( \frac{h}{2} \right)\left( 1-\frac{h}{2} \right)^{k+2}-\left( 2k+1
             \right)\left( \frac{h}{2} \right)^2 \left(
             1-\frac{h}{2} \right)^{k+1} \notag\\
           &=\frac{k}{k+2}\cdot \frac{1-\left( 1-h/2
           \right)^{k+3}}{k+3}+\frac{2}{k+2}\left( \frac{h}{2} \right)\left( 1-\frac{h}{2} \right)^{k+2}-\left( 2k+1
             \right)\left( \frac{h}{2} \right)^2 \left(
             1-\frac{h}{2} \right)^{k+1},\label{eq:lambda-k-k+2}
\end{align}
which give rise to \eqref{eq:eval_k1} and
\eqref{eq:eval_k2}.

It now remains to compute a quadratic approximation for
$\lambda_n^{\left( k \right)}\left( h \right)$ for
$h\rightarrow0$, for all $n\geq \left| k \right|$. This can
be done by direct computation using the integral form of the
incomplete beta function: for all $n\geq \left| k \right|$,
\begin{align}
  &\lambda_n^{\left( k \right)}\left( 0 \right)=0,\\
  &\partial_h\lambda_n^{\left( k \right)}\left( 0 \right)=\sum_{\nu=0}^{n-k}
             \left( -1 \right)^{\nu}{n-k\choose\nu}{n+k\choose
             \nu}\cdot \frac{1}{2} \left( \frac{h}{2}
    \right)^{\nu}\left( 1-\frac{h}{2}
    \right)^{n-\nu}\Bigg|_{h=0}=\frac{1}{2},\label{eq:eval-derivative-1}\\
   &\partial_h^2\lambda_n^{\left( k \right)}\left( 0 \right)=\sum_{\nu=1}^{n-k}
             \left( -1 \right)^{\nu}{n-k\choose\nu}{n+k\choose
             \nu}\cdot \frac{\nu}{4} \left( \frac{h}{2}
    \right)^{\nu-1}\left( 1-\frac{h}{2}
    \right)^{n-\nu}\Bigg|_{h=0}\nonumber\\
   &\qquad\qquad\qquad\qquad+\sum_{\nu=0}^{n-k}
             \left( -1 \right)^{\nu}{n-k\choose\nu}{n+k\choose
             \nu}\cdot \frac{-\left(n-\nu\right)}{4} \left( \frac{h}{2}
    \right)^{\nu}\left( 1-\frac{h}{2}
    \right)^{n-\nu-1}\Bigg|_{h=0}\nonumber\\
  &\phantom{\partial_h^2\lambda_n^{\left( h \right)}\left( 0 \right)}=-\frac{1}{4}\left( n-k \right)\left( n+k \right)-\frac{n}{4}=-\frac{1}{4}\left( n^2+n-k^2 \right)\label{eq:eval-derivative-2}
\end{align}
and \eqref{eq:eval} follows from the Taylor expansion
\begin{align*}
  \lambda_n^{\left( k \right)}\left( h
  \right)=\lambda_n^{\left( k \right)}\left( 0
  \right)+h\,\partial_h\lambda_n^{\left( h \right)}\left( 0
  \right)+\frac{h^2}{2}\,\partial_h^2\lambda_n^{\left( h
  \right)}\left( 0 \right)+O \left( h^3 \right).
\end{align*}
This completes the entire proof of Theorem~\ref{thm:eval}. $\hfill\qed$

\subsection{Proof of Theorem~\ref{thm:specgap}}
\label{sec:proof-theor-refthm-gap}

\begin{lemma}
\label{lem:left-bound}
\begin{enumerate}[(1)]
\item There exists $h_1^{(k)} \in (0, 2]$ such that
  $\lambda^{(k)}_n(h) \leq \lambda_k^{(k)}(h)$ for every $n
  \geq k+1$ and $h \in (0, h_1^{(k)}]$.
\item There exists $h_2^{(k)} \in (0, 2 ]$ such that
  $\lambda_n^{(k)}(h) \leq \lambda_{k+1}^{(k)}(h)$ for every
  $n \geq k+2$ and $h \in (0, h_2^{(k)}]$.
\end{enumerate}
\end{lemma}

\begin{proof}[Proof of Lemma~\ref{lem:left-bound}]
  Since $\lambda_n^{\left( k \right)}\left( 0 \right)=0$ for
  all $k\in\mathbb{Z}$ and $n\geq \left| k \right|$, we will
  just compare the first order derivatives
  $\partial_h\lambda_n^{\left( k \right)}\left( h \right)$ over an interval
  with $0$ as the left end point. By
  \eqref{eq:lambda-jacobi-poly-form},
  $\partial_h\lambda_n^{\left( k \right)} \left( h \right)$
  admits a closed form expression in terms of Jacobi
  polynomials:
  \begin{equation*}
    \partial_h\lambda_n^{\left( k \right)} \left( h
    \right)=\frac{1}{2}\left( 1-\frac{h}{2}
    \right)^kP_{n-k}^{\left( 0,2k \right)}\left( 1-h
    \right)\stackrel{h=\cos \alpha}{=\!=\!=\!=\!=}\frac{1}{2}\left( \frac{1+\cos \alpha}{2} \right)^k
    P_{n-k}^{\left( 0,2k \right)}\left( \cos \alpha \right).
  \end{equation*}
In particular, under change-of-coordinates $h=1-\cos \alpha$ we have
\begin{align*}
  \partial_h\lambda_k^{\left( k \right)}\left( h
  \right)&=\frac{1}{2}\left( \frac{1+\cos \alpha}{2}
  \right)^kP_0^{\left( 0,2k \right)}\left( \cos \alpha
  \right)=\frac{1}{2}\left( \frac{1+\cos \alpha}{2} \right)^k,\\
  \partial_h\lambda_{k+1}^{\left( k \right)}\left( h
  \right)&=\frac{1}{2}\left( \frac{1+\cos \alpha}{2}
  \right)^kP_1^{\left( 0,2k \right)}\left( \cos \alpha
  \right)=\frac{1}{2}\left( \frac{1+\cos \alpha}{2} \right)^k
           \left[ 1-\left( k+1 \right)\left( 1-\cos \alpha \right) \right].
\end{align*}
It is clear that $0<\partial_h\lambda_{k+1}^{\left( k
  \right)}\left( h \right)<\partial_h\lambda_k^{\left( k
  \right)} \left( h \right)$ for all $h=1-\cos \alpha\in (0,1/\left( k+1
\right)]$, which together with $\lambda_k^{\left( k
  \right)}\left( 0 \right)=\lambda_{k+1}^{\left( k
  \right)}\left( 0 \right)$ gives rise to
\begin{align}
\label{eq:specgap-reduce}
  \lambda_{k+1}^{\left( k \right)}\left( h \right)\leq
  \lambda_k^{\left( k \right)}\left( h
  \right)\quad\textrm{for all $0<h\leq \frac{1}{k+1}$}.
\end{align}
With \eqref{eq:specgap-reduce}, the proof of both (1) and
(2) of Lemma~\ref{lem:left-bound} is reduced to only the
part (2) of Lemma~\ref{lem:left-bound}. The remaining of
this proof is devoted to establishing (2) of
Lemma~\ref{lem:left-bound}.

By the classical result of Szeg\H{o}
\cite[Theorem~8.21.13]{szego1939orthogonal}, there exists a
fixed positive number $c>0$ such that
\begin{align}
\label{eq:szego-asymp}
  P_{n-k}^{\left( 0,2k \right)}\left( \cos \theta
  \right)=\frac{1}{\sqrt{n}}k \left( \theta \right) \left[
  \cos \left( N\theta+\gamma \right) + \left( n\sin\theta
  \right)^{-1}O \left( 1 \right) \right],\quad\textrm{for
  all}\quad\frac{c}{n}\leq \theta\leq \pi-\frac{c}{n}
\end{align}
where
\begin{align*}
  k \left( \theta \right)&=\frac{1}{\sqrt{\pi\sin \left(
  \theta/2 \right)\cos \left( \theta/2 \right)}\cdot
  \left[\cos \left( \theta/2 \right)\right]^{2k}}=\left(
  \frac{2}{1+\cos \theta}
                           \right)^k\sqrt{\frac{2}{\pi\sin\theta}},\\
  N&=n+\frac{2k+1}{2},\quad \lambda=-\frac{\pi}{4}.
\end{align*}
In particular, by making the $O \left( 1 \right)$ term in
\eqref{eq:szego-asymp} explicit, we have for some absolute
constant $C>0$
\begin{align}
  \label{eq:jacobi-asymp-szego}
  \left( \frac{1+\cos\theta}{2} \right)^k \left|
  P_{n-k}^{0,2k}\left( \cos \theta \right) \right|\leq
  \sqrt{\frac{2}{n\pi}}\cdot \frac{1}{\sqrt{\sin \theta}}
  \left( 1+\frac{C}{n\sin\theta} \right)\quad\textrm{for all
  $\frac{c}{n}\leq \theta\leq \pi-\frac{c}{n}$.}
\end{align}
Note that the left hand side is precisely the absolute value
of $2\partial_h\lambda_n^{\left( k \right)}\left( h
\right)=2\partial_h\lambda_n^{\left( k \right)}\left(
  1-\cos\theta \right)$. We seek an upper bound for the
right hand side of \eqref{eq:jacobi-asymp-szego} that holds
uniformly for all sufficiently large $n$. To this end,
consider the largest zero of 
$P_{n-k}^{\left(0,2k\right)} \left( x \right)$ for $x\in
[-1,1]$, denoted as $x_{n-k}^{*}=\cos
\alpha_{n-k}^{*}$ (thus $\alpha_{n-k}^{*}$ is the smallest zero of the
function $\alpha\mapsto P_{n-k}^{\left( 0,2k \right)}\left(
  \cos \alpha \right)$ on $\alpha\in \left[ 0,\pi
\right]$). Well-known estimates for the extreme zero of
Jacobi polynomials (see e.g. \cite[\S2.2]{DJ2012}) indicates
\begin{align*}
  x_{n-k}^{*}>1-O \left( \frac{1}{n^2} \right)\,\,\textrm{as
  $n\rightarrow\infty$}\quad\Rightarrow\quad \alpha_{n-k}^{*}\rightarrow0 \,\,\textrm{as
  $n\rightarrow\infty$}
\end{align*}
thus for any $\epsilon_1>0$ there exists $N_1>0$ such that for
all sufficiently large $n\geq N_1$ we have
\begin{align}
\label{eq:close-to-origin}
\sin \alpha_{n-k}^{*}\geq \left( 1-\epsilon_1
  \right)\alpha_{n-k}^{*}
\end{align}
since $\lim_{x\rightarrow0}\left(\sin x\right) /
x=1$. In the meanwhile, \cite[Theorem~3.1]{ELR1994} bounds
$x_{n-k}^{*}=\cos \alpha_{n-k}^{*}$ from above by
\begin{align}
\label{eq:extreme-upper-bound}
  x_{n-k}^{*}<\frac{\displaystyle \left( 2k+\frac{1}{2}
  \right)^2+4 \left( n-k \right)\left( n+k+\frac{1}{2}
  \right)}{\displaystyle \left( 2n-2k+1+2k
  \right)^2}=\frac{\displaystyle 4n^2+2n+1/4}{\left(
  2n+1 \right)^2}=1-\frac{\displaystyle 2n+3/4}{\left( 2n+1 \right)^2}.
\end{align}
Using the elementary inequality $1-x^2/2\leq \cos x$ for
$x\in \left[ 0,2 \right]$, \eqref{eq:extreme-upper-bound}
leads to
\begin{align*}
  1-\frac{\left(\alpha_{n-k}^{*}\right)^2}{2}\leq \cos
  \alpha_{n-k}^{*}=x_{n-k}^{*}<1-\frac{\displaystyle
  2n+3/4}{\left( 2n+1
  \right)^2}\quad\Rightarrow\quad
  \alpha_{n-k}^{*}>\sqrt{\frac{\displaystyle
  4n+3/2}{\left( 2n+1 \right)^2}} \rightarrow
  \frac{1}{\sqrt{n}}\quad\textrm{as $n\rightarrow\infty$}
\end{align*}
which further implies (1) for sufficiently large $n$,
$\alpha_{n-k}^{*}\in [c/n.\pi-c/n]$, and (2) by choosing $n$
sufficiently large we can ensure for the same arbitrary
$\epsilon_1>0$ chosen for \eqref{eq:close-to-origin}, that,
in addition to \eqref{eq:close-to-origin}, there holds
\begin{align}
\label{eq:away-from-origin}
  \alpha_{n-k}^{*}>\frac{1-\epsilon_1}{\sqrt{n}}.
\end{align}
Now consider the smallest local extremum $\mu_{n-k}^{*}$ of the
function $P_{n-k}^{\left( 0,2k \right)}\left(
  \cos \alpha \right)$ for $\alpha\in [0,\pi]$ that is larger than $\alpha_{n-k}^{*}$,
i.e., 
\begin{align*}
  \mu_{n-k}^{*}:=\min \left\{\alpha\in \left[ 0,\pi \right]\mid
  \partial_{\alpha}P_{n-k}^{\left( 0,2k \right)}\left( \cos
  \alpha \right)=0\textrm{ and $\alpha\geq \alpha_{n-k}^{*}$}\right\}
\end{align*}
which by \eqref{eq:away-from-origin} is guaranteed to fall
within $[c/n,\pi-c/n]$. For any $n\geq N_0$, by
\eqref{eq:jacobi-asymp-szego}, \eqref{eq:close-to-origin},
and \eqref{eq:away-from-origin}, we have 
\begin{align*}
\label{eq:uniform-upper-jacobi-extreme}
  \left( \frac{1+\cos \mu_{n-k}^{*}}{2} \right)^k \left|
  P_{n-k}^{\left(0,2k\right)}\left( \cos \mu_{n-k}^{*}
  \right) \right| &\leq \sqrt{\frac{2}{\pi n}}\cdot
  \frac{1}{\sqrt{\sin \mu_{n-k}^{*}}}\left(
                                                    1+\frac{C}{n\sin\mu_{n-k}^{*}} \right)\\
   &\leq \sqrt{\frac{2}{\pi n}}\cdot
  \frac{1}{\sqrt{\sin \alpha_{n-k}^{*}}}\left(
                                                    1+\frac{C}{n\sin
     \alpha_{n-k}^{*}} \right)\\
   & \leq \sqrt{\frac{2}{\pi n}} \cdot \frac{1}{\sqrt{\left(
                                                    1-\epsilon_1
                                                    \right)\alpha_{n-k}^{*}}}\left(
     1+\frac{C}{n \left( 1-\epsilon_1 \right)\alpha_{n-k}^{*}} \right)\\
    &<\frac{1}{\left(1-\epsilon_1\right)n^{\frac{1}{4}}}
  \sqrt{\frac{2}{\pi}} \left( 1+\frac{C}{\left( 1-\epsilon_1 \right)^2\sqrt{n}} \right).
\end{align*}
The same inequality holds if we replace $\mu_{n-k}^{*}$ with
any other extremum of the
function $\alpha\mapsto P_{n-k}^{\left( 0,2k \right)}\left(
  \cos \alpha \right)$ in $\alpha\in \left[
  c/n,\pi-c/n \right]$. In particular, this implies that for
all sufficiently large $n$ we have (recalling that $h=1-\cos
\alpha$)
\begin{align*}
  \partial_h\lambda_n^{\left( k \right)}\left( h \right)=\partial_h\lambda_n^{\left( k \right)}\left(
  1-\cos \alpha \right)\leq \frac{1}{2}\left( \frac{1+\cos \alpha}{2} \right)^k \left| P_{n-k}^{\left(
  0,2k \right)}\left( \cos \alpha \right)
  \right|<\frac{1}{4}\quad\textrm{for all $a\in [c/n,\pi-c/n]$.}
\end{align*}
The rest of the proof follows easily from the proof of
\cite[Theorem~4]{hadani2011representation2}: Let $a_0\in
\left( 0,\pi \right)$ be such that
\begin{align*}
  \partial_h\lambda_{k+1}^{\left( k \right)}\left( h
  \right)=\partial_h\lambda_{k+1}^{\left( k \right)}\left(
  1-\cos \alpha
  \right)&=\frac{1}{2}\left( \frac{1+\cos \alpha}{2} \right)^k
           \left[ 1-\left( k+1 \right)\left( 1-\cos \alpha
           \right) \right]>\frac{1}{4}
\end{align*}
for all $\alpha < \alpha_0$ and sufficiently large $n$; the remaining
finitely cases can be verified directly as claimed in
\cite[\S A.2.1, pp.~612]{hadani2011representation2}.
Note that such a value $\alpha_0$ exists because
when $\alpha=0$ (i.e., $h=1$)
\begin{align*}
  \partial_h\lambda_{k+1}^{\left( k
  \right)}\left(1\right)=\frac{1}{2}>\frac{1}{4}.
\end{align*}
As argued in \cite[\S A.2, pp.~611]{hadani2011representation2}, we
set $z_0=\cos \alpha_0$ and $h_1^{\left( k \right)}=h_2^{\left( k
  \right)}=1+z_0$, which ensures
$\partial_h\lambda_n^{\left( k \right)}\left( z \right)\leq
\partial_h\lambda_{k+1}^{\left( k \right)}\left( z \right)$ for all
$z\in \left[ -1, z_0\right]$, and furthermore
$\lambda_n^{\left( k \right)}\left( h \right)\leq
\lambda_{k+1}^{\left( k \right)}\left( h \right)$, for all
$n\geq k+1$. This proves (2) of
Lemma~\ref{lem:left-bound} and completes the entire proof of
Lemma~\ref{lem:left-bound}.
\end{proof}

\begin{lemma}
\label{lem:right-bound}
\begin{enumerate}[(1)]
\item There exists $N_1^{(k)}>0$ such that
  $\lambda^{(k)}_n(h) \leq \lambda_k^{(k)}(h)$ for every $n
  \geq N_1^{(k)}$ and $h \in [h_1^{(k)},2]$.
\item There exists $N_2^{(k)}>0$ such that
  $\lambda_n^{(k)}(h) \leq \lambda_{k}^{(k)}(h)$ for every
  $n \geq N_2^{(k)}$ and $h \in [h_2^{(k)},1/\left( k+1 \right)]$.
\end{enumerate}
\end{lemma}

\begin{proof}[Proof of Lemma~\ref{lem:right-bound}]
  First note, on the one hand, that the Schatten $2$-norm of $T_h^{\left( k
    \right)}$ can be easily computed: By \cite[Theorem~VI.23]{reed1980methods},
  \begin{align*}
    \left\| T_h^{\left( k \right)}\right\|_2^2&=\int_{\SO(3)}\!\int_{\SO(3)}\left|T_h^{\left(k\right)}\left(x,y\right)\right|^2\,\mathrm{d}x\,\mathrm{d}y=\int_{\SO ( 3  )}\!\int_{B \left( y,\alpha \right)}\left|T_h^{\left(k\right)}\left( x,y \right)\right|^2\,\mathrm{d}x\,\mathrm{d}y\\
    &=\int_{\SO (3)}\!\int_{B \left( y,\alpha
      \right)}\,\mathrm{d}x\,\mathrm{d}y=\int_0^\alpha \frac{\sin\vartheta}{2}\,\mathrm{d}\vartheta=\frac{1-\cos\vartheta}{2}=\frac{h}{2}
  \end{align*}
where the last equality follows from $h=1-\cos \alpha$. On the
other hand,
\begin{align*}
  \left\| T_h^{\left( k \right)} \right\|_2^2=\sum_{n=k}^{\infty}\left( 2n+1 \right)\left(\lambda_n^{\left( k \right)}\right)^2
\end{align*}
which gives the same bound as \cite[formula
(A.4)]{hadani2011representation2}:
\begin{align*}
  \lambda_n^{\left( k \right)}\left( h \right)\leq \frac{\sqrt{h}}{\sqrt{4n+2}}.
\end{align*}
Since by \eqref{eq:lambda-k-k} we have
\begin{align*}
  \lambda_k^{\left( k \right)}\left( h
  \right)=\frac{1-\left( 1-h/2\right)^{k+1}}{k+1}\geq
  \frac{1-\left( 1-h_1^{\left( k
  \right)}/2\right)^{k+1}}{k+1}\quad\textrm{for all $h\in
  [ h_1^{\left( k \right)}, 2 ]$},
\end{align*}
it is straightforward to verify by direct computation that
there exists $N_1^{\left( k \right)}>0$ such that
$\sqrt{h}/\sqrt{4n+2}\leq \lambda_k^{\left( k \right)}\left(
  h \right)$ for every $n\geq N_1^{\left( k \right)}$ and
$h\in [ h_1^{\left( k \right)},2]$. This proves (1) of
Lemma~\ref{lem:right-bound}. Furthermore, by
\eqref{eq:lambda-k-k+1}
\begin{align*}
  \lambda_{k+1}^{\left( k \right)}\left( h \right)=-\frac{k}{k+1}\cdot \frac{1-\left( 1-h/2
           \right)^{k+2}}{k+2}+\frac{2k+1}{k+1} \left(
              \frac{h}{2} \right)\left( 1-\frac{h}{2}
              \right)^{k+1}
\end{align*}
a direct computation for the derivative of the left hand
side with respect to $h$ gives
\begin{align*}
  \partial_h\lambda_{k+1}^{\left( k \right)}\left( h \right)=\frac{1}{2}\left[ 1-\left( k+1 \right)h \right]\left( 1-\frac{h}{2} \right)^k
\end{align*}
from which it is easy to directly verify that
$h\mapsto\lambda_{k+1}^{\left( k \right)}\left( h \right)$
achieves its maximum at $h=1/\left( k+1 \right)$ over $h\in
[0,2]$, and $\lambda_{k+1}^{\left( k \right)}\left( h
\right)>0$ for all $h\in [0,1/\left( k+1 \right)]$. It is
then easy to verify by direct computation that there exists
$N_2^{\left( k \right)}$ such that $\sqrt{h}/\sqrt{4n+2}\leq
\lambda_{k+1}^{\left( k \right)}\left( h \right)$ for every
$n\geq N_2^{\left( k \right)}$ and $h\in [h_2^{\left( k
  \right)},1/\left( k+1 \right)]$. This proves (2) of
Lemma~\ref{lem:right-bound}. 
\end{proof}

\begin{proof}[Proof of Theorem~\ref{thm:specgap}]
  Direct computation using \eqref{eq:lambda-k-k} and
  \eqref{eq:lambda-k-k+1} establishes \eqref{eq:specgap}:
  \begin{align*}
       G^{(k)}(h) & = \lambda_k^{(k)}(h) -
                 \lambda_{k+1}^{(k)}(h)= \frac{1-\left( 1-h/2
           \right)^{k+1}}{k+1} + \frac{(2k+1)(1 - (1-h/2)^{k+1})}{k + 1} \\ & \quad - \frac{2(k+1)( 1- (1 - h/2)^{k+2})}{k+2} =  \frac{2 - (1-h/2)^{k+1}\left((k+1) h + 2\right)} {(k + 2)}.
  \end{align*}
Unsurprisingly, the spectral gap depends on the ``frequency
channel'' parameter $k\in\mathbb{N}$. The rest of the proof
follows verbatim the proof of
\cite[Theorem~4]{hadani2011representation2}: By
Lemma~\ref{lem:left-bound} and Lemma~\ref{lem:right-bound}
we have $\lambda_n^{\left( k \right)}\leq \lambda_k^{\left(
    k \right)}\left( k \right)$ for every $n\geq N_1^{\left( k \right)}$ and
$h\in [0,2]$, as well as $\lambda_n^{\left( k \right)}\leq \lambda_{k+1}^{\left(
    k \right)}\left( k \right)$ for every $n\geq N_2^{\left( k \right)}$ and
$h\in [0,1/\left( k+1 \right)]$. We then verify directly
both $\lambda_n^{\left( k \right)}\leq \lambda_k^{\left( k
  \right)}\left( h \right)$ over $h\in[0,2]$ and
$\lambda_n^{\left( k \right)}\leq \lambda_{k+1}^{\left( k
  \right)}\left( h \right)$ over $h\in[0,1/\left( k+2
\right)]$ for the finitely many cases left ($k\leq n\leq
N_1^{\left( k \right)}$ and $k+1\leq n\leq N_2^{\left( k
  \right)}$, respectively). 
\end{proof}

\subsection{Proof of Theorem~\ref{thm:morph}}
\label{sec:proof-theor-refthm-morph}

 Our proof extends the arguments in the proof of
  \cite[Theorem~5]{hadani2011representation2}. A key
  observation is that the top eigenvector $\mathcal{H}\left(
    \lambda_k^{\left( k \right)}\left( h
    \right) \right)$ coincides with the isotypic subspace
  $\mathcal{H}_{k,-k}$ (see
  Section~\ref{sec:gen-para-transp-op}). Consider the
  morphism $\omega:=\sqrt{1/\left( 2k+1
    \right)}\cdot\tau:\mathbb{C}^{2k+1}\rightarrow\mathcal{H}$
  defined as
  \begin{equation*}
    \omega \left( v \right)\left( x \right)=\left( \delta_x^{\left( k \right)} \right)^{*}\left( v \right).
  \end{equation*}
\noindent\underline{\emph{Part 1: $\tau$ is an isomorphism between
  $\mathbb{C}^{2k+1}$ and $\mathcal{H}_{k,-k}$.}} We first show that $\mathrm{Im}\left( \omega \right)\subset
\mathcal{H}_{k,-k}$, namely, for any $x\in \SO(3)$, $v\in\mathbb{C}^{2k+1}$, and $g\in\SO( 2)$ there holds 
\begin{equation}
  \label{eq:equivariance-SO2}
  \left( \delta_{x\vartriangleleft
    g}^{\left( k \right)} \right)^{*}\left( v
\right)=\rho_k \left( g^{-1} \right)\left( \delta_x^{\left( k \right)}
\right)^{*}\left( v \right).
\end{equation}
To this end, note that for any
$z\in\mathbb{C}$ we have
\begin{align*}
    \left\langle \left( \delta_{x\vartriangleleft
    g}^{\left( k \right)} \right)^{*}\left( v
\right), z \right\rangle_{\mathbb{C}}&=\left\langle v, \delta_{x\vartriangleleft
    g}^{\left( k \right)} \left( z \right)
\right\rangle_{\mathbb{C}^3}=\left\langle v,zD_{\cdot,-k}^k
  \left( x\vartriangleleft g \right)
\right\rangle_{\mathbb{C}^3}\\
  &\stackrel{\eqref{eq:wigner-D-SO2}}{=\!=}\left\langle
    v,z\rho_{-k}\left( g^{-1} \right)D_{\cdot,-k}^k\left( x
    \right) \right\rangle_{\mathbb{C}^3}=\left\langle
    v,z\rho_k\left( g \right)D_{\cdot,-k}^k\left( x \right)
    \right\rangle_{\mathbb{C}^3} \\
    & =\left\langle \rho_k\left(
    g^{-1} \right)v,zD_{\cdot,-k}^k\left( x \right)
    \right\rangle_{\mathbb{C}^3} =\left\langle \rho_k\left(
    g^{-1} \right)v, \delta_x^{\left( k \right)}\left( z \right)
    \right\rangle_{\mathbb{C}^3} \\
    & =\left\langle \rho_k \left( g^{-1} \right)\left( \delta_x^{\left( k \right)}
\right)^{*}\left( v \right),z \right\rangle_{\mathbb{C}^3}
\end{align*}
which proves \eqref{eq:equivariance-SO2}. Next, we show that $\omega$ is a morphism of $\SO(3)$ representations, namely, for any $x\in \SO( 3)$, $v\in\mathbb{C}^{2k+1}$, and $g\in\SO(3)$ there holds
\begin{equation}
  \label{eq:equivariance-SO3}
  \left( \delta_x^{\left( k \right)}
\right)^{*}\left( D^k \left( g \right)v
\right)=\left( \delta_{g^{-1}\vartriangleright x}^{\left( k \right)}
\right)^{*}\left( v \right).
\end{equation}
To this end, again for any arbitrary
$z\in\mathbb{C}$
\begin{align*}
  \left\langle \left( \delta_x^{\left( k \right)}
  \right)^{*}\left( D^k \left( g \right)v
\right), z \right\rangle_{\mathbb{C}}&=\left\langle D^k \left( g \right)v, \delta_x^{\left( k \right)} \left( z \right)
\right\rangle_{\mathbb{C}^3}=\left\langle D^k \left( g \right)v,zD_{\cdot,-k}^k
  \left( x \right)
\right\rangle_{\mathbb{C}^3} \\
 & =\left\langle v,zD^k \left( g^{-1} \right)D_{\cdot,-k}^k
  \left( x \right)\right\rangle_{\mathbb{C}^3} \stackrel{\eqref{eq:wigner-D-mult}}{=\!=}\left\langle v,zD_{\cdot,-k}^k
  \left( g^{-1}\vartriangleright  x
    \right)\right\rangle_{\mathbb{C}^3} \\
    & =\left\langle
    v,\delta^{\left( k \right)}_{g^{-1}\vartriangleright x}
    \left( z \right)
    \right\rangle_{\mathbb{C}^3}=\left\langle \left(
    \delta_{g^{-1}\vartriangleright x}^{\left( k \right)}
  \right)^{*}\left( v
\right), z \right\rangle_{\mathbb{C}}
\end{align*}
which proves \eqref{eq:equivariance-SO3}. It now follows
immediately that the morphism $\omega$ maps 
$\mathbb{C}^{2k+1}$ isomorphically, as a unitary
representation of $\SO ( 3  )$, onto
$\mathcal{H}_{k,-k}$, the unique isotypical component in
$\mathcal{H}_{-k}$ (by \eqref{eq:equivariance-SO2}) of
unitary irreducible $\SO ( 3  )$-representation of
dimension $2k+1$. This in turn implies that $\omega$ (and
thus $\tau$) is an isomorphism between Hermitian vector
spaces. It remains to determine the suitable normalization
constant; we show that $\mathrm{Tr}\left( \tau^{*}\circ \tau
\right)=2k+1$. Indeed,
\begin{align*}
  \mathrm{Tr}\left( \tau^{*}\circ \tau
  \right)& \!=\!\left( 2k+1 \right)\mathrm{Tr}\left(
  \omega^{*}\circ\omega
  \right)\!=\!\left( 2k+1 \right)\int_{\mathbb{C}S^{2k}}\left\langle
  \omega^{*}\circ \omega \left( v \right), v \right\rangle_{\mathcal{H}}\,\mathrm{d}v\!=\!\left( 2k+1 \right)\int_{\mathbb{C}S^{2k}}\left\langle
  \omega \left( v \right), \omega\left(v\right)
           \right\rangle_{\mathcal{H}}\,\mathrm{d}v\\
   &=\left( 2k+1 \right)\int_{\mathbb{C}S^{2k}}\int_{\SO( 3)}\left\langle \left( \delta_x^{\left( k
     \right)} \right)^{*} \left( v \right), \left( \delta_x^{\left( k
     \right)} \right)^{*} \left( v
     \right)\right\rangle_{\mathbb{C}}\,\mathrm{d}x\,\mathrm{d}v\\
   & =\left( 2k+1 \right)\int_{\mathbb{C}S^{2k}}\int_{\SO(
     3)}\left\langle \left( D_{\cdot,-k}^k \left( x
     \right) \right)^{*}v, \left( D_{\cdot,-k}^k \left( x
     \right)
     \right)^{*}v\right\rangle_{\mathbb{C}}\,\mathrm{d}x\,\mathrm{d}v\\
    &=\left( 2k+1 \right) \int_{\mathbb{C}S^{2k}}\int_{( 3 )}1\,\mathrm{d}x\,\mathrm{d}v=2k+1
\end{align*}
where $\mathbb{C}S^{2k}$ is the $\left(4k+1\right)$-dimensional unit sphere in
$\mathbb{C}^{2k+1}$, and $\mathrm{d}v$ is the unique
normalized Haar measure on $\mathbb{C}S^{2k}$.

\noindent\underline{\emph{Part 2: Proof of
    \eqref{eq:composition-identity}.}} By
\eqref{eq:equivariance-SO2} we have $\left(\mathrm{ev}_x\mid
  \mathbb{W}^{\left( k \right)}\right)\circ
\omega=\left(\delta_x^{\left( k \right)}\right)^{*}$, which
is equivalent to $\left(\varphi_x^{\left( k
    \right)}\right)^{*}\circ\tau=\left(\delta_x^{\left( k
    \right)}\right)^{*}$. The conclusion now follows from
the straightforward computation as in the proof of
\cite[Theorem~5]{hadani2011representation2}:
\begin{equation*}
  \left(\varphi_x^{\left( k
    \right)}\right)^{*}\circ\tau=\left(\delta_x^{\left( k
    \right)}\right)^{*}\Rightarrow \left(\varphi_x^{\left( k
    \right)}\right)^{*}\circ \left( \tau\circ \tau^{*} \right)=\left(\delta_x^{\left( k
    \right)}\right)^{*}\circ \tau^{*}\Rightarrow \left(\varphi_x^{\left( k
    \right)}\right)^{*}=\left(\delta_x^{\left( k
    \right)}\right)^{*}\circ \tau^{*}\Rightarrow \tau \circ \delta^{\left( k \right)}_x = \varphi^{\left( k \right)}_x.
\end{equation*}
This completes the entire proof. \qed

\subsection{Proof of Theorem~\ref{thm:inner-product}}
\label{sec:proof-theor-refthm-inner-product}

By Theorem~\ref{thm:morph}, $\tau$ is a morphism between
Hermitian inner product spaces $\mathbb{C}^{2k+1}$ and
$\mathbb{W}^{\left( k \right)}$ and
\eqref{eq:composition-identity} holds, thus by the same
argument in the last step of the proof of \cite[Theorem~6]{hadani2011representation2} it suffices to prove
that for any unit-norm complex numbers $v, u\in
\mathbb{C}$ there holds
\begin{equation}
  \label{eq:inner-product-isomophism}
  \left| \left\langle \delta_x^{\left( k \right)}\left( u
      \right), \delta_y^{\left( k \right)}\left( v \right)
    \right\rangle_{\mathbb{C}^{2k+1}} \right| = \left(
    \frac{\left\langle \pi(x), \pi(y)
      \right\rangle + 1}{2} \right)^k.
\end{equation}
This boils down to the following straightforward computation:
\begin{align}
  \left| \left\langle \delta_x^{\left( k \right)}\left( u
      \right), \delta_y^{\left( k \right)}\left( v \right)
    \right\rangle_{\mathbb{C}^{2k+1}} \right|&=\left|
  u\left(D_{\cdot,-k}^k \left( x \right)\right)^{\top}\left(
  D_{\cdot,-k}^k \left( y \right)
  \right)^{*}\bar{v}\right|=\left| D_{-k,-k}^k \left(
  x^{-1}y \right) \right|\notag\\
  &\stackrel{\eqref{eq:wigner-D-d}}{=\!=}\left|
    d^k_{-k,-k}\left( \vartheta\left(x^{-1}y\right) \right)
    \right|\stackrel{\eqref{eq:wigner-d-top-left}}{=\!=}\left(
    \frac{1+\cos\vartheta \left( x^{-1}y \right)}{2}\right)^k\label{eq:final-intermediate}
\end{align}
where $\vartheta \left( x^{-1}y \right)$ is the Euler angle
$\vartheta$ of $x^{-1}y\in\SO ( 3  )$. Recall
from \eqref{eq:euler-angle} that $\cos \vartheta \left(
  x^{-1}y \right)$ is exactly the $\left( 3,3 \right)$-entry
of the $3$-by-$3$ matrix form of $x^{-1}y\in\SO( 3)$, which is exactly identical to the inner product
of the third columns of the matrix forms of $x$ and $y$,
i.e.,
\begin{equation*}
  \cos\vartheta \left( x^{-1}y \right)=\left\langle \pi
    \left( x \right),\pi \left( y \right) \right\rangle.
\end{equation*}
Plugging this back into the rightmost term of \eqref{eq:final-intermediate} completes the entire proof. \qed